\documentclass[sigconf, nonacm]{acmart}

\usepackage{amsmath,amssymb,amsfonts}
\usepackage{algorithmic}
\usepackage{graphicx}
\usepackage{textcomp}
\usepackage{graphicx}
\usepackage{xcolor,colortbl}
\usepackage{booktabs} 
\usepackage{wrapfig}
\usepackage[english]{babel}
\usepackage{url}
\usepackage[ruled,vlined,linesnumbered]{algorithm2e}
\usepackage{times,epsfig}
\usepackage[tight,footnotesize]{subfigure}
\usepackage{enumitem}
\usepackage{hyperref}
\usepackage{lipsum}
\SetKw{Break}{break}

\begin{document}
\title{Efficient Stochastic Routing in Path-Centric Uncertain Road Networks---Extended Version}



\author{Chenjuan Guo}
\affiliation{%
  \institution{East China Normal University}
}
\email{cjguo@dase.ecnu.edu.cn}

\author{Ronghui Xu}
\affiliation{%
  \institution{East China Normal University}
}
\email{rhxu@stu.ecnu.edu.cn}

\author{Bin Yang} 
\authornote{Corresponding authors}
\affiliation{%
  \institution{East China Normal University}
}
\email{byang@dase.ecnu.edu.cn} 

\author{Ye Yuan}
\affiliation{%
  \institution{Aalborg University}
}
\email{yuanye@cs.aau.dk}

\author{Tung Kieu}
\author{Yan Zhao}
\affiliation{%
  \institution{Aalborg University}
}
\email{{tungkvt, yanz}@cs.aau.dk}

\author{Christian S. Jensen}
\affiliation{%
  \institution{Aalborg University}
}
\email{csj@cs.aau.dk}

\begin{abstract}
The availability of massive vehicle trajectory data enables the modeling of road-network constrained movement as travel-cost distributions rather than just single-valued costs, thereby capturing the inherent uncertainty of movement and enabling improved routing quality. Thus, stochastic routing has been studied extensively in the edge-centric model, where such costs are assigned to the edges in a graph representation of a road network. However, as this model still disregards important information in trajectories and fails to capture dependencies among cost distributions, a path-centric model, where costs are assigned to paths, has been proposed that captures dependencies better and provides an improved foundation for routing. Unfortunately, when applied in this model, existing routing algorithms are inefficient due to two shortcomings that we eliminate. First, when exploring candidate paths, existing algorithms only consider the costs of candidate paths from the source to intermediate vertices, while disregarding the costs of travel from the intermediate vertices to the destination, causing many non-competitive paths to be explored. We propose two heuristics for estimating the cost from an intermediate vertex to the destination, thus improving routing efficiency. Second, the edge-centric model relies on stochastic dominance-based pruning to improve efficiency. This pruning assumes that costs are independent and is therefore inapplicable in the path-centric model that takes dependencies into account. We introduce a notion of virtual path that effectively enables stochastic dominance-based pruning in the path-based model, thus further improving efficiency. Empirical studies using two real-world trajectory sets offer insight into the properties of the proposed solution, indicating that it enables efficient stochastic routing in the path-centric model. The source code has been made available at \url{https://github.com/decisionintelligence/Route-sota}.
\end{abstract}
\maketitle 

\section{Introduction}
Emerging innovations in transportation with disruptive potential, such as autonomous vehicles, transportation-as-a-service, and coordinated fleet transportation, call for high-resolution routing. For example, PostNord\footnote{https://www.postnord.dk/}, a logistics service provider, and FlexDanmark\footnote{https://flexdanmark.dk/}, a fleet transportation organization, often make deliveries with associated travel cost budgets, e.g., 2 hours. High-resolution routing enables them to maximize arrivals within given time budgets when scheduling trips, thereby improving the quality of their services~\cite{DBLP:journals/sigmod/GuoJ014}.

Meanwhile, the availability of increasingly massive volumes of vehicle trajectory data enables opportunities for capturing travel costs, such as travel time, at an unprecedented level of detail. The state-of-the-art models capture travel costs as distributions. 
This level of detail enables the above services that are not possible when using only average values. 
As another example, Table~\ref{tbl:intro_paths} shows the travel-time distributions of two paths $P_A$ and $P_B$ from an office to an airport, scheduled for an autonomous vehicle.
The average travel times, i.e., weighted sum, of $P_A$ and $P_B$ are 49 and 52, respectively. 
If a person needs to arrive at the airport in 60 minutes, taking $P_A$ is more risky since it incurs a 10\% chance of being late, although it has a smaller average. 
\vspace{-1em}
\begin{table}[!htp]
	\centering
	\begin{tabular}{|c|c|c|c|c|c|}
		\hline
		Travel time (mins) & 40 & 50 & 60 & 70 & AVG \\ \hline\hline
		$P_A$ & 0.5 & 0.2 & 0.2 & 0.1 & 49\\ \hline
		$P_B$ & 0 & 0.8 & 0.2 & 0 & 52\\ \hline
	\end{tabular}
        \vspace{0.5em}
	\caption{Travel Time Distributions for $P_A$ and $P_B$}
	\label{tbl:intro_paths}
        \vspace{-3em}
\end{table}

Given a source and a destination on an uncertain road network, a departure time, and a travel cost budget, e.g., 60 minutes, we study the classical arriving on time problem~\cite{nie2006arriving, DBLP:conf/ijcai/BentH07}, which conducts stochastic routing and aims to find a path that maximizes the probability of arriving at the destination within the cost budget. 

Travel cost distributions of an uncertain road networks are often built from available vehicle trajectory data. 
Cost distributions are used on two model settings: in the classical {\it edge-centric} (EDGE) model~\cite{nie2006arriving, niknami2016tractable, DBLP:conf/icde/YangGJKS14} and in the recent {\it PAth-CEntric} (PACE) model~\cite{pvldb17pathcost, DBLP:journals/vldb/YangDGJH18}.
In the paper, we study the problem of arriving on time in PACE uncertain road networks. This is proven to be a more accurate way to model travel costs in real traffic~\cite{pvldb17pathcost, DBLP:journals/vldb/YangDGJH18}.

Assume that 100 trajectories occurred on path $P=\langle e_1, e_2 \rangle$, where 80 trajectories spent 10 minutes on $e_1$ and $e_2$, respectively; and the other 20 trajectories spent 15 minutes each on both $e_1$ and $e_2$. 
These trajectories suggest strong dependency---a driver is either fast or slow on both edges. A driver is unlikely to be fast on one edge and slow on the other edge. 
However, the EDGE model ignores such dependency and simply splits the trajectories to fit edges $e_1$ and $e_2$, and both edges are assigned the distribution $\{[10, 0.8], [15, 0.2]\}$, meaning that traversing $e_1$ or $e_2$ takes 10 minutes with probability 0.8 and 15 minutes with probability 0.2. 

The PACE model solves the problem by maintaining the ``correct'' joint distributions for paths directly. 
Specifically, following the above example, in addition to the edge distributions, the PACE model also maintains the joint distribution that captures, e.g., 10 and 10 minutes on $e_1$ and $e_2$ with a probability of 0.8, derived from the original, non-split trajectories (see Table~\ref{tbl:PACE_weight}(a)), which preserves cost dependencies.

Based on this joint distribution, the derived cost distribution of path $P$ is consistent with the original trajectories (see Table~\ref{tbl:PACE_weight}(b)). 
However, when using the PACE model, it is unlikely that there will be sufficient data to maintain distributions for all meaningful paths in the underlying road network. 
In particular, the longer a path is, the fewer trajectories occurred on the path. 
Thus, the PACE model resorts to maintaining distributions for short paths with sufficient amounts of trajectories, e.g., above a threshold.
We call such paths \emph{trajectory-paths} or \emph{T-paths} for short. 
\vspace{-1.5em}
\begin{table}[!htbp]
	\centering
	\subfigure[Joint Distribution $W_J(P)$]{
		\begin{minipage}[b]{0.2\textwidth}
		\centering
			\begin{tabular}{|c|c|} \hline
				$\langle e_1, e_2 \rangle$  & $Probability$ \\ \hline \hline
				10, 10 & 0.80  \\ \hline
				15, 15 & 0.20  \\ \hline
			\end{tabular}
		\end{minipage}
	}
	\hspace{2mm}
	\subfigure[Cost Distribution $D(P)$]{
		\begin{minipage}[b]{0.2\textwidth}\centering
			\begin{tabular}{|c|c|} \hline
				$P$ & $Probability$ \\ \hline \hline
				$10+10=20$ & 0.80 \\ \hline								
                $15+15=30$ & 0.20 \\ \hline
			\end{tabular}
		\end{minipage}
	}
	\caption{The Path-Centric Model, PACE}
	\label{tbl:PACE_weight}
\end{table}
\vspace{-2.5em}

Computing the cost distribution of a path in PACE is complicated by the existence of multiple ways of computing the cost from the distributions of different T-paths and edges. 
For example, assume that distributions for edges $e_1$, $e_2$, and $e_3$ and for T-paths $\langle e_1, e_2\rangle$ and $\langle e_2, e_3\rangle$ exist. 
When computing the distribution of path $\langle e_1, e_2, e_3\rangle$, we can use the distributions for T-path $\langle e_1, e_2\rangle$ and edge $e_3$, edge $e_1$ and T-path $\langle e_2, e_3\rangle$, or T-path $\langle e_1, e_2\rangle$ and T-path $\langle e_2, e_3\rangle$ which overlap. 
An existing study~\cite{pvldb17pathcost} proves that the last option provides the best accuracy as T-path $\langle e_1, e_2\rangle$’s distribution preserves the dependency between $e_1$ and $e_2$ and T-path $\langle e_2, e_3\rangle$’s distribution preserves the dependency between $e_2$ and $e_3$. 

Although the PACE model offers better accuracy than the EDGE model, stochastic routing on the PACE model often takes longer than the EDGE model~\cite{DBLP:journals/vldb/YangDGJH18, DBLP:conf/mdm/Andonov018}, which is mainly due to two challenges. 

\noindent
\textbf{Search Heuristics: }Stochastic routing often needs to explore many candidate paths from the source to intermediate vertices. 
The efficiency depends highly on whether an algorithm is able to  explore more promising paths before less promising ones.
To enable this, the existing solution in the PACE model considers the distributions of the candidate paths themselves~\cite{DBLP:journals/vldb/YangDGJH18}. However, it ignores the cost from the intermediate vertices to the destination, thus often leading to the exploration of non-promising paths.  

For example, in Figure~\ref{fig:intro}, assume that we have three candidate paths $P_1$, $P_2$, and $P_3$ from source $v_s$. 
When only considering the candidate paths themselves, $P_3$ is the most promising,  i.e., having the least expected cost and the largest probability that the cost is below the budget.  
However, $P_3$ is a bad choice since $v_j$ is far away from destination $v_d$. 

If we also consider the cost from the intermediate vertex to the destination, $P_1$ and $P_2$ become more promising because although $P_1$ and $P_2$ themselves have higher cost than $P_3$, intermediate vertex $v_i$ is closer to the destination. 
Thus, heuristics that are able to accurately estimate the cost of travel from any intermediate vertex to the destination can be very effective, assuming that the heuristics are admissible to ensure correctness. 
We propose two types of admissible heuristics, i.e., binary and budget-specific, for the PACE model that are able to estimate the maximum probability of arriving at the destination within the budget from an intermediate vertex. %

\begin{figure}[!htp]
	\begin{center}
		\includegraphics[scale=0.5]{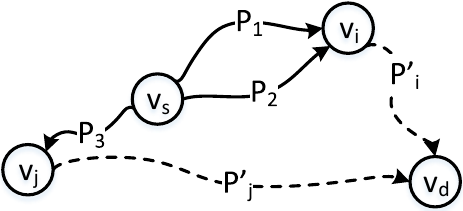}
	\end{center}
        \vspace{-1em}
	\caption{Motivating Example}
	\label{fig:intro}
\end{figure}
\vspace{-1em}

\noindent
\textbf{Effective Pruning: }Efficient stochastic routing in the EDGE model relies on stochastic dominance-based pruning~\cite{nie2006arriving, niknami2016tractable}. 
For example, consider Figure~\ref{fig:intro}, where candidate paths $P_1$ and $P_2$ both go from the source $v_s$ to an intermediate vertex $v_i$. 
If the cost distribution of $P_1$ stochastically dominates the cost distribution of $P_2$, we can safely prune $P_2$. 
This is equivalent to the deterministic case where $P_1$ has a smaller cost value than $P_2$.
More specifically, due to the independence assumption in the EDGE model, no matter which edge $e$ is extended from $v_i$, it is guaranteed that $D(P_1)\oplus D(e)$ stochastically dominates $D(P_2)\oplus D(e)$~\cite{nie2006arriving, niknami2016tractable}, where $\oplus$ represents the convolution operator and $D(\cdot)$ represents the cost distribution of a path or an edge. 
Thus, there is no need to further consider $P_2$ because there is always a better path that is extended from $P_1$.

Since the PACE model does consider dependency and does not use convolution to compute the distribution of a path, such pruning is no longer valid. 
As a result, massive amounts of candidate paths must be explored, which adversely affects efficiency. 
We propose virtual-path based stochastic routing for PACE as the second speedup technique. 
Here, we iteratively combine overlapping T-paths into so-called virtual paths (V-paths for short), such that time consuming distribution combination among overlapping T-paths is pre-computed before routing. 
Further, as the V-paths preserve the dependencies in the T-paths, it is possible to use stochastic dominance based pruning with the V-paths, which improves efficiency. 

To summarize, we make three contributions in the paper. 
First, we propose two heuristics that are able to estimate the costs of travel from intermediate vertices to the destination in the PACE model. 
Second, we propose a method that introduces V-paths in the PACE model, which enables effective pruning. 
Third, we conduct comprehensive experiments using two real-world data sets to gain detailed insight into the effectiveness of the speedup techniques. 

The paper is organized as follows. 
Section~\ref{sec:preliminaries} covers the basics of the PACE model. 
Section~\ref{sec:heur} presents the two versions of search heuristics followed by virtual-path based stochastic routing in Section~\ref{sec:vrouting}. 
Section~\ref{sec:exp} offers the experimental studies, Section~\ref{sec:related} reviews related work, and Section~\ref{sec:con} concludes.

\section{Preliminaries}
\label{sec:preliminaries}

We cover the \emph{edge}-centric and \emph{path}-centric models and formalize the problem to be solved. 

\vspace{-1mm}
\subsection{Edge-Centric Model (EDGE)}
\label{subsec:rn}

The standard model for stochastic routing is the \emph{edge-centric model (EDGE)}~\cite{nie2006arriving, niknami2016tractable, DBLP:conf/icde/YangGJKS14}. 
Here, a road network is modeled as a directed graph $\mathcal{G}=(\mathbb{V}, \mathbb{E}, W)$, where $\mathbb{V}$ is a set of vertices that represent road intersections or ends of roads, and $\mathbb{E} \subseteq \mathbb{V} \times \mathbb{V}$ is a set of edges that represent directed road segments. 
Weight function $W: \mathbb{E}\rightarrow \mathbb{D}$ maintains uncertain weights, e.g., travel-time distributions, for all edges in $\mathbb{E}$, where $\mathbb{D}$ is a set of distributions. 
As already explained in the introduction, the weight function is instantiated based on trajectory pieces. 
Figure~\ref{fig:origin_graph} shows a road network with 8 vertices and 10 edges along with the stochastic weights of the edges. 
For example, the stochastic weight of edge $e_1$ is a distribution of $\{[8, 0.9], [10, 0.1]\}$, meaning that traversing $e_1$ may take 8 or 10 time units with probability 0.9 and 0.1, respectively.

A path is a sequence of adjacent edges $P=\langle e_1, e_2$, $\ldots$, $e_n\rangle$, where $n \geq 1$, $e_i\neq e_j$ if $1 \leq i, j \leq n$ and $i\neq j$, and $e_{i}$ and $e_{i+1}$ are adjacent where $1 \leq i <n$. 
Next, we define sub-path $P'=\langle e_i, e_{i+1}$, $\ldots$, $e_{j-1}$, $e_j\rangle$ of path $P$ as a sub-sequence of edges of the path, where $1 \leq i<j\leq n$. 
For example, path $\langle e_1, e_4 \rangle$ is a sub-path of path $\langle e_1, e_4, e_9 \rangle$ in Figure~\ref{fig:origin_graph}. 

\vspace{-1.2em}
\begin{figure}[!htbp]
	\begin{center}
		\includegraphics[scale=0.5]{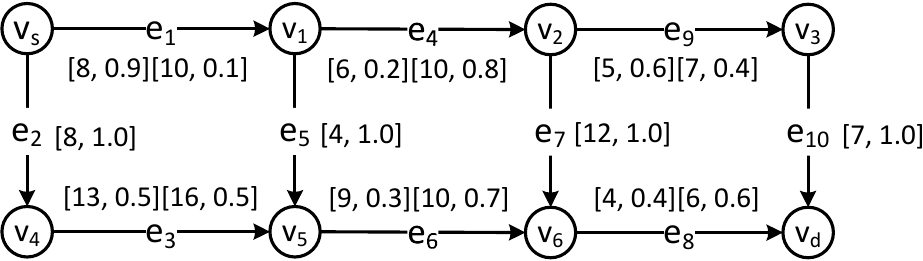}
	\end{center}
         \vspace{-1em}
	\caption{Edge-Centric Uncertain Road Network, EDGE}
	\label{fig:origin_graph}
\end{figure}
\vspace{-1.3em}

An important function in stochastic routing is to compute the cost distribution of a path. 
The EDGE model assumes that the uncertain weights of different edges are independent of each other. 
Thus, \emph{convolution} is used to sum distributions. 
Given a path $P=\langle e_1, e_2, \ldots, e_n \rangle$, the cost distribution of $P$ is computed as 
$D(P)=\oplus_{i=1}^{n} W(e_i)=W(e_1)\oplus W(e_2) \oplus \ldots \oplus W(e_n)$, where $\oplus$ is the convolution operator. 

\subsection{PAth-CEntric Model (PACE)}
\label{subsec:pace}

The path-centric model (PACE)~\cite{pvldb17pathcost, DBLP:journals/vldb/YangDGJH18}
models a road network as a directed graph $\mathcal{G}^p=(\mathbb{V}, \mathbb{E}, \mathbb{P}, \mathbb{W})$, where $\mathbb{V}$ and $\mathbb{E}$ still represent vertices and edges. 
Next, $\mathbb{P}$ is a set of paths that have been traversed by at least $\tau$ trajectories. 
We call these trajectory paths (a.k.a., T-paths). 
We evaluate the impact of threshold $\tau$ in experiments. 
Finally, $\mathbb{W}$ is a weight function set. The first weight function $W: \mathbb{E}\rightarrow \mathbb{D}$ returns a cost distribution for each edge $e\in \mathbb{E}$, as in the EDGE model. 
For each T-path, we maintain weight functions $W_J$ and $W$, both with signature $\mathbb{P}\rightarrow \mathbb{D}$. 
Given a T-path, $W_J$ returns its joint distribution that models the cost dependency among the edges in the T-path, e.g.,  as seen in Table~\ref{tbl:PACE_weight}(a); 
and $W$ returns a distribution that captures the total cost of the T-path, e.g., as seen in Table~\ref{tbl:PACE_weight}(b).

Figure~\ref{fig:pace_graph} shows a PACE graph with 5 T-paths and the same road network in Figure~\ref{fig:origin_graph}.  
Due to the space limitation, Figure~\ref{fig:pace_graph} only shows the total cost distributions of T-paths and omits the joint distributions. Next, we sketch how to compute the distribution of a path $P=\langle e_1, e_2, \ldots, e_n \rangle$ in PACE. 
As illustrated in the introduction, this distribution can generally be assembled from different sets of edge and T-path costs. 
It has been shown that the coarsest combination, i.e., the combination with the longest overlapping T-paths, gives the most accurate uncertain travel time~\cite{pvldb17pathcost}. 
For example, when computing the distribution of $P=\langle e_1, e_4, e_9 \rangle$, the coarsest combination is $\{p_1, p_2\}$, which is coarser than combinations $\{p_1, e_9\}$, $\{e_1, p_2\}$, and $\{e_1, e_4, e_9\}$.
\begin{figure}[!htp]
	\begin{center}
		\includegraphics[scale=0.45]{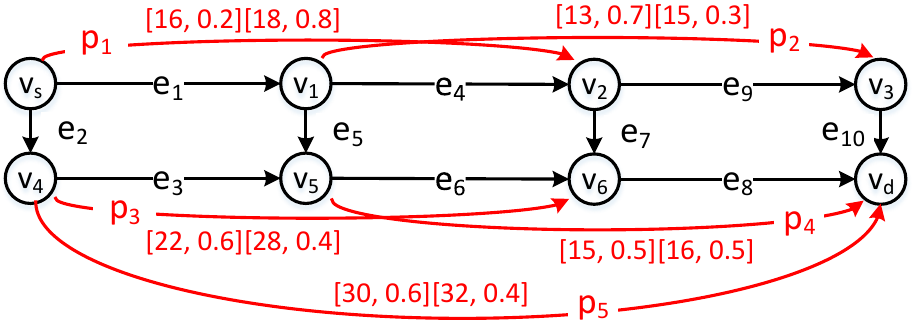}
	\end{center}
        \vspace{-1em}
	\caption{Path-Centric Uncertain Road Network, PACE}
	\label{fig:pace_graph}
 \vspace{-2em}
\end{figure}
After identifying the coarsest T-path sequence ${CPS(P)}=(p_1, p_2, \ldots, p_m)$ for path $P$, we use a T-path assembly operation $\diamond$ to compute the joint distribution of $P$~\cite{DBLP:journals/vldb/YangDGJH18} as follows:   
\begin{equation}
\label{eqn_pace}
D_J(P)  
= p_1 \diamond p_2 \diamond \ldots \diamond p_m
= \frac{\Pi_{i=1}^{m}W_J(p_i)}{\Pi_{i=1}^{m-1}W_J(p_i\cap p_{i+1})},
\end{equation}
where $p_i\cap p_{i+1}$ denotes the overlap sub-path of the T-paths $p_i$ and $p_{i+1}$.
From the joint distribution maintained in PACE, we can derive the distribution of the cost of the path $D_J(P)$ that captures cost dependencies. 
Consider the example of ${P}=\langle e_1, e_4, e_9\rangle$. 
In the PACE model, we have $D_J({P})=\frac{W_J(p_1)\cdot W_J(p_2)}{W_J(\langle e_4 \rangle)}$ as $p_1 \cap p_2=\langle e_4 \rangle$. 
In the case of no overlapping between two sequential T-paths, they are assembled by a convolution operation, as implemented in the EDGE.

We make a note on the notation in the paper: (1) $W_J$ and $W$ represent distributions that are maintained in EDGE or PACE; (2) $D_J$ and $D$ represent distributions that are derived from $W_J$ and $W$; (3) $W_J$ and $D_J$ indicate joint distributions. 

Finally, EDGE and PACE each supports time-dependent uncertainty modeling by maintaining different uncertain graphs for different time periods, e.g., peak vs. off-peak hours.  

\vspace{-0.5em}
\noindent
\subsection{Problem Definition}
\label{subsec:probdef}

Given a source vertex $v_s$, a destination vertex $v_d$, a departure time $t$, and a travel cost budget $B$, we aim at identifying path $P^*$ from $v_s$ to $v_d$ that has the largest probability of arriving at destination $v_d$ within the budget $B$ when leaving at $t$. 
\vspace{-0.5em}
\begin{equation}
P^*=\arg \max_{{P}\in \mathit{Path}} \mathit{Prob}(D({P})\leq B),   
\end{equation}
where set $\mathit{Path}$ contains all paths from $v_s$ to $v_d$, and $D({P})$ refers to the cost distribution of path ${P}$ leaving at $t$ using the PACE model. 
For example, given travel time budget 
$B=60$ minutes, we aim at finding the path that has the highest probability of arriving $v_d$ within 60 minutes, if a driver departs at 8:00 a.m. 

\vspace{-1em}
\begin{algorithm}[h]
	\caption{Naive Stochastic Routing in PACE}
	\label{algo:naivepace}
	\begin{small}
		\KwIn{Source $v_s$; Destination $v_d$; Budget $B$; PACE graph $\mathcal{G^{P}}$;}
		\KwOut{The path with the largest probability of getting $v_d$ within $B$;}
		Priority queue $Q \leftarrow \emptyset$; \\
		Double $maxProb \leftarrow 0$;  Path $P^*\leftarrow  \emptyset$; \\
		$Q.\mathit{push}(\langle v_s, v_s \rangle, 0)$;\\
		\While{$Q \neq \emptyset$}
		{
		    Path $\hat{P} \leftarrow Q.peek()$;\\
		    \If{$\hat{P}$ reaches $v_d$ and $Prob(D(\hat{P})\leq B)>maxProb$}
		    {
		        $maxProb \leftarrow Prob(D(\hat{P})\leq B)$; $P^*\leftarrow \hat{P}$; 
		    }
		    \For{each outgoing edge or T-path $e$}
		    {
		        Extend $\hat{P}$ by edge $e$; \\
		        $Q.insert(\hat{P}, Expectation(D(\hat{P})))$;\\
		    }
		}
		\Return $P^*$;%
	\end{small}
\end{algorithm}
\vspace{-1em}

\noindent
\textbf{Limitations of existing stochastic routing in PACE: }
The Existing stochastic routing in PACE~\cite{DBLP:journals/vldb/YangDGJH18} (see Algorithm~\ref{algo:naivepace}) faces the two challenges.
First, it starts exploring candidate paths from source $v_s$, each time extending a candidate path with its neighboring edges. 
A priority queue is used to maintain candidate paths, and each candidate path has a priority, e.g., the expected cost of the candidate path or the probability that the candidate path has cost below the budget $B$. 
It keeps exploring candidate paths until no path with higher probability exists. 
This approach determines the candidate path exploration priority purely based on the distributions of different candidate paths, while ignoring the costs from intermediate vertices to the destination. 
This leads to unnecessary exploration of non-competitive paths (e.g., $P_3$ in Figure~\ref{fig:intro}), yielding reduced efficiency. 

Second, stochastic dominance based pruning is inapplicable in the PACE model. 
Consider $P_1$ and $P_2$ in Figure~\ref{fig:intro}. 
When extending both $P_1$ and $P_2$ with edge $e$ to get two new paths $P_1^\prime$ and $P_2^\prime$, $D(P_1)\diamond p_x$ and $D(P_2) \diamond p_y$ are the distributions for $P_1^\prime$ and $P_2^\prime$. 
Here, $p_x$ and $p_y$ are the longest T-paths that cover $e$ and some edges in $P_1$ and $P_2$, respectively, and $p_x$ and $p_y$  are often different. 
Thus, although $D(P_1)$ stochastically dominates $D(P_2)$, $D(P_1)\diamond p_x$ may not stochastically dominates $D(P_2) \diamond p_y$, as $p_x$ and $p_y$ may be different,  rendering stochastic dominance based pruning inapplicable. 
\section{Search Heuristics}
\label{sec:heur}
\subsection{Definition of Heuristics}
\label{subsec:heurDef}

A routing algorithm explores a search space of candidate paths.
During the exploration, it is beneficial to first explore paths that are expected to have high probabilities of arriving at the destination within the budget. 
Assume that we explore a path $P_i$ from source $v_s$ to an intermediate vertex $v_i$. 
We are able to compute the cost distribution of $P_i$ using the PACE model. 
However, we do not know the cost distribution of the path from intermediate vertex $v_i$ to the destination $v_d$ because we do not yet know that path---there may possibly be many paths from $v_i$ to $v_d$. 

We target a heuristic function $U(v_i, v_d, x)$ that is able to estimate the largest possible probability, among possibly multiple paths from $v_i$ to $v_d$, of having a cost within $x$ units.  
Supported by the branch-and-bound principle that is often used to solve optimization problems with admissible heuristics~\cite{goodrich2006algorithm}, we need to ensure that the heuristic function is \emph{admissible}, meaning that the function never \emph{underestimates} the probability of being able to travel from $v_i$ to $v_d$ within $x$ units. In other words, the heuristic function must provide an \emph{upper bound} on the probability. 

If the upper bound probability of a candidate path is lower than the probability of a path $P^*$ that already reaches the destination, then the candidate path can be pruned safely because it cannot get a higher probability than that of $P^*$. This guarantees correctness.

Since the destination is fixed for a given query, we often omit $v_d$ and use $U(v_i, x)$ when the query is clear. 
However, the heuristic function is destination-specific, meaning that $U(v_i, x)$ is different for different query destinations. Based on the heuristic function, we are able to estimate the maximum probability of arriving at destination $v_d$ within budget $B$ when using path $P_i$, as shown in Eq.~\ref{eqn_est}. 
\begin{equation}
\label{eqn_est}
\mathit{maxProb}(P_i, B) = \sum\limits_{t=1}^{B} D(P_i).\mathit{pdf}(t) \cdot U(v_i, B-t),  
\end{equation}
where $D(P_i).\mathit{pdf}(t)$ denotes the probability of the travel taking $t$ units when following path $P_i$ from source $v_s$ to intermediate vertex $v_i$, which is computed using the PACE model. Next, $U(v_i, B-t)$ denotes the largest probability of travel from $v_i$ to $v_d$ costing at most $B-t$ units, which is returned by the heuristic function.   

The intuition behind Eq.~\ref{eqn_est} is as follows. The maximum probability of arriving at the destination within $B$ units when using path $P_i$ is the sum, over all possible costs $t$, of the products of the probability that $P_i$ takes $t$ units and the maximum probability that the remaining path from $v_i$ to $v_d$ takes at most $B-t$ units, as illustrated in Figure~\ref{fig:heuristic_func}(a). 
In Figure~\ref{fig:heuristic_func}(b), $P_i$'s distribution is \{[8, 0.9], [10, 0.1]\}. 
When budget $B=25$, we get $\mathit{maxProb}(P_i$, $25) = 0.9 \cdot U(v_1,17) + 0.1 \cdot U(v_1, 15)$. 
\vspace{-1em}
\begin{figure}[H]
	\begin{center}
        \includegraphics[scale=0.55]{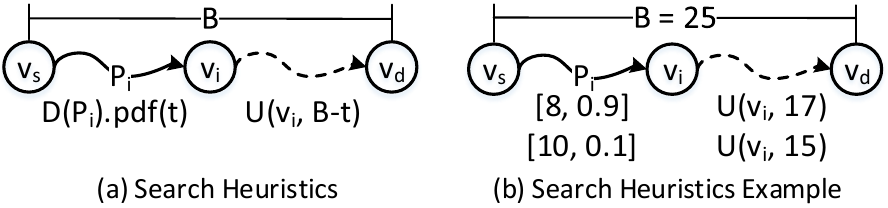}
	\end{center}
        \vspace{-1.5em}
	\caption{Search Heuristics}
	\label{fig:heuristic_func}
 \vspace{-1em}
\end{figure}

With the heuristic function, we are able to first explore more promising paths, i.e., paths that have higher $\mathit{maxProb}$ values. 
In addition, it enables early stopping. If the path with the largest $\mathit{maxProb}$ value, say $P^*$, already reaches the destination, we can safely stop exploring other paths. 
This is because, even when using admissible heuristics, the other paths can only achieve probabilities that are no greater than $P^*$'s probability. 
Thus, in the best case, the other paths cannot have a larger probability than that of $P^*$ and it is safe to stop. 

We proceed to introduce two different strategies for instantiating the admissible heuristic function $U(v_i, x)$---a binary heuristic and a budget-specific heuristic. 

\subsection{Binary Heuristic}
\label{subsec:T_heuristic}
The binary heuristics simplifies function $U(v_i, x)$ by only distinguishing between probability values 0 and 1---whether or not it is possible to go from $v_i$ to the destination within $x$ units. 

To this end, we maintain an auxiliary function $v_i.\mathit{getMin}()$ for each vertex $v_i$ that returns the {\it least travel cost} required to travel from $v_i$ to destination $v_d$. 
We define the binary heuristics in Eq.~\ref{eqn_ui}. 
\vspace{-0.5em}
\begin{equation}
\label{eqn_ui}
U(v_i, x)= \begin{cases}
0 &\text{if $x < v_i.\mathit{getMin}()$}\\
1 &\text{if $x \geq v_i.\mathit{getMin}()$}
\end{cases}
\vspace{-0.5em}
\end{equation}

If $x$ is below $v_i.\mathit{getMin}()$, it is impossible to go from $v_i$ to $v_d$ within $x$ units. 
As $v_i.\mathit{getMin}()$ is the smallest travel cost within which it is possible to reach $v_d$, it is impossible to reach $v_d$ with a travel cost $x$ that is smaller than $v_i.\mathit{getMin}()$. 
Thus, the highest probability of reaching $v_d$ within $x$ units is 0. 
Otherwise, if $x \geq v_i.\mathit{getMin}()$, it is possible to reach $v_d$ within $x$ units. 
In this case, we set $U(v_i, x)=1$, which is a very optimistic estimate. 
The resulting binary heuristic is \emph{admissible} because the probability of reaching $v_d$ within $x$ units is at most 1, meaning that the heuristic never underestimates the probability. 
In other words, the actual probability of reaching $v_d$ from $x_i$ may be less than 1 because it can also take more than $x$ units. 

Next, we instantiate function $v_i.\mathit{getMin}()$ for a given destination $v_d$.
The main idea is to perform a backward search from destination $v_d$ to all other vertices and then annotate each vertex with the least travel cost from the vertex to $v_d$.  

To this end, we first build a reversed graph $\mathcal{G}^p_\mathit{rev}=(\mathbb{V}, \mathbb{E}', \mathbb{P}', \mathbb{W}')$ from $\mathcal{G}^p=(\mathbb{V}, \mathbb{E}, \mathbb{P}, \mathbb{W})$, where the vertices remain unchanged but the directions of the edges in $\mathbb{E}'$ and T-paths in $\mathbb{P}'$ are reversed compared to the counterparts in $\mathbb{E}$ and $\mathbb{P}$. 
Weight function $\mathbb{W}'$ maintains only deterministic costs, thus mapping a reversed edge or T-path to the minimum cost in its distribution. 
For example, Figure~\ref{fig:rev_graph} includes a reversed edge $e'_1 = (v_1, v_s)$ has cost $\mathbb{W}'(e'_1)=8$ because of edge $e_1 = (v_s, v_1)$ in $\mathcal{G}^p$ (cf.\ Figure~\ref{fig:origin_graph}) that has cost distribution $\mathbb{W}(e_1)=\{[8, 0.9]$, $[10, 0.1]\}$, in which 8 is the minimum cost. 
Similarly, %
T-path $p_1 = \langle e_1, e_4\rangle$ in Figure~\ref{fig:origin_graph} yields the reversed T-path $p'_1 = \langle e'_4, e'_1\rangle$ in Figure~\ref{fig:rev_graph}, with $W'(p'_1)=16$.

Next, we explore the reversed graph $\mathcal{G}^p_\mathit{rev}$ starting from $v_d$, annotating each vertex with the least cost from $v_d$ to the vertex.  
If only considering reversed edges and disregard reversed T-paths~\cite{DBLP:conf/mdm/Andonov018}, Dijkstra's algorithm could produce a shortest path tree, as exemplified in Figure~\ref{fig:heuristic_graph}(a). However, the reversed T-paths render Dijkstra's algorithm inapplicable.

For example, when only considering reversed edges, the vertex $v_5$ is annotated with 13 (i.e., the sum of the weights of $e'_8$ and $e'_6$).
 \vspace{-1em}
\begin{figure}[h]
	\begin{center}
		\includegraphics[scale=0.44]{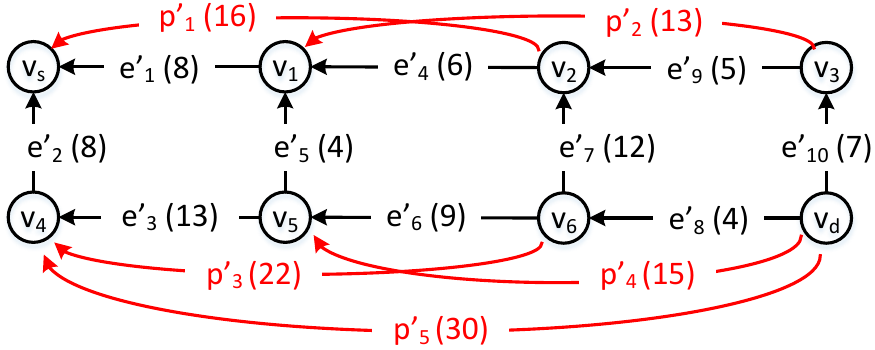}
	\end{center}
         \vspace{-1.5em}
	\caption{Reversed Graph $\mathcal{G}^p_\mathit{rev}$} 
	\label{fig:rev_graph}
 \vspace{-2em}
\end{figure}

\begin{algorithm}[h]
	\caption{Shortest Path Tree Generation in PACE}
	\label{algo:1-all-Dij}
	\begin{small}
		\KwIn{$\mathcal{G_{\mathit{rev}}^{P}}=(\mathbb{V}, \mathbb{E'}, \mathbb{P'}, \mathbb{W'})$, $v_d$;}
		\KwOut{Shortest path tree $T$ from $v_d$;}
		$\mathit{Init}(\mathbb{V})$;\space\space\space
		$/\ast$ \texttt{set} $c_1 \leftarrow \infty$ and $c_2 \leftarrow 0 \texttt{ for each vertex}$	$\ast/$\\
		Shortest path tree $T \leftarrow \emptyset$; \\
		Priority queue $Q \leftarrow \emptyset$; \\
		$v_d.c_1 \leftarrow 0$; \\
		$Q.\mathit{push}(v_d)$;\\
		\While{$Q \neq \emptyset$}
		{
			$v \leftarrow Q.peek()$;\\
			$T.add(v)$;\\
			\For{each vertex $u$ that is reachable from $v$ by a reversed edge or T-path $ep$} 
			{
				$\hat{c}_1 \leftarrow v.c_1 + W'(ep)$;\\
				$\hat{c}_2 \leftarrow v.c_2 + \mathit{countEdges}(ep)$;\\
				\Switch{$\mathit{checkDominance}(\hat{c}_1, \hat{c}_2, u.c_1, u.c_2)$}{
					\Case{\textsc{Non-domination}}{
						$\hat{P}_{\mathit{old}} \leftarrow \mathit{tracePath}(v_d,u)$;\\
						$\hat{P}_{\mathit{new}} \leftarrow \mathit{tracePath}(v_d,v) + \langle v, u\rangle$;\\
						\If{($\hat{P}_{\mathit{old}} \neq \hat{P}_{\mathit{new}}$) \&\& ($\hat{c}_1 < u.c_1$)}{					
							$u.c_1 \leftarrow \hat{c}_1$; $u.c_2 \leftarrow \hat{c}_2$; $u.\mathit{parent}\leftarrow v$;\\
							$Q.\mathit{add}(u)$;					
						}
						\If{($\hat{P}_{\mathit{old}} == \hat{P}_{\mathit{new}}$) \&\& ($\hat{c}_2 > u.c_2$)}{							
							$u.c_1 \leftarrow \hat{c}_1$; $u.c_2 \leftarrow \hat{c}_2$; $u.\mathit{parent}\leftarrow v$;\\
							$Q.\mathit{add}(u)$;						
						}
					}
					\Case{\textsc{Domination}}{
						$u.c_1 \leftarrow \hat{c}_1$; $u.c_2 \leftarrow \hat{c}_2$; $u.\mathit{parent}\leftarrow v$;\\
						$Q.\mathit{add}(u)$;
					}
				}
			}
		}
		\Return $T$;%
	\end{small}
\end{algorithm}
\noindent When considering reversed T-paths, reversed T-path $p'_4$ has a weight $W'(p'_4) =15$, which exceeds the sum of the weights of $e'_8$ and $e'_6$. 
According to the theory underlying the PACE model~\cite{pvldb17pathcost, DBLP:journals/vldb/YangDGJH18}, the T-path weight is more accurate as it preserves the dependency between the two edges better. 
Thus, we should annotate $v_5$ with the reversed T-path cost $W'(p'_4) =15$. 
However, if we simply run Dijkstra's algorithm, $v_5$ is annotated with 13 as it is smaller than 15.
This exemplifies why we need a new algorithm that is able to consider reversed T-paths that produce more accurate weights in PACE.

Algorithm \ref{algo:1-all-Dij} shows a new shortest path tree generation algorithm that takes into account the weights of the T-paths. 
We maintain two costs $c_1$ and $c_2$ for each vertex $v_i$. 
Cost $c_1$ represents the shortest path weight from $v_d$ to $v_i$, where the shortest path may use both reversed edges and reversed T-paths.  
Cost $c_2$ represents the number of edges in the reversed T-paths that are used to connect $v_d$ to $v_i$. 
\noindent We prefer low values for $c_1$ and high values for $c_2$, meaning that we prefer shortest paths and reversed T-paths with many edges. 
We use a priority queue $Q$ to maintain the vertices, where the priority is according to cost $c_1$. 
\vspace{-1em}
\begin{figure}[!htp]
	\begin{center}
		\includegraphics[scale=0.45]{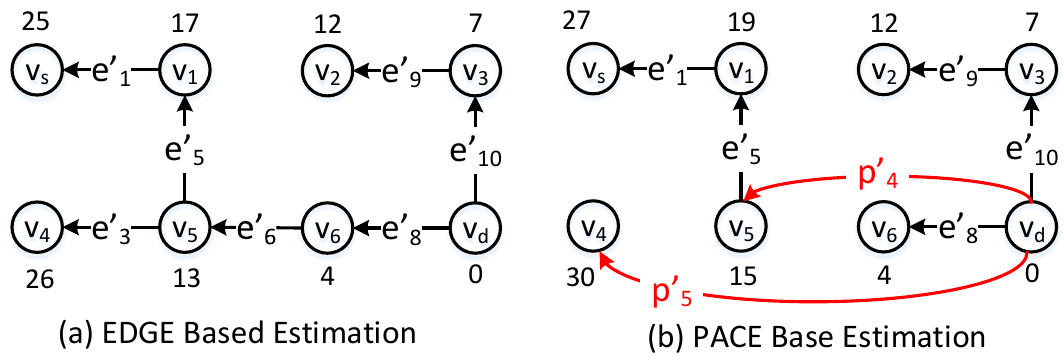}
	\end{center}
        \vspace{-1.5em}
	\caption{Computing $v_i.\mathit{getMin}()$, Given Destination $v_d$}
	\label{fig:heuristic_graph}
\end{figure}
\vspace{-1em}

\vspace{-1em}
 \begin{table}[!htp]
	\centering
	\renewcommand{\arraystretch}{1}
	\begin{tabular}{|c|l|c||c|l|c||c|l|c|} \hline
	    \multicolumn{3}{|c||}{Iteration 1} & \multicolumn{3}{c||}{Iteration 2} &
		\multicolumn{3}{c|}{Iteration 3} \\\hline
		\multicolumn{3}{|c||}{$v_d.c_1=0$} & \multicolumn{3}{c||}{$v_6.c_1=4$} &
		\multicolumn{3}{c|}{$v_3.c_1=7$} \\\hline
		$Q$ & $c_1$, $c_2$  & $\mathit{Par}$ & $Q$ & $c_1$, $c_2$  & $\mathit{Par}$ & $Q$ & $c_1$, $c_2$  & $\mathit{Par}$\\ \hline 
		$v_6$ & 4, 0 & $e'_8$ & $v_3$ & 7, 0 & $e'_{10}$ & $v_5$ & 15, 3 & $p'_4$\\ \hline 
		$v_3$ & 7, 0 & $e'_{10}$ & $v_5$ & 15, 2 & $p'_4$ & $v_2$ & 12, 0 & $e'_9$\\ \hline 
		$v_5$ & 15, 2 & $p'_4$ & $v_2$ & 16, 0 & $e'_7$ & $v_1$ & 20, 2 & $p'_2$\\ \hline 
		$v_4$ & 30, 3 & $p'_5$ & $v_4$ & 30, 3 & $p'_5$ & $v_4$ & 30, 3 & $p'_5$\\ \hline 		
	\end{tabular}
	\vspace{0.5em}
	\caption{Shortest Path Tree Generation in PACE} 
	\label{tbl:dij}
\vspace{-2.5em}
\end{table}

Table~\ref{tbl:dij} shows the first three iterations when applying Algorithm~\ref{algo:1-all-Dij} to $\mathcal{G}^p_\mathit{rev}$ from Figure~\ref{fig:rev_graph}. In the first three iterations, the priority queue peeks vertices $v_d$, $v_6$, and $v_3$, respectively, as shown in the top row of Table~\ref{tbl:dij}. 
For each iteration, we show the vertices in the priority queue $Q$, the costs $c_1$ and $c_2$ of each vertex, and the reversed edge or T-path that connects the vertex from its parent in the shortest path tree, denoted by $\mathit{Par}$.

Starting from $v_d$, its neighbor vertices $v_6$, $v_3$, $v_5$, and $v_4$ are explored by following reversed edges $e'_8$ and $e'_{10}$ and reversed T-paths $p'_4$ and $p'_5$, respectively (Line 9), resulting in $c_1$ values 4, 7, 15, and 30 (Line 10). 
Since vertices $v_5$ and $v_4$ are reached from $v_d$ via reversed T-paths $p'_4 = \langle e'_8, e'_6\rangle$ and $p'_5 = \langle e'_8, e'_6, e'_3\rangle$, their $c_2$ values are 2 and 3, indicating that there are 2 and 3 edges in the reversed T-paths that connect $v_d$ to $v_5$ and to $v_6$, respectively (using function $\mathit{countEdge}(ep)$ in Line 11). 
The $c_2$ values for $v_6$ and $v_3$ are zero because they do not use reversed T-paths. 

In the second iteration, as shown in Table~\ref{tbl:dij}, the priority queue peeks $v_6$, which has three neighbors $v_5$, $v_2$, and $v_4$ by following $e'_6$, $e'_7$, and $p'_3$. When visiting $v_5$ by the reversed edge $e'_6$, we obtain $\hat{c}_1=13$ and $\hat{c}_2=0$ (Lines 10 and 11). 
Now, we need to consider whether we should replace the existing $c_1$ and $c_2$ values for vertex $v_5$ in the priority queue with $\hat{c}_1$ and $\hat{c}_2$.

To utilize the accurate cost values maintained in reversed T-paths, a smaller $c_1$ but a larger $c_2$ are preferred, meaning that we aim at (1) identifying the least cost from $v_d$ to a vertex and at (2) maximizing the use of reversed T-paths, especially long ones, during the search. 

Following this principle, we need to consider both values when making a decision on whether to update an existing vertex in the priority queue based on a new candidate path. 
More specifically, vertex $v_5$ has $c_1=15$ and $c_2=2$ through $\langle p'_4\rangle$ in the priority queue. 
Now, the new candidate path $\langle e'_8, e'_6\rangle$ has $\hat{c}_1=13$ and $\hat{c}_2=0$. 
We apply a $\mathit{checkDominance}(\cdot)$ function (Line 12) to compare the two values following the concept of pareto-optimality. 
We define two relationships---\textsc{non-domination} and \textsc{domination}.

\textsc{Non-domination} (lines 13--21). If the new candidate path is better in one value but worse in the other than those already in the priority queue, we are in the non-domination case. 
For example, since $\hat{c}_1<v_5.c_1$ and $\hat{c}_2<v_5.c_2$, this is a non-domination case.

In the non-domination case, we consider the new candidate path $\hat{P}_{new}$ and the old path $\hat{P}_{old}$ that leads to the values in the priority queue. 
We check whether they correspond to the same path in the road network (using function $\mathit{tracePath}(\cdot)$ in lines 14--15). 
If yes, following criterion (2), we use the path with more edges covered by T-paths as the corresponding cost should be more accurate. 
If no, we follow criterion (1) and use the path that has a smaller $c_1$ value to make sure that we find the minimum cost of candidate paths. 

In the above example, paths $\hat{P}_{new}$ and $\hat{P}_{old}$ both correspond to the same path $\langle e'_8, e'_6\rangle$. 
We then keep the values in the priority queue since $v_5.c_2=2$ indicates that two edges are covered by T-paths, while $\hat{c}_2=0$ indicates that no edges are covered by T-paths in the new candidate path  $\hat{P}_{new}$.

\textsc{Domination} (Lines 22--24). If the candidate path is better in one value and not worse in the other value than those already in the priority queue, domination occurs. 
In the domination case, we update the priority queue using the candidate path. 

For example, at the third iteration, where $v_3$ is peeked from the priority queue, $v_2$ is visited again as $v_3$'s neighbor via $e'_9$. We now have $v_2.c_1=16$ and $v_2.c_2=0$, and $\hat{c}_1=12$ and $\hat{c}_2=0$. 
Since $\hat{c}_1 < u.c_1$ and $\hat{c}_2 = u.c_2$, we update $v_2.c_1$ and $v_2.c_2$ to be $\hat{c}_1$ and $\hat{c}_2$. 
If the candidate path is dominated by the values already in the priority queue, we do nothing. 

Figure~\ref{fig:heuristic_graph}(b) shows the identified shortest path tree using both reversed edges and T-paths.

The time complexity for building the binary heuristics (Algorithm~\ref{algo:1-all-Dij}) is $O((|\mathbb{E'}|+|\mathbb{P'}|)lg|\mathbb{V}|)$, where $|\mathbb{E'}|$, $|\mathbb{P'}|$ and $|\mathbb{V}|$ represent the number of reversed edges, reversed T-paths, and vertices (i.e., road intersections or ends of roads), respectively. 

\begin{figure}[h]
	\begin{center}
		\includegraphics[scale=0.6]{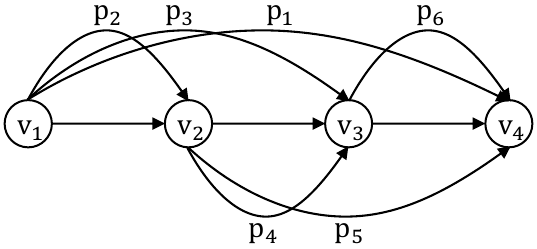}
	\end{center}
	\caption{Motivating Example of Generating T-path}
	\label{fig:t-path}
 \vspace{-1.5em}
\end{figure}

The number of T-paths $|\mathbb{P'}|$ is polynomial in the length of the longest T-path. For instance, as illustrated in Figure~\ref{fig:t-path}, assume that the longest T-path is $p_1=\langle v_1, v_2, v_3, v_4 \rangle$. From $v_1$, it generates three T-paths: $p_1=\langle v_1, v_2, v_3, v_4 \rangle$, $p_2=\langle v_1, v_2 \rangle$ , and $p_3=\langle v_1, v_2, v_3 \rangle$. 
Likewise, from $v_2$, it generates two T-paths: $p_4=\langle v_2, v_3 \rangle$, $p_5=\langle v_2, v_3, v_4\rangle$; and from $v_3$, it generates one T-path: $p_6=\langle v_3, v_4 \rangle$. Hence, the total number of T-paths we can generate is $\frac{n(n-1)}{2}$, where $n$ is the length of the longest T-path.

Specifically, in the worst case, we visit all edges and T-paths, thus having $|\mathbb{E'}|+|\mathbb{P'}|$. For each visit, we need to update the priority queue, which at most has $|\mathbb{V}|$ elements, and thus the update takes at most $lg|\mathbb{V}|$.

The space complexity of Algorithm~\ref{algo:1-all-Dij} is $O(|\mathbb{V}|^2)$. For each destination, each vertex maintains a binary heuristic value.

\vspace{-1em}
\subsection{Budget-Specific Heuristic}
\label{subsec:policy}

The binary heuristic often provides overly optimistic probabilities from an intermediate vertex to the destination, which can lead to unnecessary exploration of paths that cannot provide high probabilities of arriving at the destination within the budget. 
To enable more accurate heuristics, we propose a budget-specific heuristic that formulates function $U(v_i, x)$ at a finer granularity so that it can return probabilities in the range [0, 1].
Specifically,
\begin{equation}
\vspace{-0.5em}
\begin{aligned}
    U(v_i, x) &= \max_{z\in ON(v_i)} H(v_i, z, x), \\
H(v_i, z, x)& =  \sum_{k=1}^{x} \mathbb{W}(\langle v_i,  z\rangle).\mathit{pdf}(k) \cdot U(z, x-k),
\end{aligned}
\label{eqn_U}
\end{equation} 
\noindent and $ON(v_i)$ includes $v_i$'s outgoing neighbor vertices. 
Recall that the heuristic function $U(v_i, x)$ returns the largest probability of reaching $v_d$ from $v_i$ within $x$ units. 
Eq.~\ref{eqn_U} computes this probability by getting the largest $H(v_i, z, x)$ over all $v_i$'s outgoing neighbor vertices $z$, where  $H(v_i, z, x)$ represents the probability of arriving at $v_d$ within time budget $x$ when following $\langle v_i,  z\rangle$ by first going from $v_i$ to vertex $z$ and then to $v_d$. Note that $\langle v_i,  z\rangle$ may be an edge or a T-path.  
Since the heuristic function never ignores any possibility leaving from $v_i$, the function never underestimates the probability and thus is \emph{admissible}. 
More specifically, we define $H(v_i, z, x)$ as the sum, over all possible costs $k\in[1, x]$, of the products of the probability that it takes $k$ units from $v_i$ to vertex $z$, i.e., $\mathbb{W}(\langle v_i,  z\rangle).\mathit{pdf}(k)$, and the probability that it takes at most $x-k$ units to travel from $z$ to destination $v_d$, i.e., $U(z, x-k)$.

\subsubsection{Heuristic Tables}
As we aim at providing $U(v_i, x)$ at a fine granularity, we need to consider different values for $x$. To this end, we define the finest budget granularity $\delta$ and then consider multiple budget values $\delta$, $2\cdot \delta$, $\ldots$, and $\eta \cdot \delta$ for $x$. The largest value  $\eta \cdot \delta$ ensures that $U(v_i, \eta\cdot \delta)=1$ for all vertices $v_i$ given a destination vertex $v_d$. 
Then, we can represent function $U(v_i, x)$ as a \emph{heuristic table} with $|\mathbb{V}|$ rows and $\eta$ columns. 
The cell in the $i$-th row and $j$-th column represents value  $U(v_i, j \cdot \delta)$. 

Table~\ref{tbl:U} shows an example heuristic table for destination vertex $v_d$ based on the graph shown in Figures~\ref{fig:origin_graph} and~\ref{fig:pace_graph}.
Here, $\delta$ is set to 3, and the largest budget is 36 because no matter at which vertex we start, the probability that we reach destination $v_d$ within 36 units is 1. 
The blue cell in the row of vertex $v_5$ and the column of budget $x=15$ denotes $U(v_5, 15)=0.5$, meaning that the largest probability of reaching $v_d$ from $v_5$ within time budget 15 is 0.5.  

The binary heuristic can also be represented as a heuristic table. 
Specifically, for vertex $v_i$, we identify the smallest budget value in the table that is larger than or equal to $v_i.\mathit{getMin}()$, say $j \cdot \delta$. 
Then, in the row for $v_i$, the first $j-1$ cells are 0, and the remaining cells are 1. 
For example, Table~\ref{tbl:Ubinary} shows the binary heuristic for vertex $v_1$ in a heuristic table with $\delta=3$. 
 
\begin{table}[!htp]
	\centering
	\scriptsize
	\begin{tabular}{|c|c|c|c|c|c|c|c|c|c|c|c|c|} \hline
		x& 3&6&9&12&15&18&21&24&27&30&33&36 \\\hline	\hline
		$v_s$	& 0&0&0&0&0&0&0&0&0&0.18&0.91&1 \\\hline	
		$v_1$	& 0&0&0&0&0&0&\cellcolor{green!50}0.12&1&1&1&1&1 \\\hline	
		$v_2$	& 0&0&0&0&0.6&1&1&1&1&1&1&1 \\\hline	
		$v_3$	& 0&0&1&1&1&1&1&1&1&1&1&1 \\\hline	
		$v_4$	& 0&0&0&0&0&0&0&0&0&0.6&1&1 \\\hline	
		$v_5$	& 0&0&0&0&\cellcolor{blue!35}0.5&1&1&1&1&1&1&1 \\\hline	
		$v_6$	& 0&1&1&1&1&1&1&1&1&1&1&1 \\\hline
		$v_d$	& 1&1&1&1&1&1&1&1&1&1&1&1 \\\hline			
	\end{tabular}
        \vspace{0.5em}
	\caption{Heuristic Table of $U(v_i, x)$, Destination  $v_d$}\label{tbl:U}
 \vspace{-3.5em}
\end{table}

\begin{figure*}[h]
    \centering
    \centering
	\subfigure[Computing $U(v_i, x)$]{
		\centering
			\includegraphics[scale=0.425]{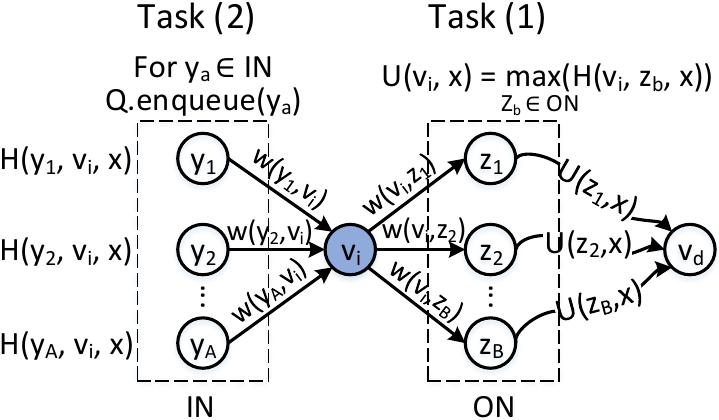}
	}
	\subfigure[Computing $U(v_d, x)$]{
		\centering
			\includegraphics[scale=0.425]{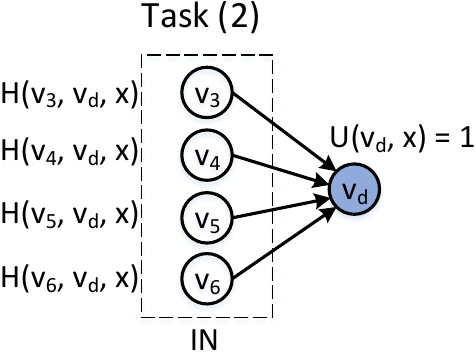}
	}
	\subfigure[Computing $U(v_3, x)$]{
		\centering
			\includegraphics[scale=0.425]{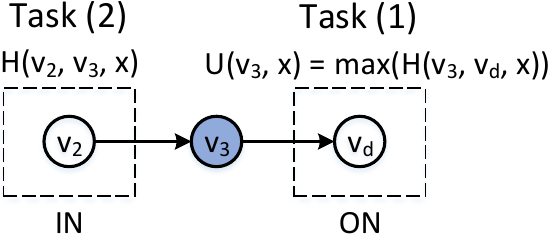}
	}
	\subfigure[Computing $U(v_4, x)$]{
		\centering
			\includegraphics[scale=0.425]{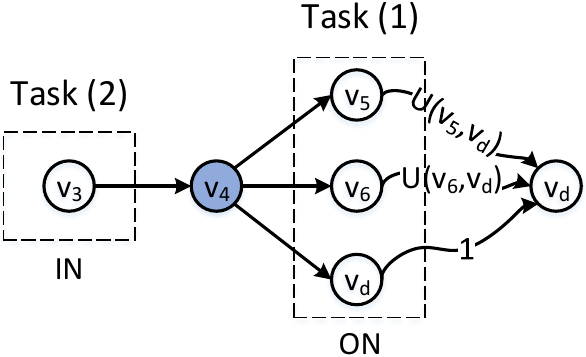}
	}
    \vspace{-1.5em}
    \caption{Computing Budget Specific Heuristic: (a) Provides an Abstraction of the Heuristic Computation for any Vertex $v_i$; (b), (c) and (d) Show Concrete Heuristic Computations for Vertices $v_d$, $v_3$, and $v_4$, Respectively}
    \vspace{-1.2em}	
    \label{fig:policy}
\end{figure*}

\vspace{-1em}
\begin{table}[!htp]
	\centering
	\scriptsize
	\begin{tabular}{|c|c|c|c|c|c|c|c|c|c|c|c|c|} \hline
		x& 3&6&9&12&15&18&21&24&27&30&33&36 \\\hline	\hline
		$v_1$	& 0&0&0&0&0&0&\cellcolor{green!50}1&1&1&1&1&1 \\\hline	
	\end{tabular}
        \vspace{0.5em}
	\caption{Heuristic Table of $v_1$, Binary Heuristics, Destination $v_d$}
	\label{tbl:Ubinary}
 \vspace{-4em}
\end{table} 

The first 6 cells are 0, and the remaining cells are 1. 
This is because $7\cdot \delta=21$ is the first budget value that exceeds $v_1.\mathit{getMin}()=19$ (see Figure~\ref{fig:heuristic_graph}(b)). 
If we compare the two rows of $v_1$ in the budget-specific vs. 
binary heuristic tables, we see that $U(v_1, 21)=0.12$ in the former, which is $U(v_1, 21)=1$ in the latter. 
This suggests that the budget-specific heuristic is able to provide finer-granularity probabilities than the binary heuristic, which estimates the probability of an intermediate vertex $v_i$ reaching $v_d$ within the time budget more accurately. Thus, budget-specific heuristic helps path exploration such that more promising candidate paths can be  explored first.

\subsubsection{Instantiating Heuristic Tables}

Efficiently instantiating a budget-specific heuristic table is non-trivial. 
A naive and computationally expensive approach is to compute according to Eq.~\ref{eqn_U} for each of $|\mathbb{V}|\cdot \eta$ cells for each destination, i.e., for each of $|\mathbb{V}|^2\cdot \eta$ cells in the heuristic table.

We make two observations that can speed up the heuristic table instantiation. 
First, for each row, two budget values $l$ and $s$ exist, where $l$ is the smallest budget value whose cell is larger than 0 and $s$ is the smallest budget values whose cell is 1. 
Thus, all budget values smaller than $l$ should have cells of 0, and all budget values larger than $s$ should have cells of 1. 
In other words, the cells to the left of $l$'s cell are 0, and the cells to the right of $s$'s cell are 1. 
For example, in Table~\ref{tbl:U}, for $v_5$, we have $l=15$ and $s=18$, and cells to the left of 15 are 0 and the cells on the right of 18 are 1. 
This observation suggests that for each vertex, we need to identify the two budget values $l$ and $s$ and compute values only for the cells between $l$ and $s$. 

Second, all cells for destination vertex $v_d$ are 1, since the probability of going from the destination vertex to itself within any budget is always 1. 
Thus $U(v_d, x)=1$ no matter what $x$ is. 
This motivates us to compute according to Eq.~\ref{eqn_U} starting from the destination $v_d$. 

\begin{algorithm}[h]
	\caption{BudgetSpecificHeuristics}
	\label{algo:policy}
	\begin{small}
		\KwIn{$\mathcal{G}^{p}=(\mathbb{V}, \mathbb{E}, \mathbb{P}, \mathbb{W})$, %
			$v_d$, $\delta$, $\eta$, $\mathbb{U}\leftarrow \emptyset$;}
		\KwOut{Heuristic table $\mathbb{U}$;} %
		Queue $Q \leftarrow \emptyset$;\\  %
		$Q.\mathit{enqueue}(v_d)$;\\		
		\While{$Q \neq \emptyset$}
		{
			$v^* \leftarrow Q.dequeue()$;%
			 $/\ast$ \texttt{Perform Task (1)} $\ast/$\\
			\If{$v^*=v_d$}
			{				
				\For{each $i$ from 1 to $\eta$}
				{
				    $\mathbb{U}(v^*, i\cdot \delta) \leftarrow 1$ \\
				}
			}
			\Else
			{
				$ComputeOneRowU(\mathcal{G}^{p}, v^*, \delta, \eta, \mathbb{U})$;\\
			}
			$/\ast$ \texttt{Perform Task (2)} $\ast/$\\
			\For{each vertex $v$, where $\langle v, v^*\rangle \in\mathbb{E}\cup \mathbb{P}$} 
			{
				\If{The row for vertex $v$ in $\mathbb{U}$ is still empty} 
				{
					$Q.\mathit{enqueue}(v)$;		\\	  
				}  
			}
		}
		\Return $\mathbb{U}$; 
	\end{small}
\end{algorithm}

\begin{algorithm}[h]
	\caption{ComputeOneRowU}
	\label{algo:computeU}
	\begin{small}
		\KwIn{$\mathcal{G}^{p}=(\mathbb{V}, \mathbb{E}, \mathbb{P}, \mathbb{W})$, %
			$v_i$, $\delta$, $\eta$, $\mathbb{U}$;}
		$l \leftarrow$ the smallest budget that is no less than $v_i.\mathit{getMin}()$; \\
		$s \leftarrow \eta\cdot\delta$;\\ %
		\For{each each budget $x$ that is smaller than $l$}
		{
		    $U(v_i, x)\leftarrow 0$; $/\ast$ \texttt{left cells of $l$ are all 0} $\ast/$\\
		}
		$/\ast$ \texttt{Compute} $U(v_i, x)$ \texttt{using Eq.~\ref{eqn_U}} $\ast/$\\
		\For{each vertex $z$, where $\langle v_i, z\rangle \in\mathbb{E}\cup\mathbb{P}$,}
		{
			\If{The row for vertex $z$ in $\mathbb{U}$ is still empty} 
			{
				$ComputeOneRowU(\mathcal{G}^{p}, z, \delta, \eta,  \mathbb{U})$;\\
			}
			$x \leftarrow l$; \\ %
			\While{$x < \eta \cdot \delta$} %
			{
			    $H(v_i, z, x)\leftarrow 0$;\\
				\For{each cost $c\in \mathbb{W}(\langle v_i,z\rangle)$}
				{
					
					$H(v_i, z, x) \leftarrow H(v_i, z, x) + \mathbb{W}(\langle v_i,z\rangle).\mathit{pdf}(c)\cdot U(z, x-c)$;\\
					
				}
				
				\If{$H(v_i, z, x)<1$}
				{
				$x \leftarrow x + \delta$; \space\space\space $/\ast$ \texttt{next time budget} $\ast/$\\
				}
				\Else
				{
				    \Break;\\
				}
				
			}
			$s \leftarrow \min(s,x)$;\\ %
		}
		\For{each budget $x \in [l, s]$}
		{
			$U(v_i, x) \leftarrow \max_{\langle v_i, z\rangle\in\mathbb{E}\cup\mathbb{P}}(H(v_i, z, x))$;\\
		}	
		\For{each budget $x>s$}
		{
			$U(v_i, x) \leftarrow 1$;\\
		}
	\end{small}
\end{algorithm}

Based on the above two observations, we propose Algorithm~\ref{algo:policy} to instantiate a heuristic table for a given destination $v_d$. 
Algorithm~\ref{algo:policy} takes as input graph $\mathcal{G}^{p}$ and destination $v_d$, as well as $\delta$ and $\eta$ that specify different budget values. 
Algorithm~\ref{algo:policy} returns heuristic table $\mathbb{U}$, where each row includes the heuristic function values $U(v_i, x)$ for a vertex, as exemplified by Table~\ref{tbl:U}. To enable compact storage, we only store budget values between $l$ and $s$. The search strategy in Algorithm~\ref{algo:policy} is shown in Figure~\ref{fig:policy}. 

We use a FIFO queue to maintain the vertices. First, we insert the destination $v_d$. 
We iteratively check each vertex in the queue until the queue is empty. For each iteration, we remove a vertex, say $v_i$, from the queue and perform two tasks---(1) compute $U(v_i, x)$ using Eq.~\ref{eqn_U} and the above two observations and (2) insert into the queue all incoming vertices that are connected to $v_i$ by an edge or a T-path and that are not yet computed in the heuristic table (see Figure~\ref{fig:policy}(a)).


\noindent
\textbf{Task (1): }
Based on the second observation, if $v_i$ is the destination vertex $v_d$, $U(v_i, x)=1$ for all budget value $x$. 
Otherwise, according to the first observation, we need to identify the two boundary budget values $l$ and $s$. 
Here, the $l$ value should be the smallest budget value that is no less than $v_i.\mathit{getMin}()$ from the binary heuristics. 
This is because $v_i.\mathit{getMin}()$ represents the least travel cost from $v_i$ to $v_d$, and for any budget value that is smaller than $v_i.\mathit{getMin}()$, it is impossible to reach the destination within the budget. 
Then, we start computing $U(v_i, x)$, where $x$ starts from $l$ and each time increases by $\delta$. 
Once we find that $U(v_i, x)$ is 1, we set the current $x$ to $s$ and we do not need to compute for budgets that are larger than $s$ because they must be 1 as well. 

Then, we consider vertex set $ON$ that includes $v_i$'s outgoing neighbor vertices (see Figure~\ref{fig:policy} (a)) to compute $U(v_i, x)$ using Eq.~\ref{eqn_U}: 
\vspace{-1em}
\begin{equation}
    \begin{aligned}
        U(v_i, x)&=\max_{z\in ON} H(v_i, z, x)\\
        &=\max_{z\in ON} \sum_{c=1}^x W(\langle v_i, z \rangle).\mathit{pdf}(c) \cdot U(z, x-c).
    \end{aligned}   
\end{equation}
The first term $W(\langle v_i, z \rangle).\mathit{pdf}(c)$ is available from graph $\mathcal{G}^{p}$, no matter $\langle v_i, z \rangle$ is an edge or a T-path. 
We use dynamic programming to compute the second term $U(z, x-c)$. As we start searching from the destination, vertex $z$'s row in the heuristic table often has already been computed, we may get $U(z, x-c)$ directly from the table. Otherwise, we recursively apply the same equation to compute $U(z, x-c)$.

\noindent
\textbf{Task (2): }
We perform the second task using $IN$ that includes $v_i$'s incoming neighbor vertices (see Figure~\ref{fig:policy} (a)). 
For each vertex $y_a\in IN$, we insert the vertex into the queue $Q$. In addition, we now can easily compute $H(y_a, v_i, x)=\sum_{c=1}^x W(\langle y_a, v_i \rangle).\mathit{pdf}(c) \cdot U(v_i, x-c)$ since the first term is available from graph $\mathcal{G}^{p}$ and since we have just computed $U(v_i, x)$ in the first task. 

Figure~\ref{fig:policy} (b), (c), and (d) illustrate the first three steps of computing the heuristic table for the graph shown in Figure~\ref{fig:pace_graph} for destination $v_d$. 
The first step starts from destination $v_d$ itself. Next, it explores $v_d$'s incoming neighbors $v_3$ and $v_4$.

The time complexity of Algorithm \ref{algo:policy} and \ref{algo:computeU} is $O(|\mathbb{V}|\cdot \eta^2\cdot |\mathit{out}|)$, where $\eta$ is the number of columns in the heuristic table and $|\mathit{out}|$ is the largest outdegree of a vertex. We apply dynamic programming to build the heuristics table using Eq.~\ref{eqn_ui}. The heuristics table has in total $\mathbb{V} \cdot \eta$ elements. To fill in an element in the table, we need to visit at most $|\mathit{out}| \cdot \eta$ other elements because we may visit up to $|\mathit{out}|$ rows (see $Z\in ON(v_i)$ in Eq.~\ref{eqn_ui}) and each row we may visit up to $\eta$ elements (see $\sum_{k=1}^x$ in Eq.~\ref{eqn_ui}). 

The space complexity is $O(|\mathbb{V}|^2\cdot \eta)$.  For each destination, we maintain a heuristic table taking $|\mathbb{V}| \cdot \eta$, therefore a total of $|\mathbb{V}| \cdot  |\mathbb{V}| \cdot \eta$.

\section{Virtual Path Based Routing}
\label{sec:vrouting}

Stochastic dominance based pruning is inapplicable in PACE because PACE captures cost dependencies and does not assume independence.
More specifically, the cost distribution of a path is computed using the T-path assembly operation $\diamond$ (c.f. Eq.\ref{eqn_pace}), not the convolution operation $\oplus$. 
To enable stochastic dominance based pruning, we propose to build V-paths, such that the cost distribution of a path can be computed by the convolution of cost distributions of V-paths, T-paths, and edges. 
\subsection{Building Virtual Paths}
\label{subsec:vpath}

We use a concrete example to illustrate the intuition of introducing V-paths. 
The example in Figure~\ref{fig:buildingvpath} includes T-paths $p_1$, $p_2$, and $p_3$. 
Given a path $P=\langle e_1, e_2, e_3, e_4, e_5 \rangle$, its joint distribution is computed by $D_J(P)=p_1 \diamond p_2 \diamond p_3 \diamond \langle e_5 \rangle$ according to Eq.~\ref{eqn_pace}. 
Then, we derive the cost distribution $D(P)$ from the joint distribution using a simple transformation as shown in Table~\ref{tbl:PACE_weight}.

As $p_3$ and $\langle e_5 \rangle$ do not overlap, we have 
\begin{equation}
  \label{eq:vpath}
D_J(P)=\frac{W_J(p_1)W_J(p_2)W_J(p_3)}{W(\langle e_2 \rangle)W_J(\langle e_3 \rangle)} W(e_5) = D_J(P') W(e_5),  
\end{equation}
where $P^\prime=\langle e_1, e_2, e_3, e_4 \rangle$ is a sub-path of $P$. %
This suggests that $P'$ and $e_5$ are independent. 
\begin{figure}[ht!]
    \centering
    \includegraphics[scale=0.42]{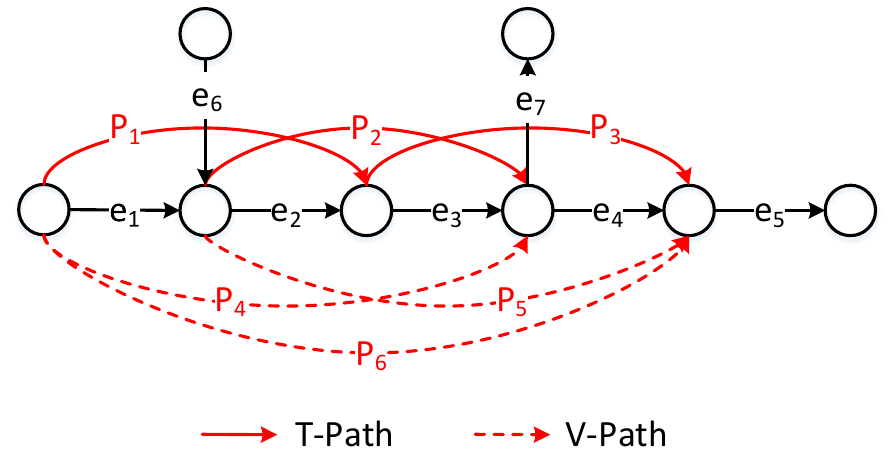}
    \vspace{-1em}
    \caption{Example of Building V-paths from T-paths}
    \label{fig:buildingvpath}
\end{figure}

Thus, we could apply the convolution operation to compute the cost distribution of $P$: $D(P)=D(P^\prime) \oplus W(e_5)$, where 
the cost distribution $D(P^\prime)$ can be derived from the joint distribution of $D_J(P^\prime)$,  i.e., the fraction part in Eq.~\ref{eq:vpath}, using T-paths $p_1$, $p_2$, and $p_3$.

We call $P^\prime$ a \emph{V-path}. 
First, a V-path must have fewer than $\tau$ trajectories because, otherwise, it should already have been a T-path. 
Second, the distribution of a V-path needs to be computed from the T-path assembly operation $\diamond$ (see Eq.~\ref{eqn_pace}) using distributions of multiple T-paths. 

This example suggests that if we pre-compute the distribution of V-path $P^\prime$, we are able to use only convolution $\oplus$ to compute the distribution of $P$.
This motivates us to systematically identify all V-paths and pre-compute their distributions.
The intuition is that we move the time consuming online computation of Eq.~\ref{eqn_pace} offline by pre-computing the distributions of all V-paths. As  Eq.~\ref{eqn_pace} is used both online and offline, there is no accuracy loss.
This enables us to use only convolution to compute the distributions of candidate paths during routing, thus making it possible to use stochastic dominance based pruning in the PACE model. 

The idea is to first combine overlapping T-paths into V-paths and then recursively combine overlapping V-paths into longer V-paths.  

\noindent
\textbf{Combining Two T-Paths: }
In the first iteration, we combine overlapping T-paths into V-paths. 
If two T-paths overlap and the path underlying the two T-paths does not have a corresponding T-path, we combine the two T-paths into a V-path. 
For example, consider T-paths $p_1$ and $p_2$ in Figure~\ref{fig:buildingvpath} that overlap on edge $e_2$. 
The path underlying $p_1$ and $p_2$ is $\langle e_1, e_2, e_3 \rangle$, and no T-path exists that also corresponds to $\langle e_1, e_2, e_3 \rangle$. Thus, we build a V-path $p_4$ with joint distribution $D_J(p_4)=p_1\diamond p_2$. Similarly, we build V-path $p_5$ with joint distribution $D_J(p_5)=p_2\diamond p_3$.  

\noindent
\textbf{Combining Two V-Paths: }
In the second iteration, we combine overlapping V-paths into longer V-paths. 
For example, in Figure~\ref{fig:buildingvpath}, since the two V-paths $p_4$ and $p_5$ overlap, we build V-path $p_6$. 
The joint distribution of the new, longer V-path is also computed based on the corresponding T-paths. 
Here, $D_J(p_6)=p_1 \diamond p_2 \diamond p_3$. 
When combining two V-paths, we do not need to check if a T-path exists that corresponds to the same underlying path. 
The existence of the two V-paths implies that there are fewer than $\tau$ trajectories. And since the V-paths are sub-paths of the combined path, the combined path also cannot have more than $\tau$ trajectories. 
Thus, no T-path can exist for the combined path. 

We keep combining overlapping V-paths to obtain longer V-paths until there are no more overlapping V-paths to combine. 
Next, we observe that we need all V-paths rather than only the longest ones. 
Table~\ref{tbl:convvpatnv} shows that all V-paths $p_4$, $p_5$, and $p_6$ in Figure~\ref{fig:buildingvpath} contribute to computing the distribution of some path. 
So although we have a longer V-path $p_6$, we still need the two short V-paths $p_4$ and $p_5$. 

\begin{table}[!htbp]
	\centering
	\begin{tabular}{|c|c|} \hline
		Path & Distribution  \\\hline	\hline
		$\langle e_1, e_2, e_3, e_7\rangle$	& $D(p_4)\oplus W(e_7)$ \\\hline	
		$\langle e_6, e_2, e_3, e_4\rangle$	& $W(e_6)\oplus D(p_5)$ \\\hline	
		$\langle e_1, e_2, e_3, e_4, e_5\rangle$	& $D(p_6)\oplus W(e_5)$ \\\hline
	\end{tabular}
        \vspace{0.5em}
	\caption{Computing Path Distributions Using Only Convolution  }\label{tbl:convvpatnv}		
 \vspace{-2.5em}
\end{table} 

Given graph
$\mathcal{G^{P}}=(\mathbb{V}, \mathbb{E}, \mathbb{P}, \mathbb{W})$,
the time complexity for building V-paths is $O(|\mathbb{P}|\cdot |\mathbb{V}|)$. 

Recall that V-paths are generated iteratively. Each iteration can at most generate $|\mathbb{P}|$ new V-paths and we can at most have $|\mathbb{V}|$ iterations. 
In the first iteration, overlapping T-paths are combined into V-paths. Thus, the number of generated V-paths must be smaller than the number of T-paths because the generated V-paths must have higher cardinality than the corresponding T-paths and they both represent the same path. In the next iteration, the same principle applies. Thus, each iteration can at most generate $O(|\mathbb{P}|)$ new V-paths.  

Next, we explain why we can get most have $|\mathbb{V}|$ iterations. At each iteration, the cardinality of V-paths are at least one more than the V-paths in the previous iteration. In a graph with $|\mathbb{V}|$ vertices, the longest simple path (i.e., without loops) in the graph can at most have $|\mathbb{V}|$ vertices. Otherwise, there must be at least one vertex that appears twice in the path making a loop and thus the path is not a simple path anymore. Thus, we only need to perform at most $|\mathbb{V}|$ iterations as we need to ensure that V-paths are simple paths.  

Space complexity is $O(|\mathbb{P}|\cdot|\mathbb{V}|\cdot |\mathit{St}| )$ where $|\mathit{St}|$ is the storage used for the longest V-path.

\noindent
\textbf{Updated PACE graph:} 
After generating all V-paths, we define an updated PACE graph  $\mathcal{G}^{p^+}=(\mathbb{V}, \mathbb{E}, \mathbb{P^+}, {W^+})$, where $\mathbb{P^+}$ is a union of T-paths, i.e., $\mathbb{P}\in\mathcal{G}^p$, and the newly generated V-paths, and where weight function $W^+$ takes as input an edge and a T-path or a V-path and returns the total cost distribution. With V-paths, there is no need to maintain the joint distributions for T-paths and V-paths. 

By introducing V-paths, Lemma~\ref{lm:vpath} shows that the cost distribution of any path $P$ can be computed using convolution only, thus offering a theoretical foundation for using stochastic dominance based pruning in PACE. 
Put differently, by introducing V-paths, Eq.~\ref{eqn_pace} is no longer needed, and we have turned PACE into EDGE, where stochastic dominance-based pruning guarantees correctness. 

\begin{lemma}
	\label{lm:vpath}
Given any path $P$ in the updated PACE graph $\mathcal{G}^{p^+}$, the distribution of $P$ is computed by using convolution of the weights of edges, T-paths, and V-paths maintained in $\mathcal{G}^{p^+}$. 
\end{lemma}

\begin{proof}
In the original PACE model, the distribution $D_J(P)$ of path $P$ is computed using Eq.~\ref{eqn_pace} on the coarsest T-path sequence $\mathit{CPS}(P)=(p_1, p_2, \ldots, p_m)$, where each $p_i$ is an edge or a T-path that maintains a joint distribution.
If adjacent $p_i$ and $p_{i+1}$, where $i\in[1,m-1]$, do not overlap, we split $\mathit{CPS}(P)$ such that $p_i$ and $p_{i+1}$ go to two different sub-sequences. 
Finally, $\mathit{CPS}(P)$ is split into multiple sub-sequences $(\mathit{cps}_1, \mathit{cps}_2, \ldots, \mathit{cps}_n)$, where each sub-sequence $\mathit{cps}_k = (p_j, \ldots, p_{j+x})$, with $k\in[1,n]$, $j\in[1,m]$, and $x\in[0,m-1]$, does not overlap any others. 
Here, a sub-sequence $\mathit{cps}_k$ can be an edge, a T-path, or multiple overlapping T-paths. When $x>0$, $\mathit{cps}_k$ includes multiple overlapping T-paths.
Since we already built V-paths for overlapping T-paths, there must be a V-path that corresponds to $\mathit{cps}_k$. 
Thus, $\mathit{cps}_k$ is an edge, a T-path, or a V-path.

Since there are no overlaps among different $\mathit{cps}_k$, Eq.~\ref{eqn_pace} yields $D_J(P)=\Pi_{k=1}^{n}W_J(\mathit{cps}_k)$, meaning that the edges, T-paths, and V-paths are independent of each other.  
Thus, the cost distribution $D(P)$, which is derived from $D_J(P)$ in the PACE model, is equivalent to $\oplus_{k=1}^{n}W^+(\mathit{cps}_k)$ that only involves convolution of items in the weight function $W^+(\cdot)$ in the updated PACE graph with V-paths. 
\end{proof}

\subsection{Routing Algorithm}

The V-paths introduced in the updated PACE graph $\mathcal{G}^{p^+}$ enable us to use stochastic dominance to prune non-competitive candidate paths as in the EDGE model. 
However, after adding the V-paths, the degrees of many vertices also increase. 
This indicates that when exploring an vertex, more candidate paths may be generated and examined. 
For example, consider the left most vertex in  Figure~\ref{fig:buildingvpath}. 
Before introducing the V-paths, its degree is 2 with edge $e_1$ and T-path $p_1$. 
After introducing the V-paths, its degree becomes 4 due to V-paths $p_4$ and $p_6$. 
Luckily, with the proposed search heuristics, we are able to choose the most promising candidate paths among many candidate paths. 
Thus, the increased degrees do not adversely affect the query efficiency. 

Algorithm~\ref{algo:stochastic} shows the final routing algorithm that uses the two proposed speedup techniques---search heuristics and V-paths. 

\begin{algorithm}[ht]
	\caption{V-PathBasedStochasticRouting}
	\label{algo:stochastic}
	\begin{small}
		\KwIn{
		PACE Graph $\mathcal{G}^{p^+}$, source $v_s$, destination $v_d$, departure time $t$, budget $B$;}
		\KwOut{Path $P^*$;}

		Priority queue: $Q \leftarrow \emptyset$; \\
		
		\For{each vertex $v_i$ that is connected from $v_s$ by an edge, a T-path or a V-path
		}
		{
			${P}_i.\mathit{path} \leftarrow \mathit{tracePath}(v_s,v_i)$;\\
			$D({P}_i) \leftarrow W^+(v_s,v_i)$;\\
			\If{$D({P}_i).\min + v_i.\mathit{getMin}() \leqslant B$}
			{
				${P}_i.\mathit{maxProb} \leftarrow \mathit{maxProb}({P}_i, B)$; \space\space\space\space
				$/\ast$ \texttt{Eq.}~\ref{eqn_est} $\ast/$\\
				$Q.\mathit{add}(P_i)$;\\
			}
		}		
		\While{$Q \neq \emptyset$}
		{
			$\hat{P} \leftarrow Q.peek()$;\\
			$v \leftarrow$ last vertex in $\hat{P}$;\\
			\If{$v=v_d$}
			{
				$P^* \leftarrow \hat{P}$; \\
				\Break; 
			}
			\For{each vertex $u$ that is connected from $v$ by an edge, a T-path or a V-path abd $u$ has not appeared in ${P}.\mathit{path}$}
			{
				$c \leftarrow D({P}).\min + W^+(v,u).\min$;\\
				\If{$c + u.\mathit{getMin}() \leqslant B$}
				{
					${P}'.\mathit{path} \leftarrow {P}.\mathit{path} + \mathit{tracePath}(v, u)$; \\
					$D({P}') \leftarrow D({P}) \oplus W(v, u)$; $/\ast$ \texttt{Use convolution to compute a distribution.} $\ast/$
					\\
					${P}'.\mathit{maxProb} \leftarrow \mathit{maxProb}({P}', B)$; \space\space\space\space
					$/\ast$ \texttt{Eq.}~\ref{eqn_est} $\ast/$\\		
					\If{$\mathit{Prune}({P}', Q)$}
					{
						$Q.add({P}')$;
					}			
				}				
			}
		}
		\Return $P^*$;%
	\end{small}
\end{algorithm}

Given a PACE graph $\mathcal{G}^{p^+}$, for each candidate path ${P}_i$ that departures at time $t$ and connects source $v_s$ to an intermediate vertex $v_i$, 
we maintain two attributes.  The first, ${P}_i.\mathit{path}$, refers to the underlying path, i.e., a sequence of edges. 
For example, the $\mathit{path}$ attribute of $\langle p_4, e_4\rangle$ in Figure~\ref{fig:buildingvpath} is $\langle e_1, e_2, e_3, e_4\rangle$ since $p_4$ is a V-path, which is a sequence of edges in the road network. 
We maintain this attribute to avoid cycles in candidate paths. 

The second attribute ${P}_i.\mathit{maxProb}$ is the maximum probability of reaching $v_d$ within cost budget $B$ when following ${P}_i$ to $v_i$ and then proceeding to $v_d$. 
This probability can be computed using either the proposed binary heuristic or the budget-specific heuristic. 

We use a priority queue to maintain all candidate paths, where the priority is according to the $\mathit{maxProb}$ attribute of each candidate path. 
In each iteration, we explore the path $P$ with the largest $\mathit{maxProb}$ attribute, as it is the most promising candidate path (line 10). 
If this path $P$ has already achieved to the destination $v_d$, it is the path with the largest probability of arriving within budget $B$. 
This is because all the other candidates in the priority queue cannot have a larger probability as the heuristics are admissible (lines 12--13).

If path $P$ has not reached the destination, we extend the path with an adjacent edge, T-path, or V-path to get a new candidate path $P'$. 
The distribution of $P'$ is computed using convolution (line 18). 
Finally, we check stochastic dominance between $P'$ and each path $\hat{P}$ in $Q$ that also reaches $u$. 
If $P'$ stochastically dominates $\hat{P}$, we can safely prune $\hat{P}$ from $Q$. If $\hat{P}$ stochastically dominates $P'$, we can safely prune $P'$, i.e., not add $P'$ to $Q$. 
If eventually, no path in $Q$ stochastically dominates $P'$, we add $P'$ to $Q$ as a new candidate path (lines 20--21).

\section{Empirical Study}
\label{sec:exp}

\subsection{Experimental Setup}
\label{ssec:expsetup}

\noindent
{\bf Road Networks and GPS Trajectories:} 
We use two pairs of a road network and an associated GPS dataset: $N_1$ and $D_1$ covering Aalborg, Denmark, and $N_2$ and $D_2$ covering Xi'an, China, where $N_1$ and $N_2$ are obtained from OpenStreetMap. The sampling rates of the GPS records in $D_1$ and $D_2$ are 1~Hz and 0.2~Hz, respectively. We use travel times extracted from the GPS records as travel costs, and we utilize an existing tool~\cite{DBLP:conf/gis/NewsonK09} to map match $D_1$ and $D_2$ to $N_1$ and $N_2$. The mapped trajectories cover 23\% and 4\% of the edges in $N_1$ and $N_2$. The covered edges, typically main roads, show high uncertainty. The uncovered edges are typically small roads with low uncertainty. For these, we use speed limits to derive deterministic travel times. After detecting and filtering the abnormal data~\cite{DBLP:conf/icde/KieuYGCZSJ22,DBLP:conf/icde/KieuYGJZHZ22,davidpvldb}, statistics of the Aalborg and Xi'an data are shown in Table~\ref{table:data}.
\begin{table}[!htp]

\vspace{-0.5em}
\renewcommand{\arraystretch}{1.1}
\centering
\setlength{\tabcolsep}{1mm}{\begin{tabular}{c c c}
\hline
&Aalborg& Xi'an\\
\hline
Number of vertices&32,226& 117,415\\
Number of edges&78,348& 236,733\\
AVG vertex degree&2.43& 2.02\\
AVG edge length (m)&172.85& 58.36\\
\hline
Number of traj.&553,904& 363,308\\
AVG number of vertices per traj.&5.41& 27.05\\
\hline
\end{tabular}}
\caption{Data Statistics}
\vspace{-3em}
\label{table:data}
\end{table}

We conduct the empirical analysis using five-fold cross validation. Specifically, the trajectory dataset is partitioned evenly into five disjoint groups. Each time, four groups are used for training, and the remaining group is used as the testing set. This process is repeated five times so that each group is used as the testing set.

\noindent 
{\bf Stochastic Routing Queries:} 
A stochastic routing query takes three parameters: a source, a destination, and a travel time budget. 
We select source-destination pairs from the testing set to obtain meaningful source-destination pairs. %
We categorize the pairs into groups based on their Euclidean distances (km): (0, 5], (5, 10], (10, 25], and (25, 35]. 
We ensure that each category has at least 90 pairs. 
Note that 35 km is a long distance in a city. 
Following existing studies~\cite{DBLP:journals/pvldb/PedersenYJ20}, we focus on intra-city routing because cities often have uncertain traffic and have many alternative paths between a source and a destination. For inter-city travel, there is often limited choices, e.g., using vs. not using highways. We leave the support for country-level, inter-city routing as future work.

Next, we generate meaningful travel time budgets for the source-destination pairs. This is important because too small budgets result in many paths having probability 0, while too large budgets result in all paths having probability 1.  
For each source-destination pair, we run Dijkstra's algorithm based on a road network with expected travel times as edge weights. 
This yields a path with the least expected travel time $\bar{t}$. 
Then, we set 5 time budgets that correspond to 50\%, 75\%, 100\%, 125\%, and 150\% of $\bar{t}$, respectively, meaning that the time required to travel between the source-destination pair will not stay beyond the range of [50\% $\bar{t}$, 150\% $\bar{t}$]. This enables us to evaluate a range of meaningful time budgets. 

\noindent 
{\bf Routing Algorithms:} 
We consider a baseline routing algorithm in the PACE model~\cite{DBLP:journals/vldb/YangDGJH18} called \emph{T-None} (cf. Section~\ref{subsec:probdef}) that does not use any heuristics to estimate the cost from intermediate vertices to the destination and that does not use any V-paths.
Next, we consider routing algorithms using only T-paths with different heuristics:
\begin{itemize}[leftmargin=3.5mm]
    \item \emph{T-B-EU}: Binary heuristic using Euclidean distance divided by the maximum speed limit in the road network to derive $v_i.\mathit{getMin}()$;
    \item \emph{T-B-E}: Binary heuristic using shortest path trees derived from only edges. 
    \item \emph{T-B-P}: Binary heuristic with shortest path trees derived from both edges and T-paths.
    \item \emph{T-BS-$\delta$}: Budget-specific heuristics using a specific $\delta$. 
\end{itemize}

Finally, we consider V-path based routing algorithms: \emph{V-None} uses no heuristics, and \emph{V-B-P} and \emph{V-BS-$\delta$} use the above heuristics.  

\begin{figure*}[!h]
	\centering
	\subfigure[Numbers of T-Paths]{
	\centering
			\includegraphics[height=0.14\linewidth]{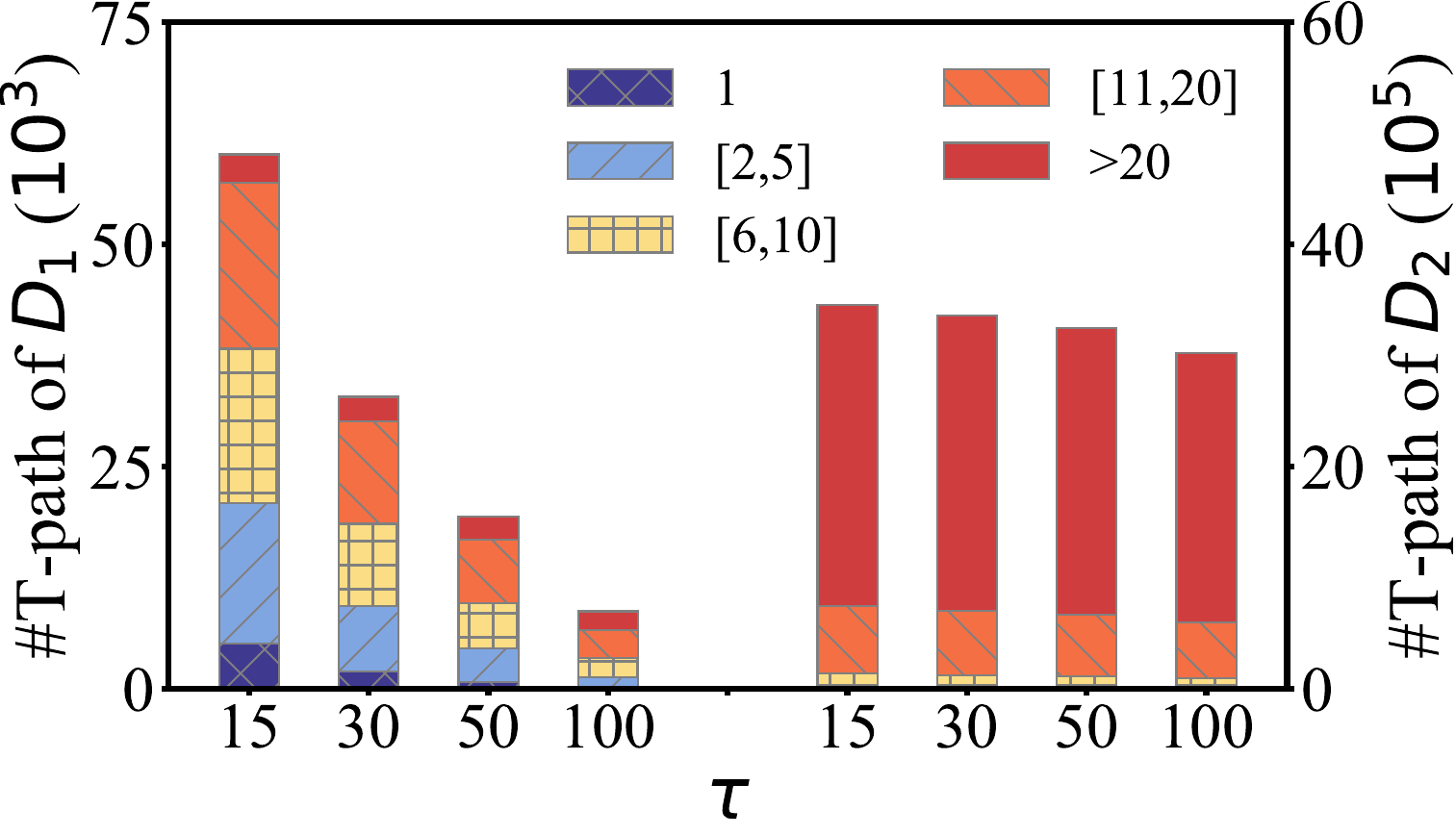}
	}
	\subfigure[Accuracy]{
	\centering
			\includegraphics[height=0.148\linewidth]{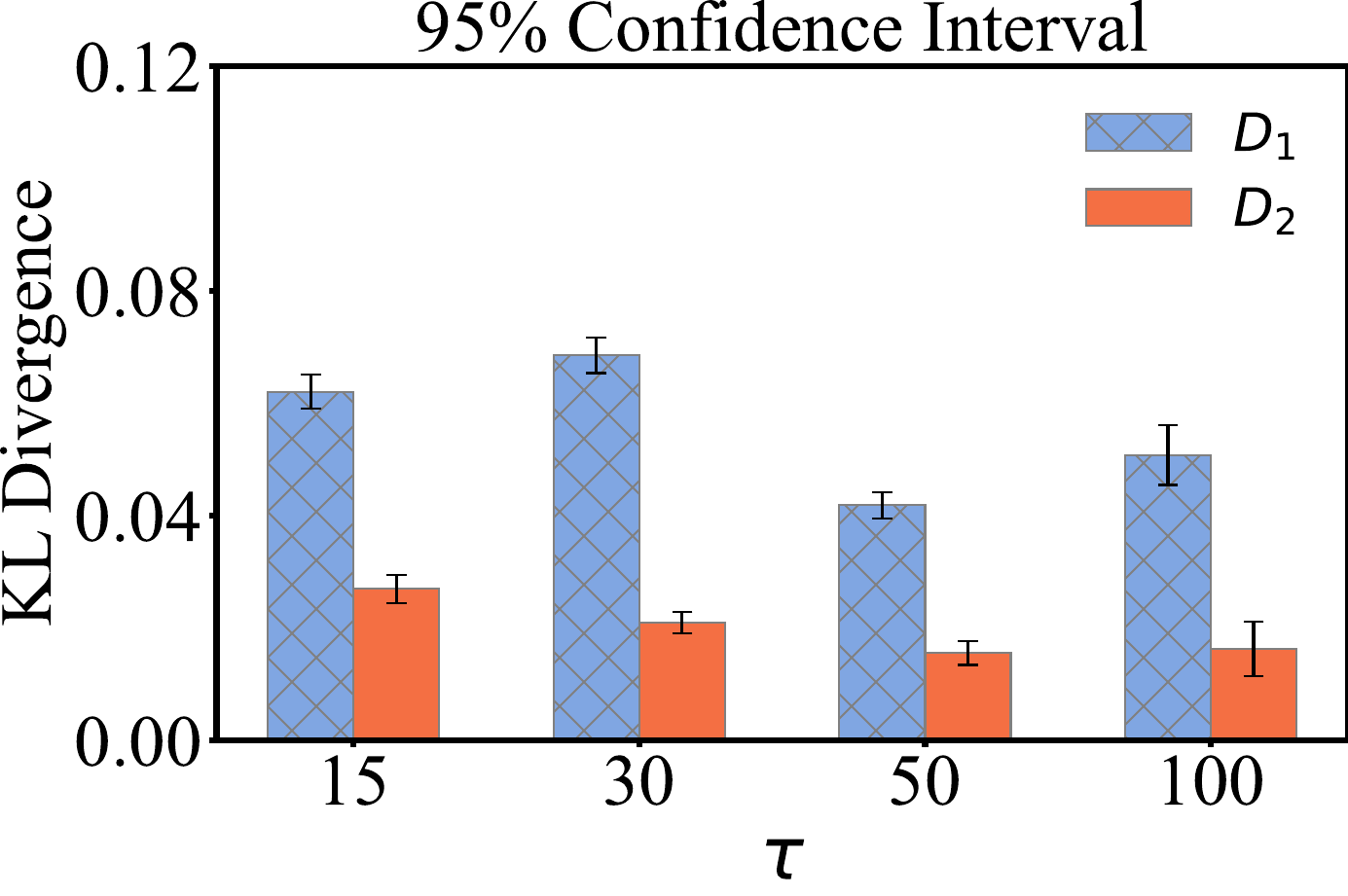}
	}
	\subfigure[Number of V-paths]{
	\centering
			\includegraphics[height=0.14\linewidth]{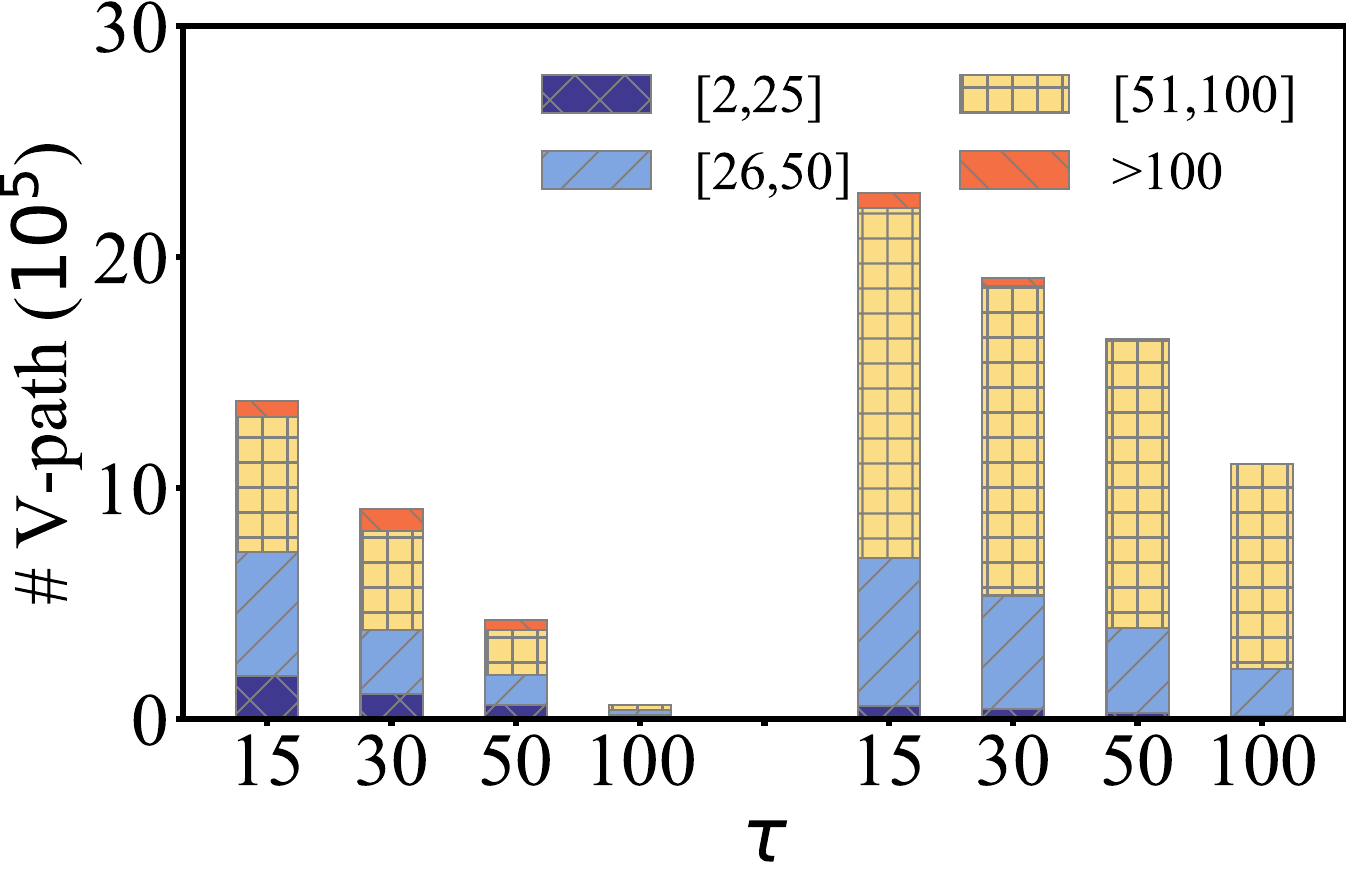}
	}
	\subfigure[Runtime and Out-Degrees]{
	\centering
			\includegraphics[height=0.14\linewidth]{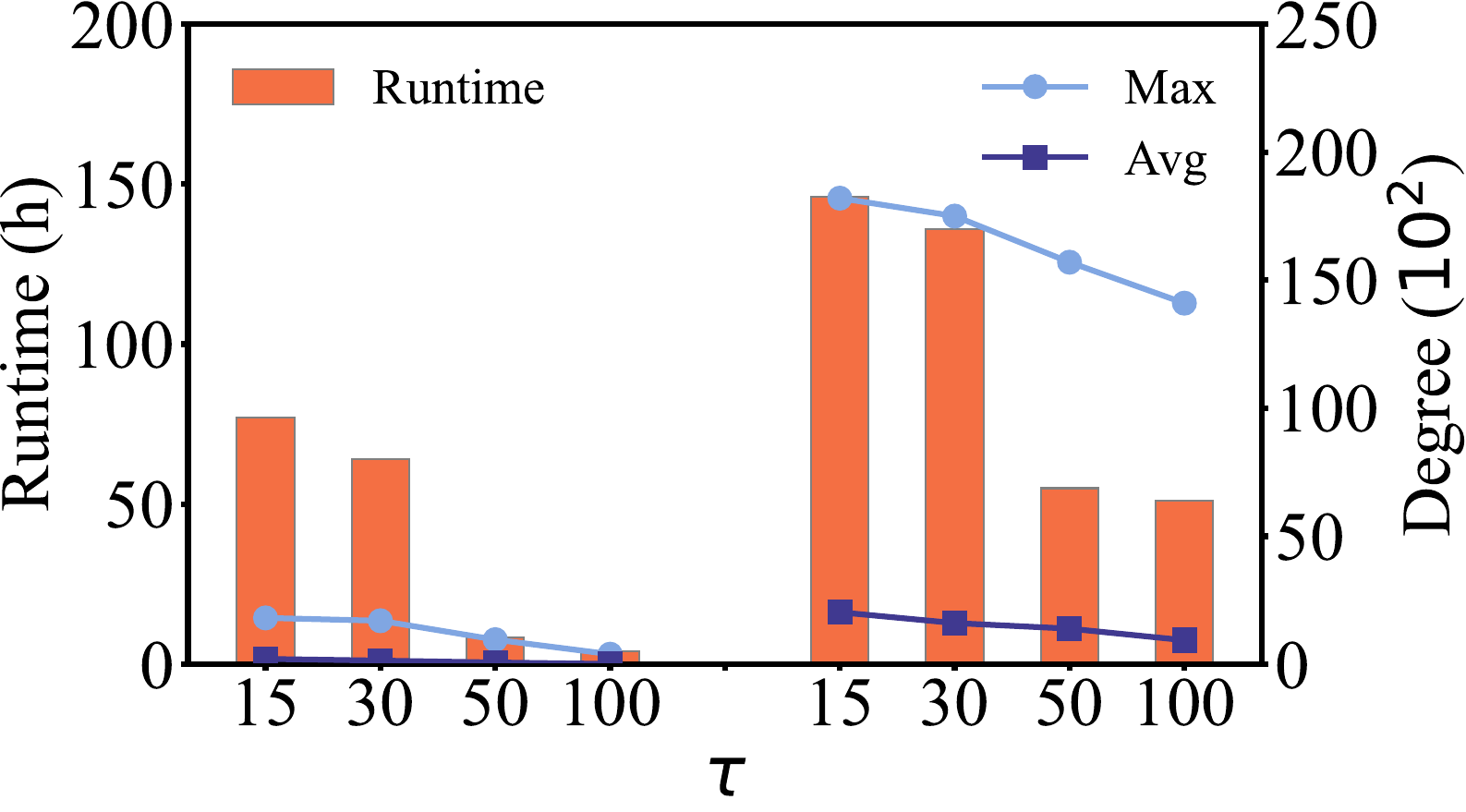}
	}
	\vspace{-1.5em}
	\caption{Effect of Varying $\tau$: In (a), (c), and (d), the First Four and last Four $\tau$ Values Correspond to $D_1$ and $D_2$, Respectively}
	\vspace{-10pt}
	\label{fig:exp_tau}
\end{figure*}

\noindent
{\bf Parameters:} 
We vary parameters $\tau$ among 15, 30, \textbf{50}, and 100, and $\delta$ among 30, \textbf{60}, 120, and 240, where $\tau$ is the trajectory threshold used when generating %
T-paths (cf. Section~\ref{subsec:pace}) and $\delta$ is the finest budget value  used in the budget-specific heuristics (cf. Section~\ref{subsec:policy}). Default values are shown in bold. %
We conduct sensitivity analyses on the two parameters to identify the most appropriate values. %

\noindent 
{\bf  Evaluation Metrics and Settings:} 
We evaluate the runtime and space overhead needed for maintaining the different heuristics. 

\noindent
\emph{Comparison with the EDGE model: }
We do not compare the paths returned by the PACE models vs. the EDGE model, since a previous study describes their differences and shows the benefits of the PACE model over the EDGE model~\cite{DBLP:journals/vldb/YangDGJH18}. 
Thus, we focus on evaluating the routing algorithm with different heuristics in the PACE model. 

\noindent
\emph{Time-dependency: }
For each network, we build two PACE models using the trajectories from peak (7:00-8:30 and 16:00-17:30) vs. off-peak (others) hours.

\noindent
{\bf Implementation Details:} All algorithms are implemented in Python 3.7.3. All experiments are conducted on a server with a 64-core AMD Opteron 2.24 GHZ CPU and 528 GB main memory under Ubuntu 16.%

\vspace{-0.5em}
\subsection{Instantiating T-Paths and V-Paths}
\label{subsec:vpath_exp}
\noindent
{\bf Instantiating T-Paths:} If a path is traversed by more than $\tau$ trajectories, we instantiate a T-path for it. %
Figure~\ref{fig:exp_tau}(a) shows the numbers of T-Paths on $D_1$ and $D_2$, respectively, when varying $\tau$. %
We group T-paths according to their cardinalities, i.e., the numbers of edges they cover. When the cardinality is 1, a T-path is just one edge. %
Intuitively, a larger $\tau$ requires that more trajectories occurred on the T-paths. Thus a larger $\tau$ yields fewer T-paths. 
However, large $\tau$ also yields more accurate travel time distributions for the T-paths.  

Next, we evaluate the accuracy when estimating the path cost distributions using the T-paths with different $\tau$ values. 
To do so, we use the testing set (not used for instantiating T-paths), where a travel time distribution, i.e., the ground truth distribution, is instantiated for each path from the trajectories that traversed on it. 

\begin{figure}[!t]
\centering
  \centering
	\subfigure[$D_1$]{
	\centering
			\includegraphics[width=0.45\linewidth]{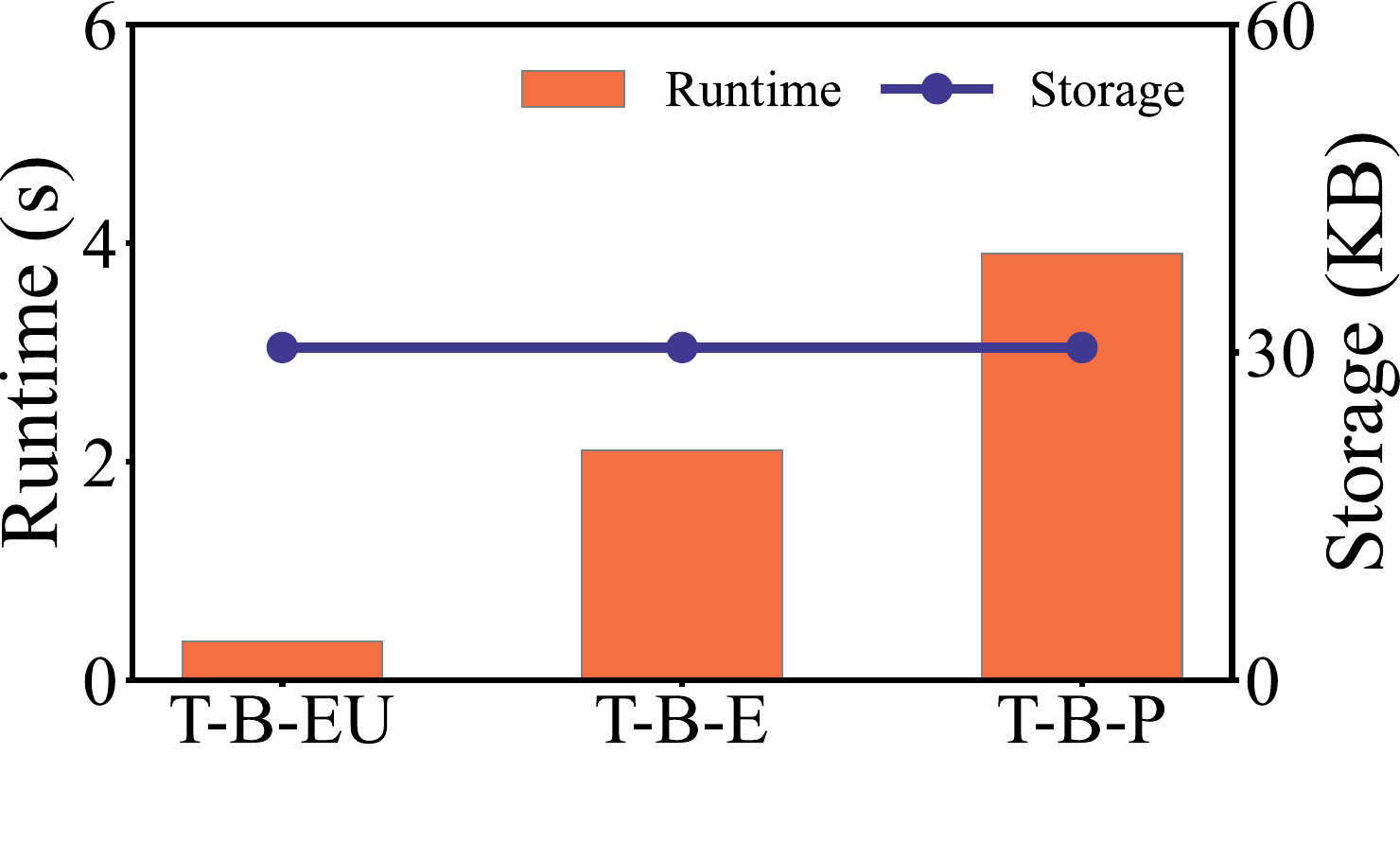}
	}
	\subfigure[$D_2$]{
	\centering
			\includegraphics[width=0.45\linewidth]{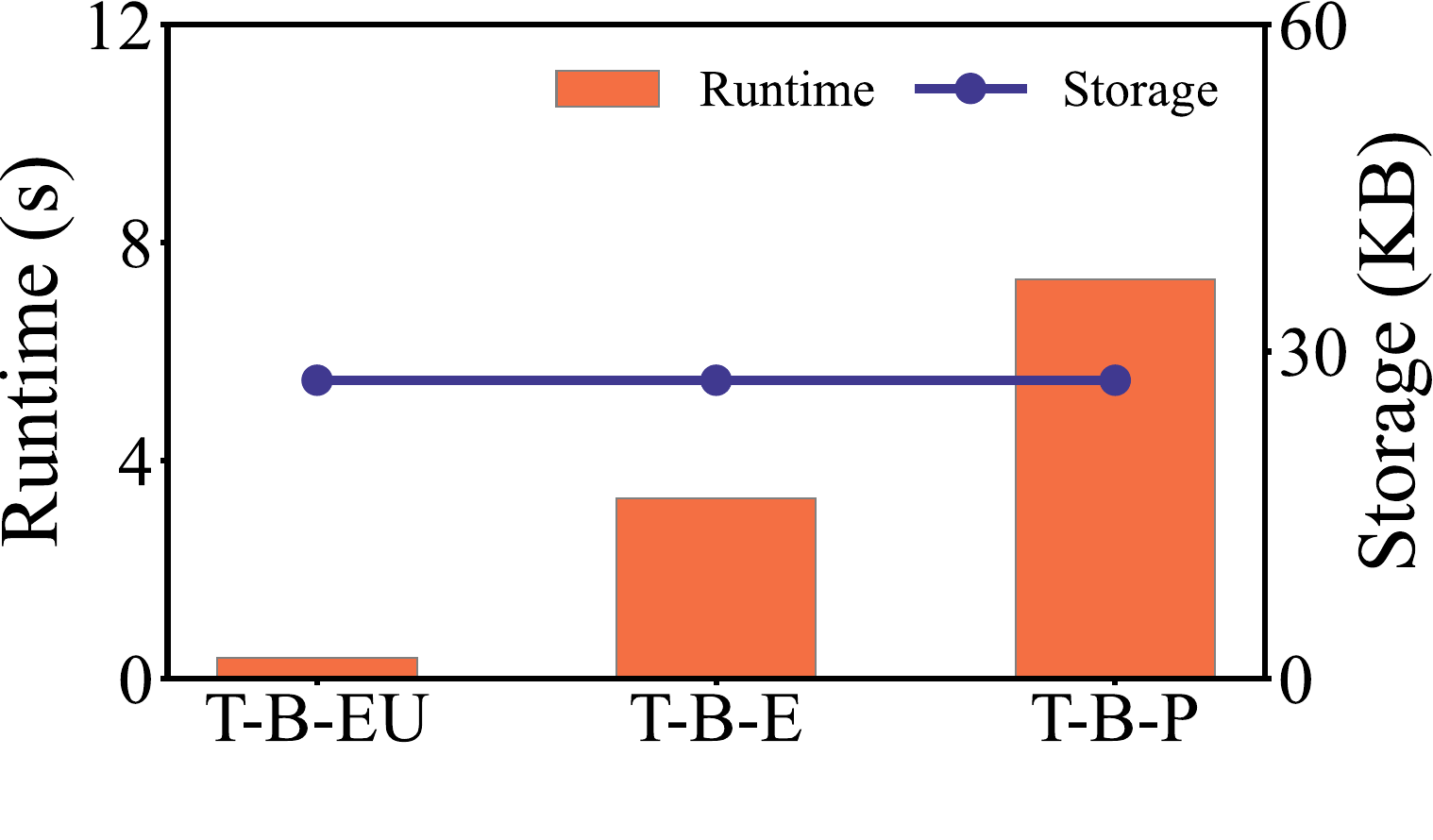}
	}
	\vspace{-1.5em}
    \caption{Building Offline Binary Heuristics}
	\label{fig:Binary-preprocessing}
    \vspace{-1.5em}
\end{figure}

\begin{figure}[!t]
\centering
  \centering
	\subfigure[$D_1$]{
	\centering
			\includegraphics[width=0.45\linewidth]{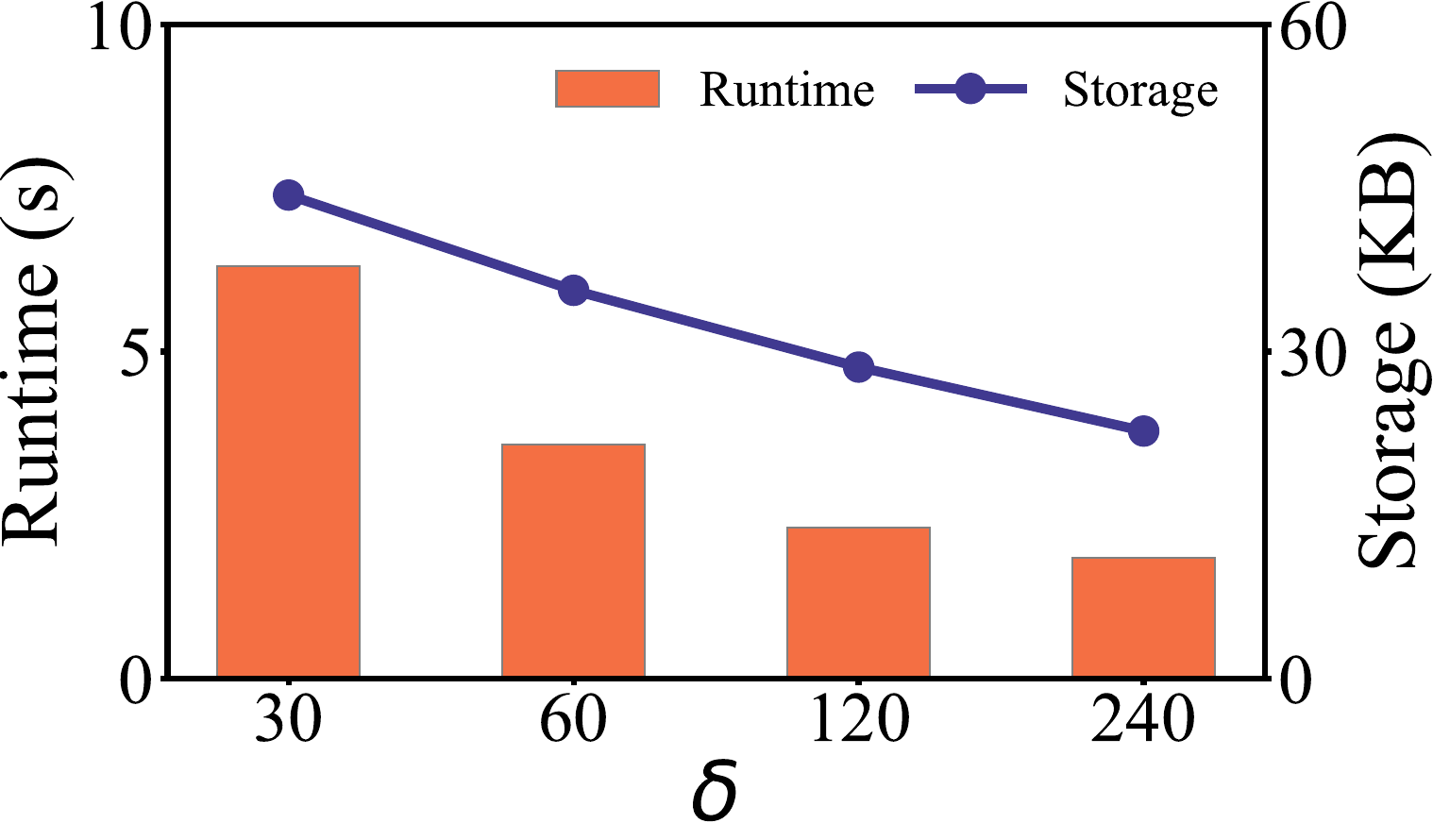}
	}
	\subfigure[$D_2$]{
	\centering
			\includegraphics[width=0.45\linewidth]{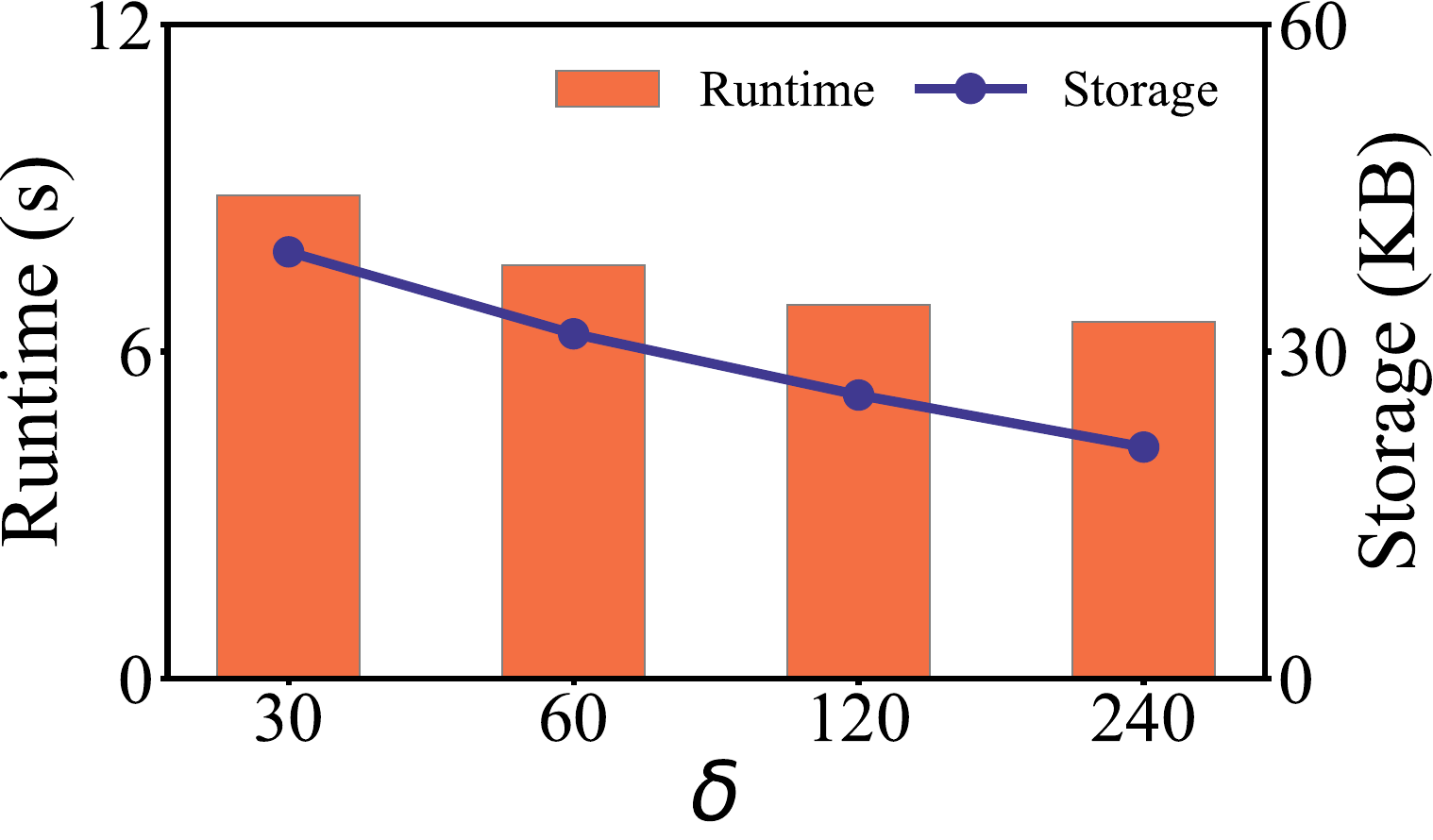}
	}
	\vspace{-1.5em}
    \caption{Building Offline Budget-Specific Heuristics}
	\label{fig:Budget-Specific-preprocessing}
    \vspace{-1.5em}
\end{figure}
\noindent
Then, we use the T-path weights obtained in the training set to estimate the distribution of each path in the testing set using Eq.~\ref{eqn_pace}.
We use KL-divergence to quantify the distances between the ground truth and the estimated distributions. 
A smaller KL-divergence value indicates higher accuracy, meaning that the estimated distributions are closer to the ground truth. 

Based on the five-fold cross validation, we report the 95\% confidence interval for the KL-divergence values for different $\tau$ values in Figure~\ref{fig:exp_tau}(b). %
When $\tau$ increases from 15 to 50, the KL-divergence values decrease, suggesting that a large $\tau$ indeed yields T-paths with more accurate cost distributions. %
However, when $\tau=100$, the KL-divergence values start increasing, and the 95\% confidence interval grows. This is because too few T-paths are instantiated, which in turn affects the accuracy of cost distributions. 
Thus, many cost dependencies that were kept with settings like $\tau=50$ are now lost, which reduces accuracy.  
We choose $\tau=50$ as the default value because it provides the most accurate results with substantial amounts of T-paths. 

\noindent
{\bf Instantiating V-Paths:} 
The numbers of V-paths for varying $\tau$ values are shown in Figure~\ref{fig:exp_tau}(c).
Smaller $\tau$ values result in more T-paths, which also lead to more V-paths. The cardinalities of V-paths are much higher than those of T-paths because V-paths are obtained by merging T-paths. %
Next, we report the average and maximum out-degrees of vertices after introducing the V-paths, in Figure~\ref{fig:exp_tau}(d).  
The out-degrees are often large because a vertex can now be connected to edges, T-paths, and V-paths. 
We show that the increased out-degrees do not adversely affect the routing efficiency in Section~\ref{ssec:vpathrouting}, due to the proposed search heuristics. 
Figure~\ref{fig:exp_tau}(d) also reports the runtime for generating V-paths. 
This procedure is conducted offline. When using the default of $\tau=50$, it takes around 8.5 and 55 hours for $D_1$ and $D_2$, 
which is acceptable. 
Thus, we trade affordable pre-computation time for interactive response times, i.e., at sub 0.1 second level, to be seen next.

\vspace{-1em}
\subsection{Search Heuristics with Only T-Paths}
\label{subsec:heurexp}

We study both search heuristics when using only T-paths. 

\noindent
{\bf Binary Heuristics:} We consider three variations of binary heuristics---T-B-EU, T-B-E, and T-B-P, as described in Section~\ref{ssec:expsetup}. 
We first study the pre-processing step. Specifically, we study the runtime of computing the $v.\mathit{getMin}()$ function and the storage needed for the results at peak hours, as shown in Figure~\ref{fig:Binary-preprocessing}. %

Recall that the binary heuristics are destination-specific. 
To support routing queries with arbitrary destinations, we need to maintain a $v.\mathit{getMin}()$ function for each destination vertex.  
Figure~\ref{fig:Binary-preprocessing} shows the average runtime of $v.\mathit{getMin}()$ and the storage needed for one destination. 
T-B-EU is the fastest, as it computes the Euclidean distance between an intermediate vertex and the destination and divides by the maximum speed limit. 
T-B-E and T-B-P take longer time as they both need to generate shortest path trees. T-B-P takes the longest time as it also involves dominance checking when taking into account the accurate costs maintained in T-paths. 
Note that T-B-P takes less than 4.0 seconds. 
We argue that this is reasonable because it is an offline computation and the $v.\mathit{getMin}()$ function can be computed in parallel for different destinations. %

\begin{figure*}[!htp]
\centering
  \centering
  \vspace{-1em}
	\subfigure[By distances, $D_1$]{
	\centering
			\includegraphics[width=0.23\linewidth]{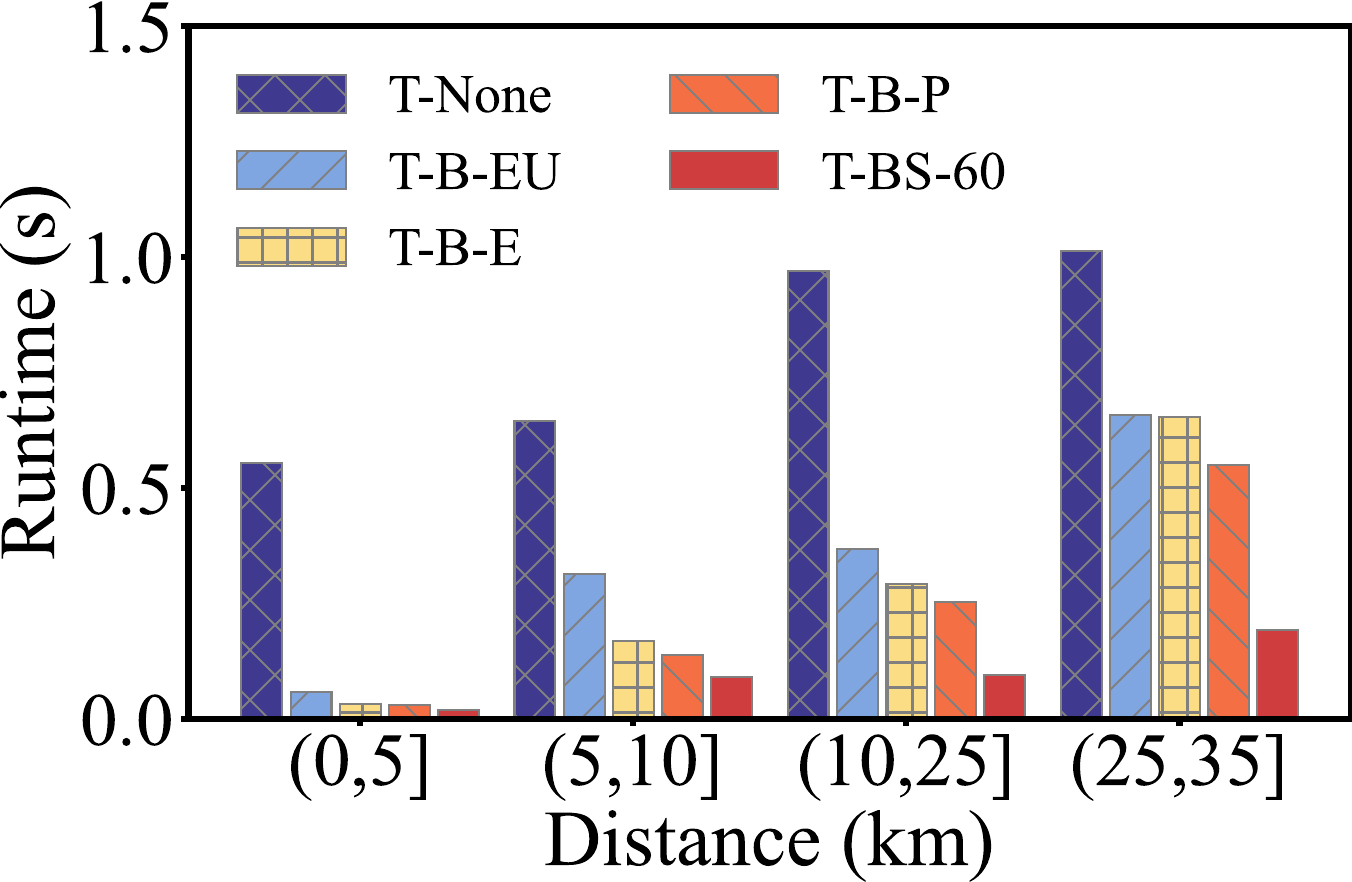}
	}
	\subfigure[By distances, $D_2$]{
	\centering
			\includegraphics[width=0.23\linewidth]{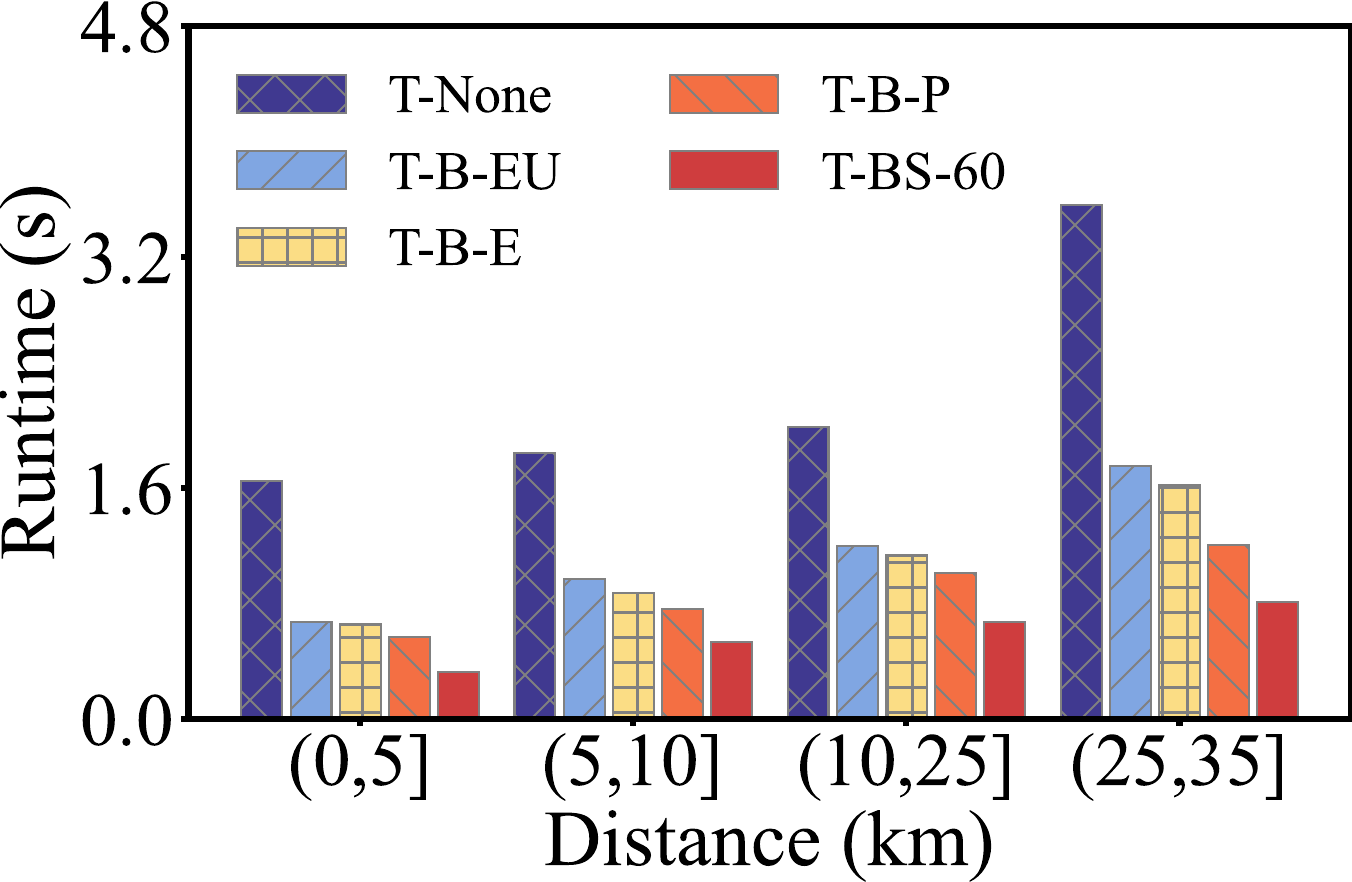}
	}
    \subfigure[By budget values, $D_1$]{
	\centering
			\includegraphics[width=0.23\linewidth]{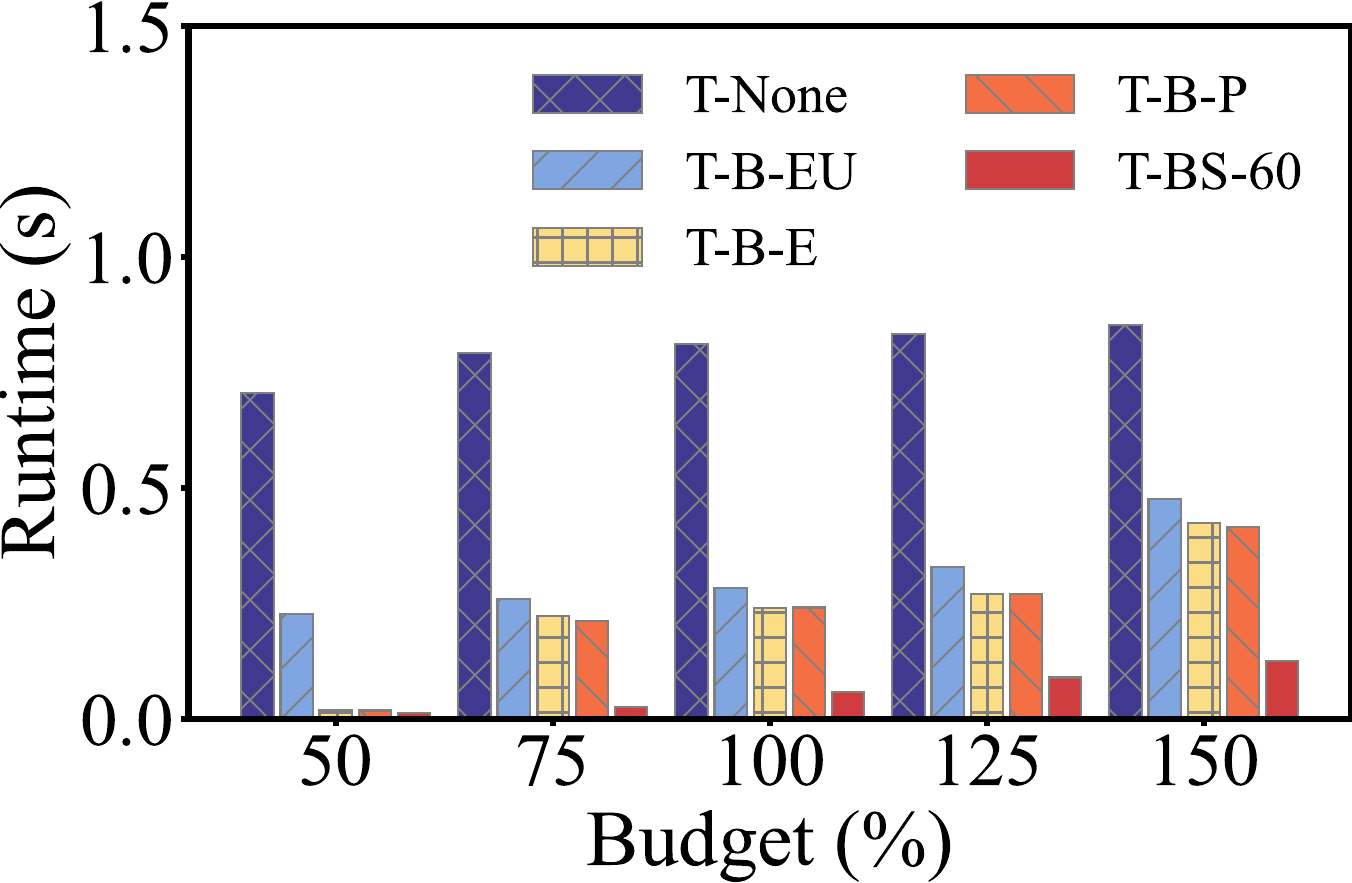}
	}
	\subfigure[By budget values, $D_2$]{
	\centering
			\includegraphics[width=0.23\linewidth]{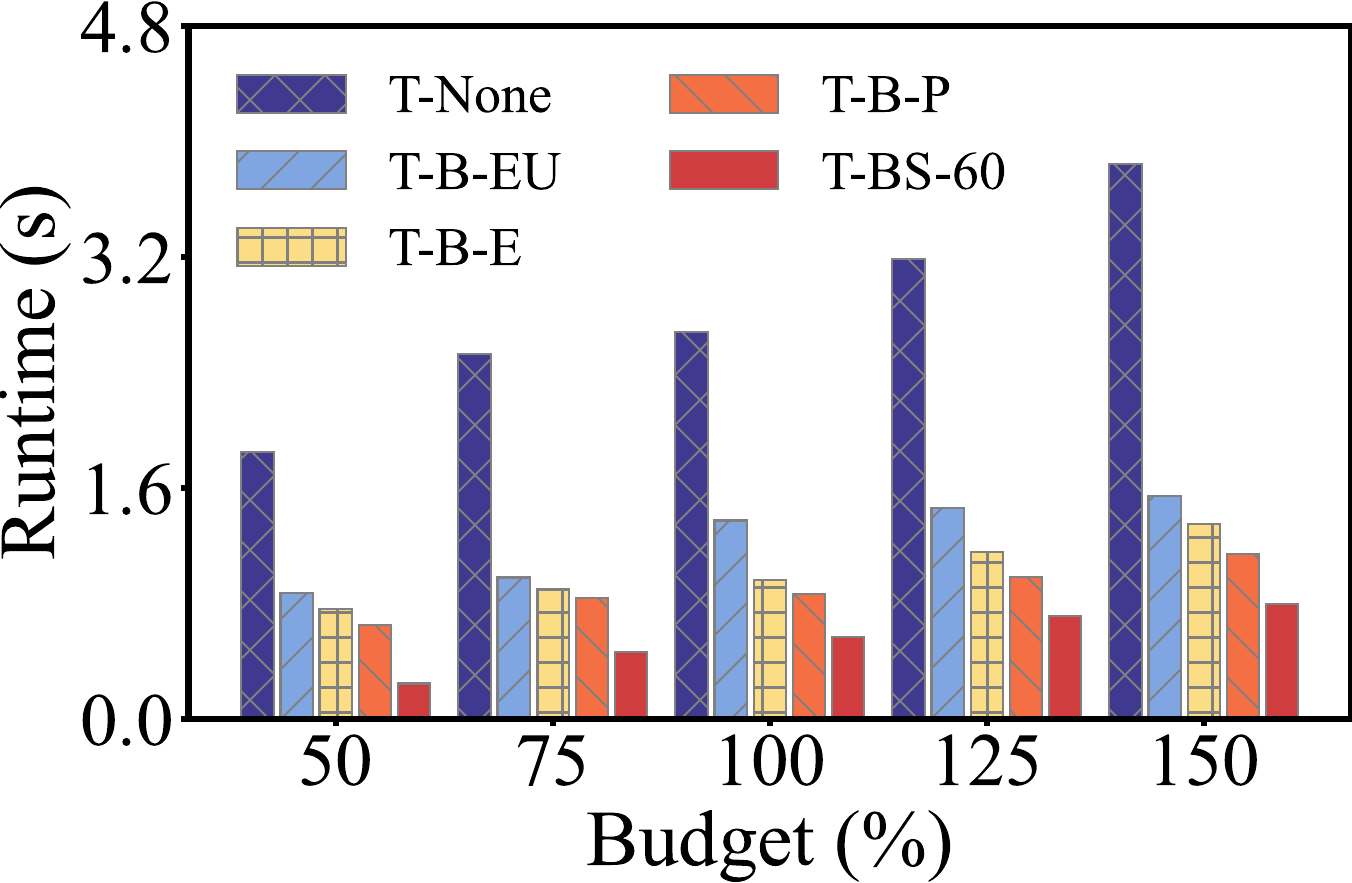}
	}
	\vspace{-1.5em}
	\caption{Stochastic Routing with Binary Heuristics at Peak Hours}
	\vspace{-10pt}
	\label{fig:binaryRouting-peak}
\end{figure*}
\begin{figure*}[!htp]
\centering
  \centering
	\subfigure[By distances, $D_1$]{
	\centering
			\includegraphics[width=0.23\linewidth]{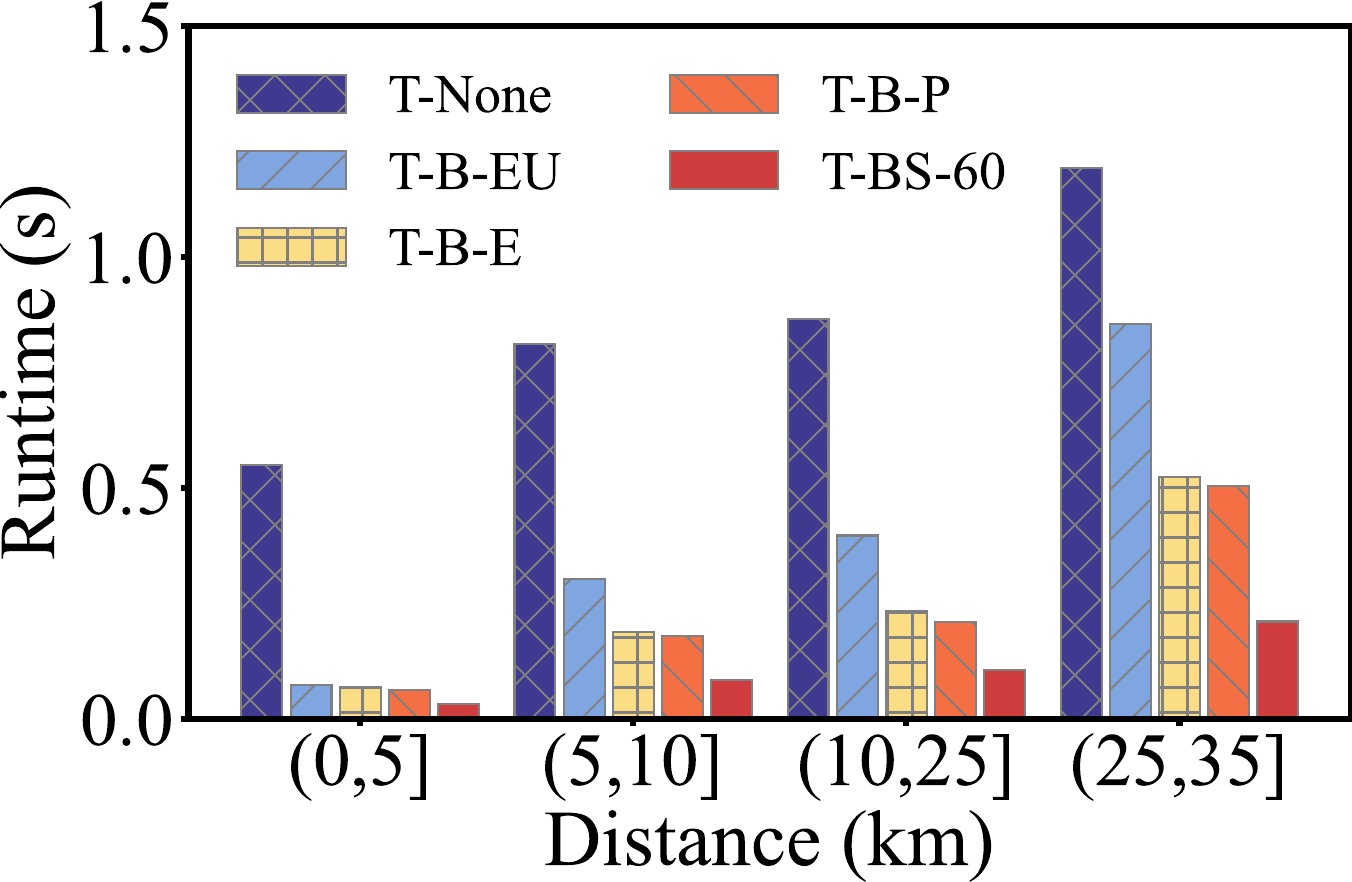}
	}
	\subfigure[By distances, $D_2$]{
	\centering
			\includegraphics[width=0.23\linewidth]{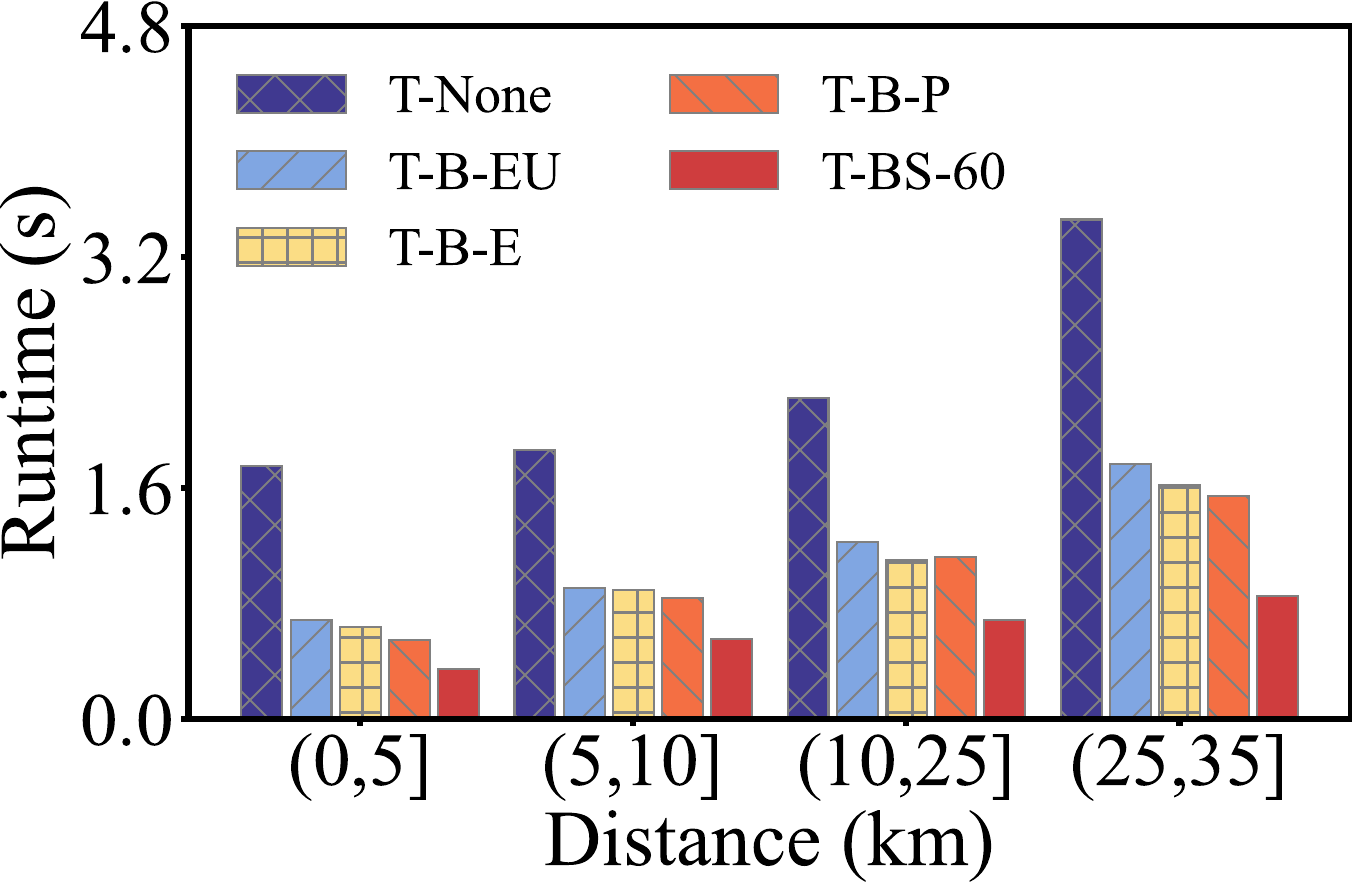}
	}
    \subfigure[By budget values, $D_1$]{
	\centering
			\includegraphics[width=0.23\linewidth]{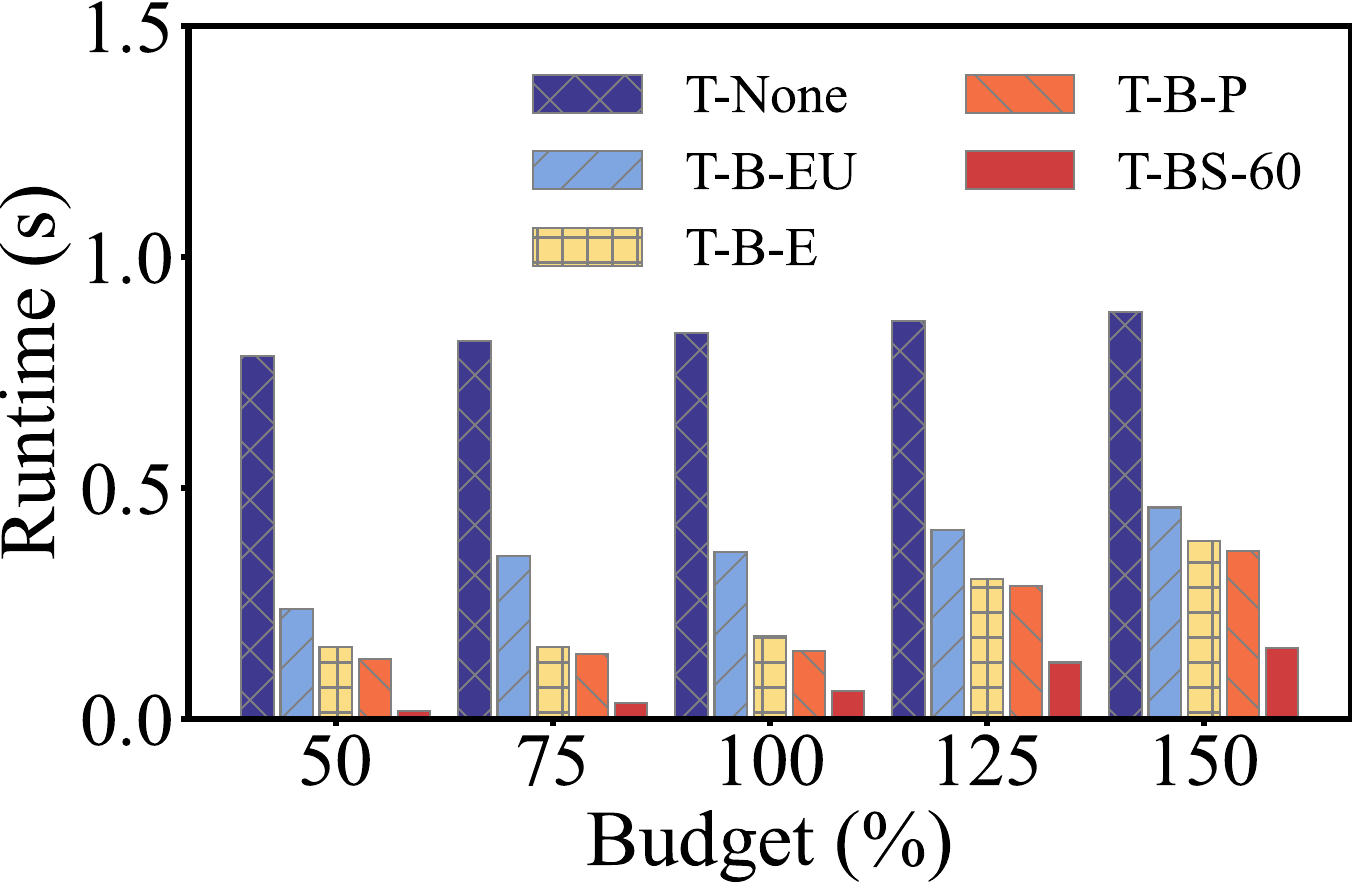}
	}
	\subfigure[By budget values, $D_2$]{
	\centering	\includegraphics[width=0.23\linewidth]{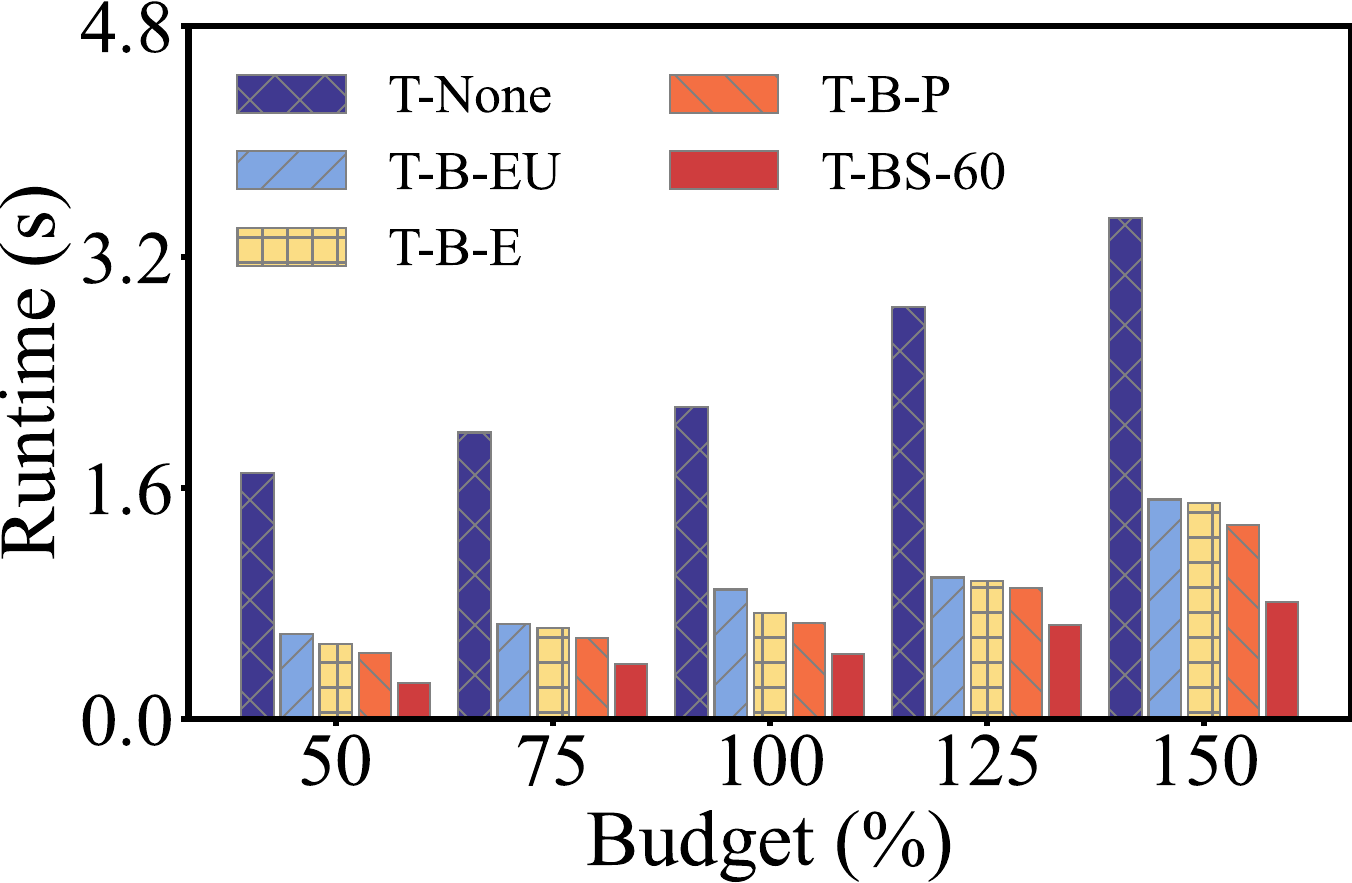}
	}
	\vspace{-1.5em}
	\caption{Stochastic Routing with Binary Heuristics at Off-Peak Hours}
	\vspace{-10pt}
	\label{fig:binaryRouting-offpeak}
\end{figure*}

\begin{figure*}[!htp]
\centering
  \centering
	\subfigure[By distances, $D_1$]{
	\centering
			\includegraphics[width=0.23\linewidth]{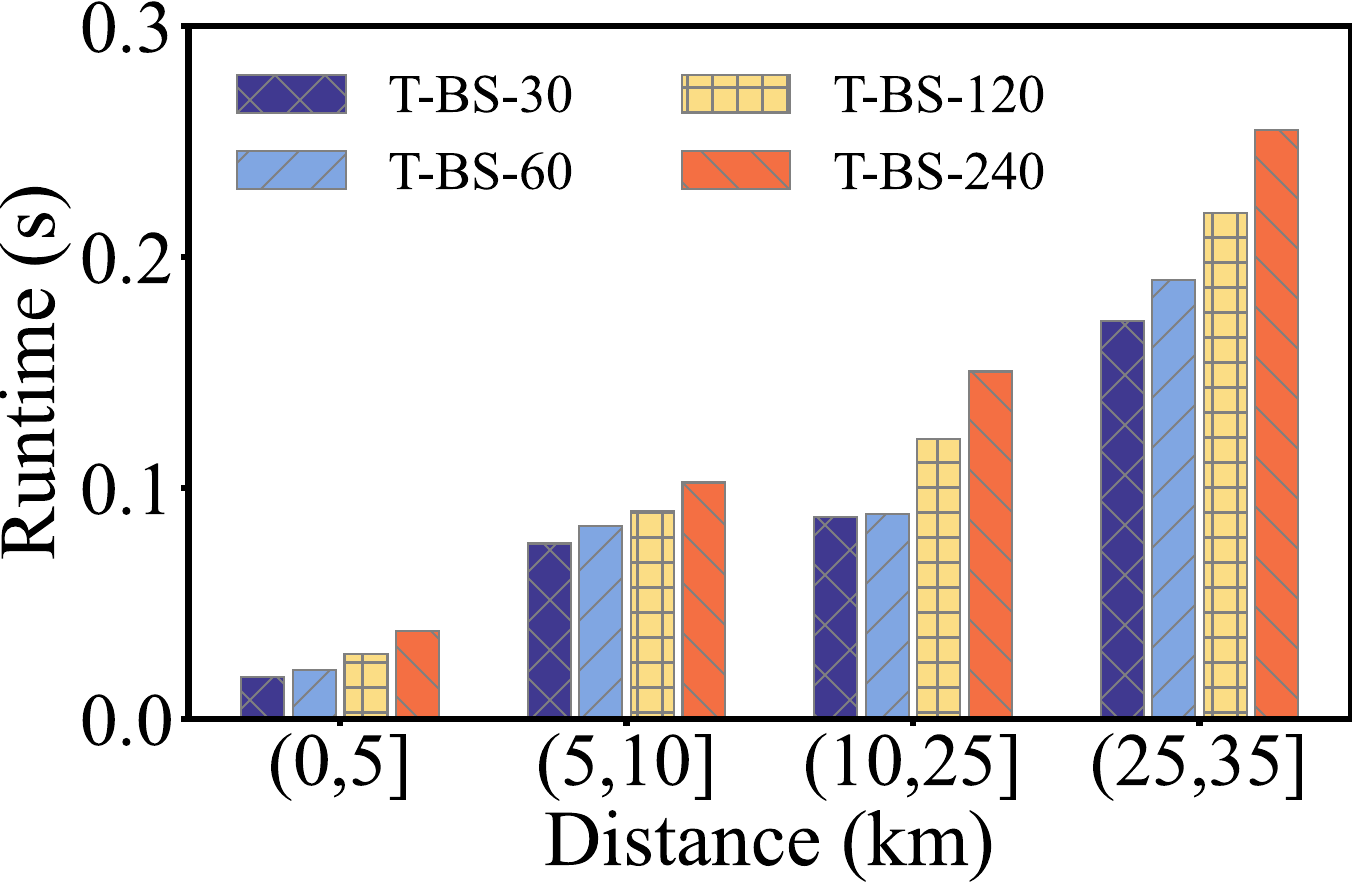}
	}
	\subfigure[By distances, $D_2$]{
	\centering
			\includegraphics[width=0.23\linewidth]{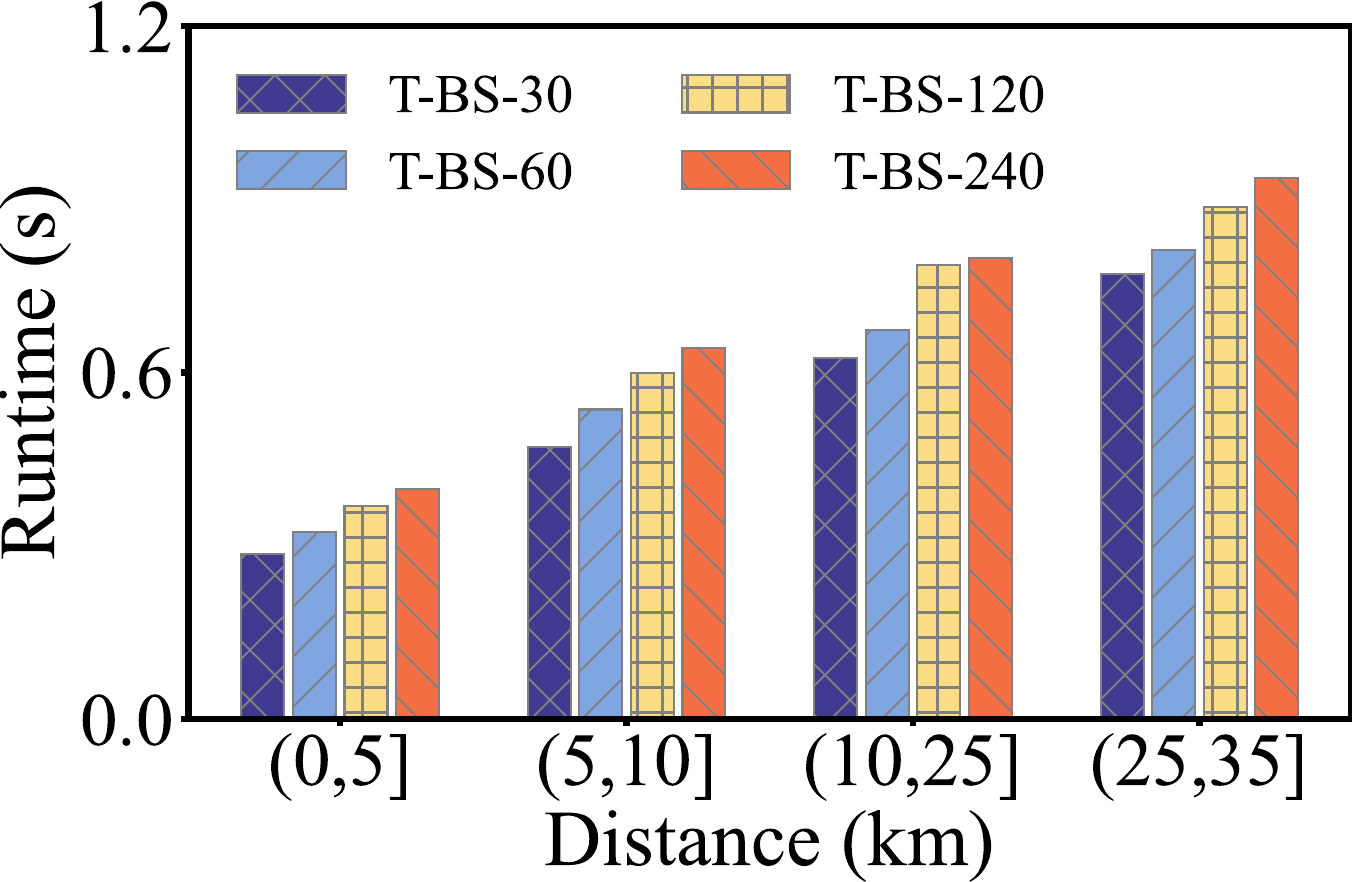}
	}
    \subfigure[By budget values, $D_1$]{
	\centering
			\includegraphics[width=0.23\linewidth]{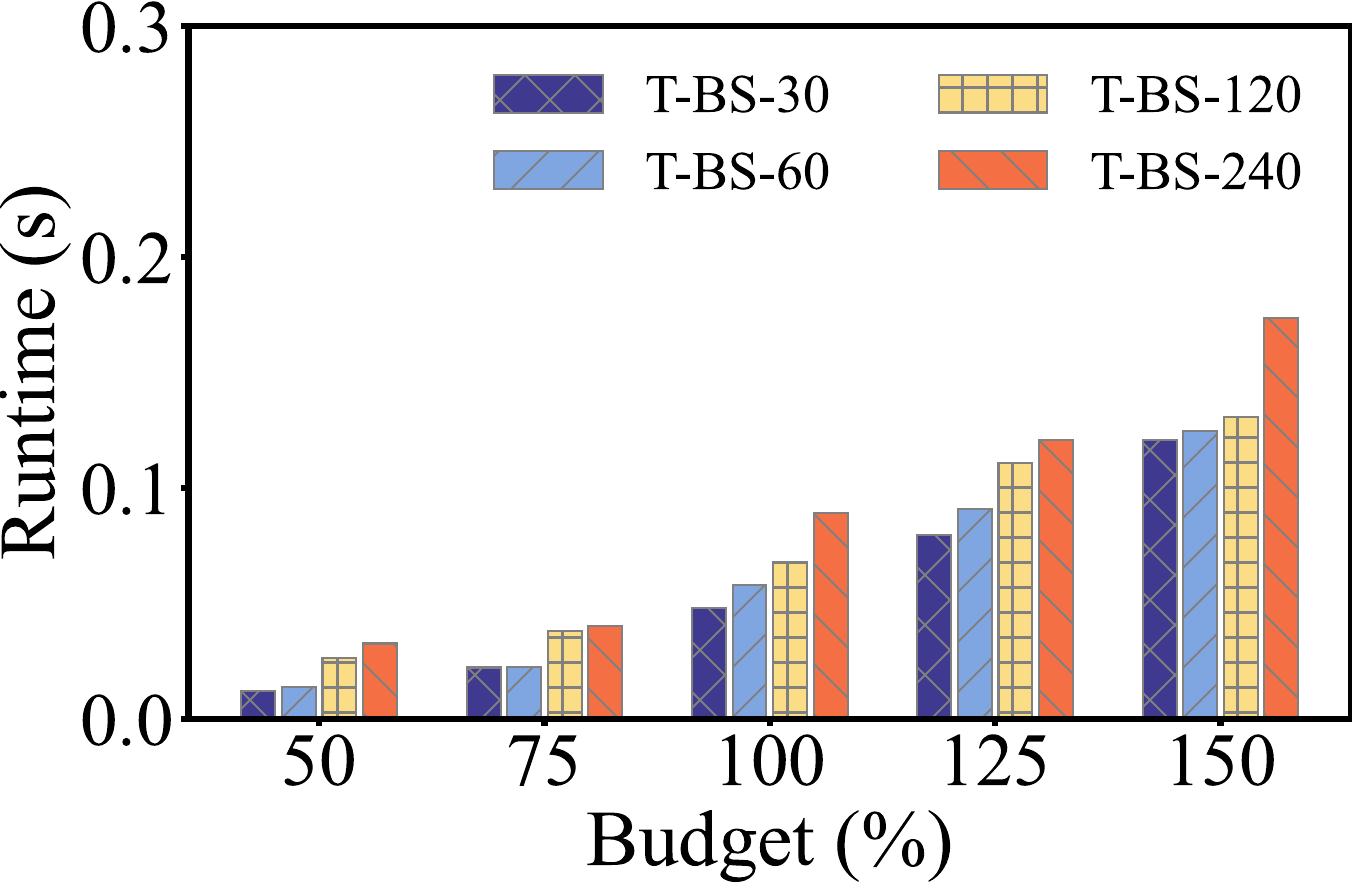}
	}
	\subfigure[By budget values, $D_2$]{
	\centering
			\includegraphics[width=0.23\linewidth]{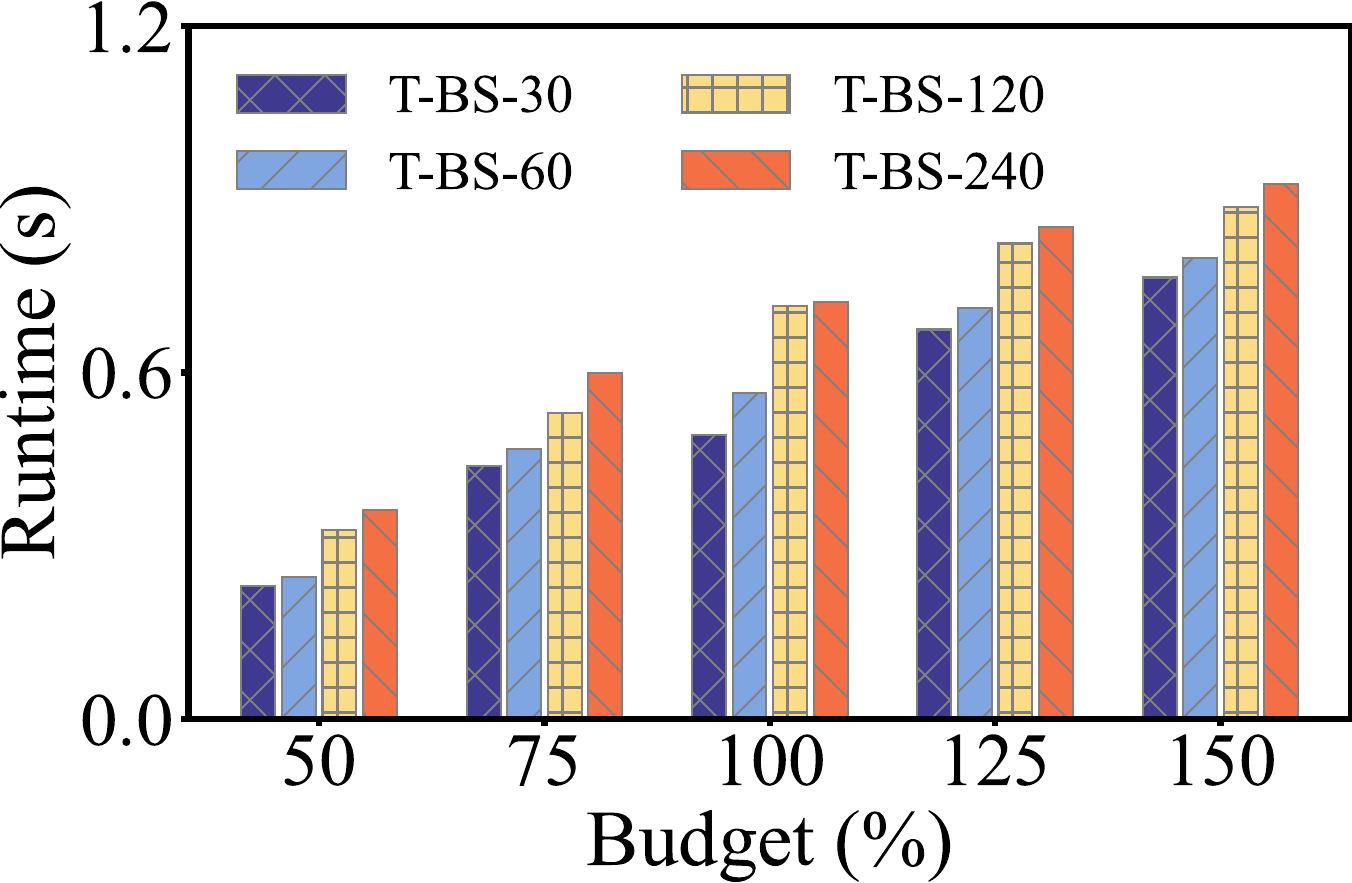}
	}
	\vspace{-1.5em}
	\caption{Stochastic Routing with Budget-Specific Heuristics at Peak Hours}
	\vspace{-10pt}
	\label{fig:budgetRouting-peak}
\end{figure*}

\begin{figure*}[!htp]
\centering
  \centering
	\subfigure[By distances, $D_1$]{
	\centering
			\includegraphics[width=0.23\linewidth]{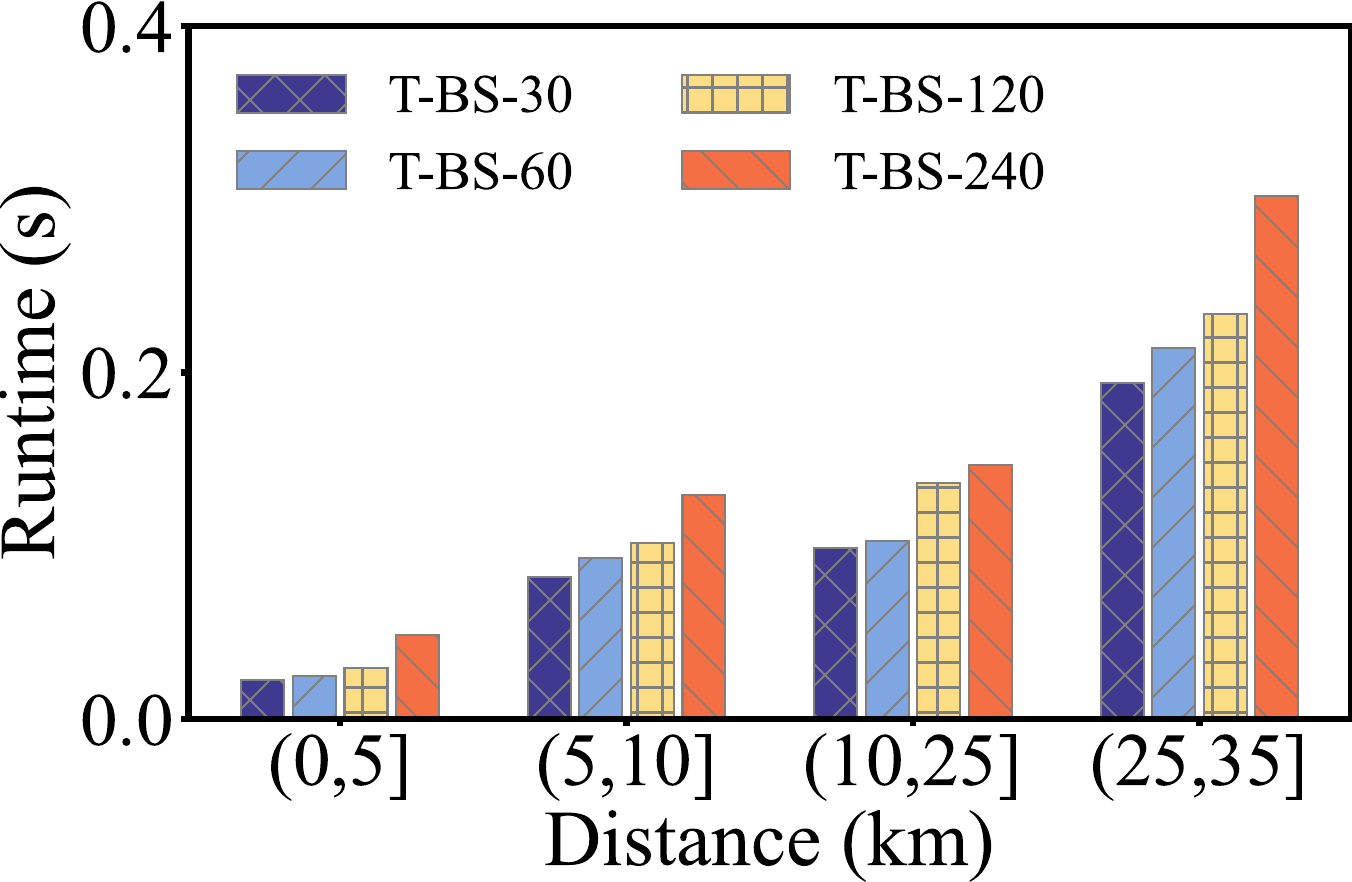}
	}
	\subfigure[By distances, $D_2$]{
	\centering
			\includegraphics[width=0.23\linewidth]{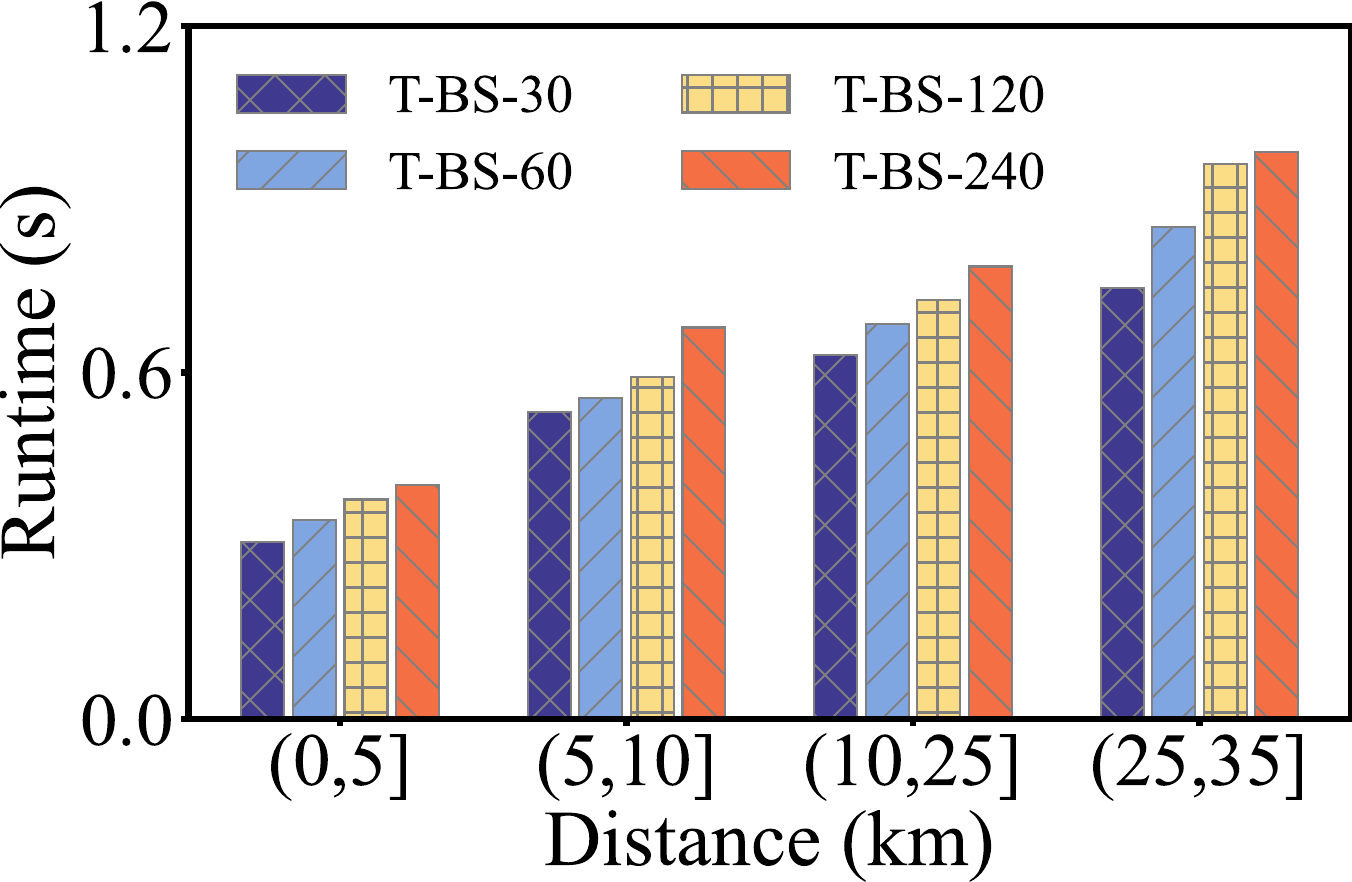}
	}
    \subfigure[By budget values, $D_1$]{
	\centering
			\includegraphics[width=0.23\linewidth]{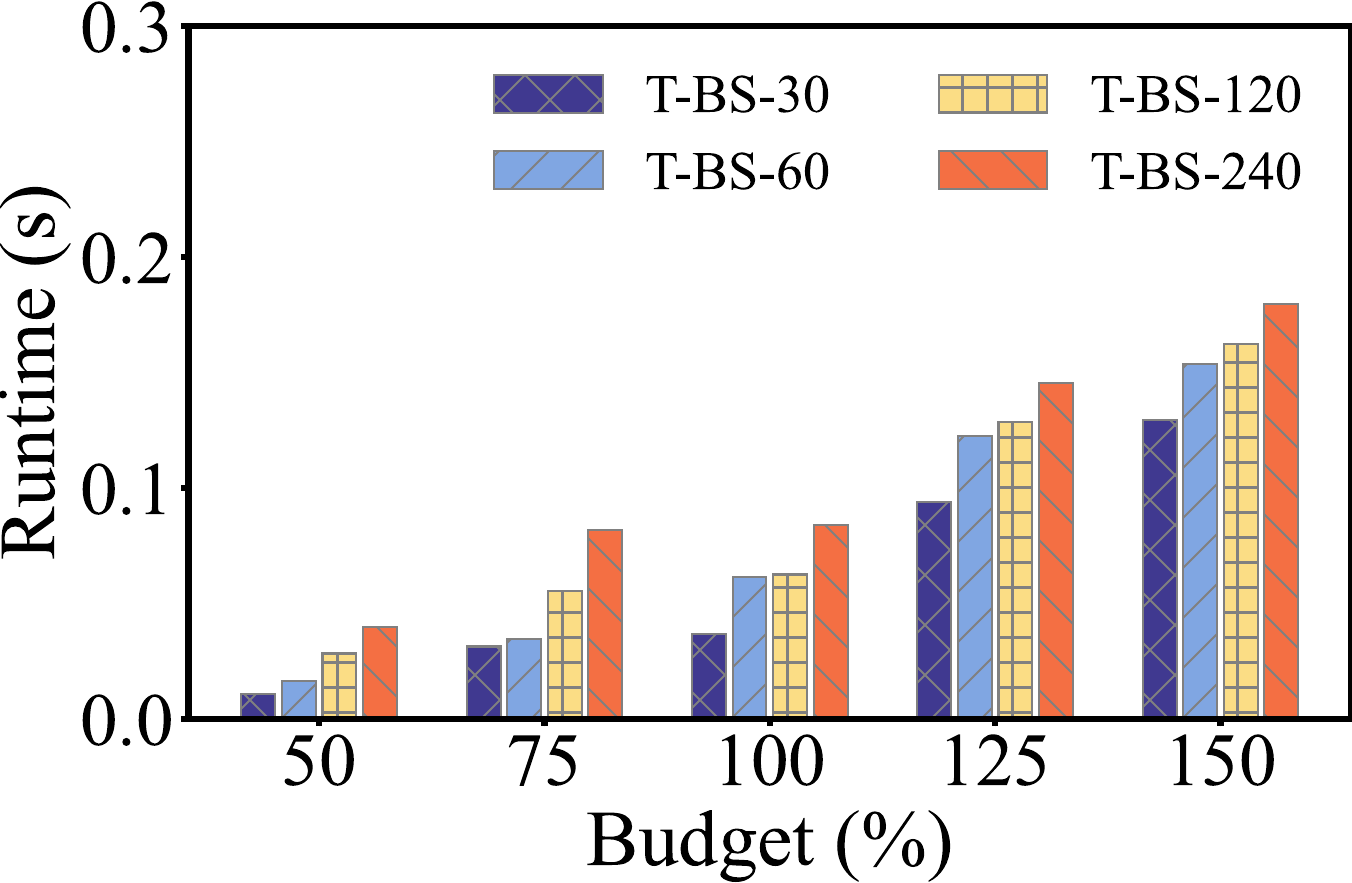}
	}
	\subfigure[By budget values, $D_2$]{
	\centering
			\includegraphics[width=0.23\linewidth]{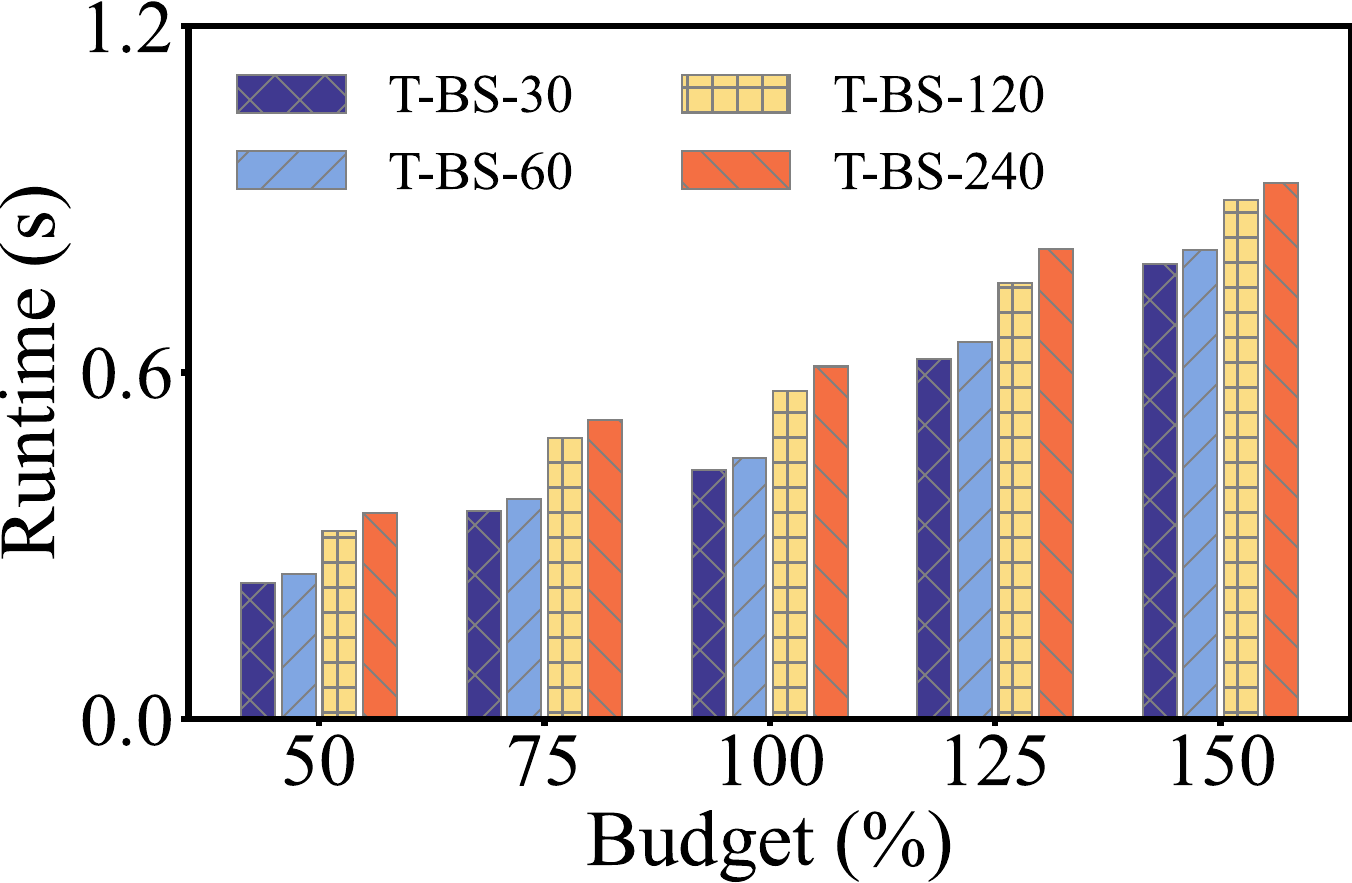}
	}
	\vspace{-1.5em}
	\caption{Stochastic Routing with Budget-Specific Heuristics at Off-Peak Hours}
	\vspace{-10pt}
	\label{fig:budgetRouting-offpeak}
\end{figure*}

\begin{figure*}[!htp]
\centering
  \centering
	\subfigure[By distances, $D_1$]{
	\centering
			\includegraphics[width=0.23\linewidth]{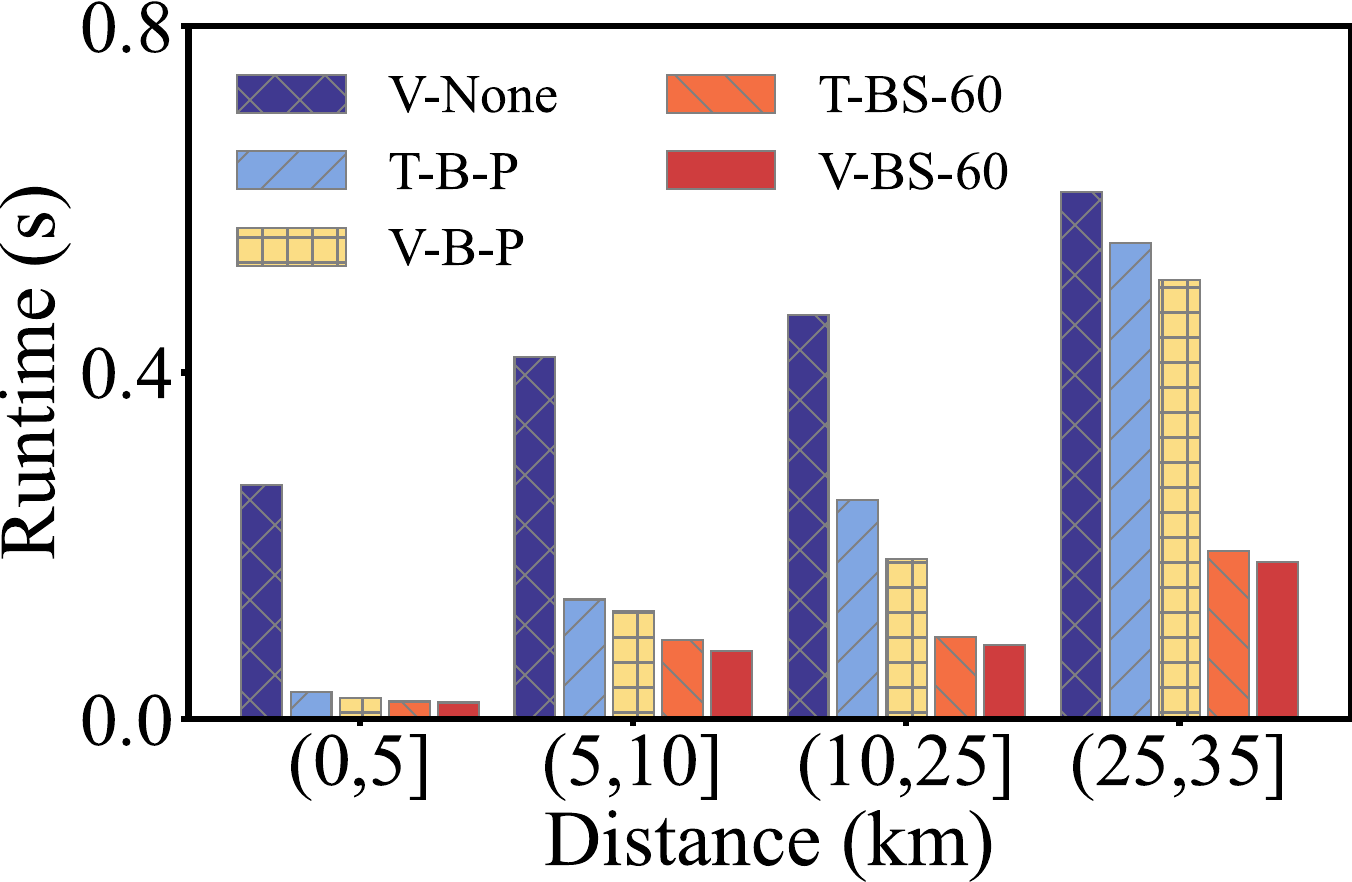}
	}
	\subfigure[By distances, $D_2$]{
	\centering
			\includegraphics[width=0.23\linewidth]{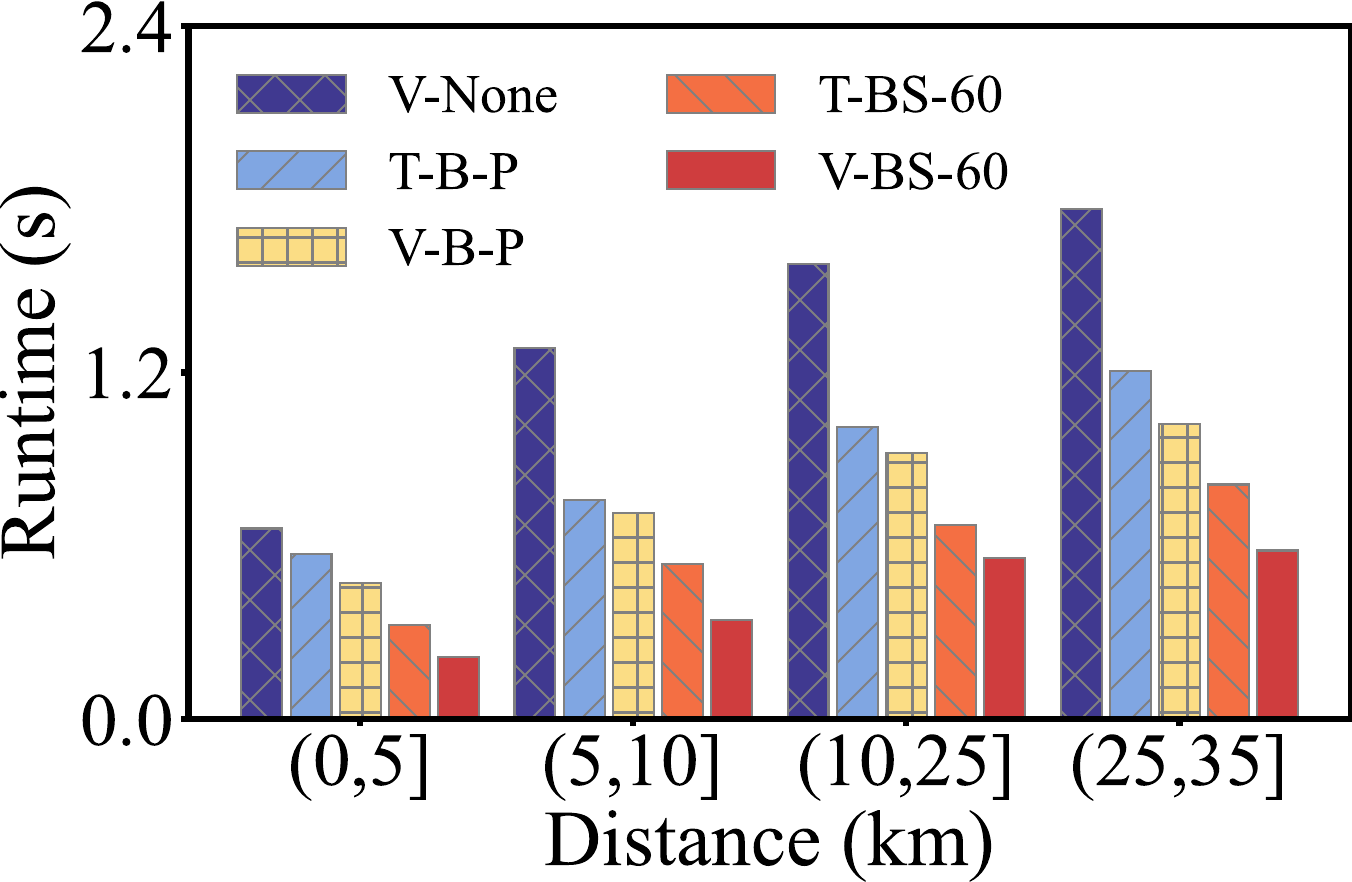}
	}
    \subfigure[By budget values, $D_1$]{
	\centering
			\includegraphics[width=0.23\linewidth]{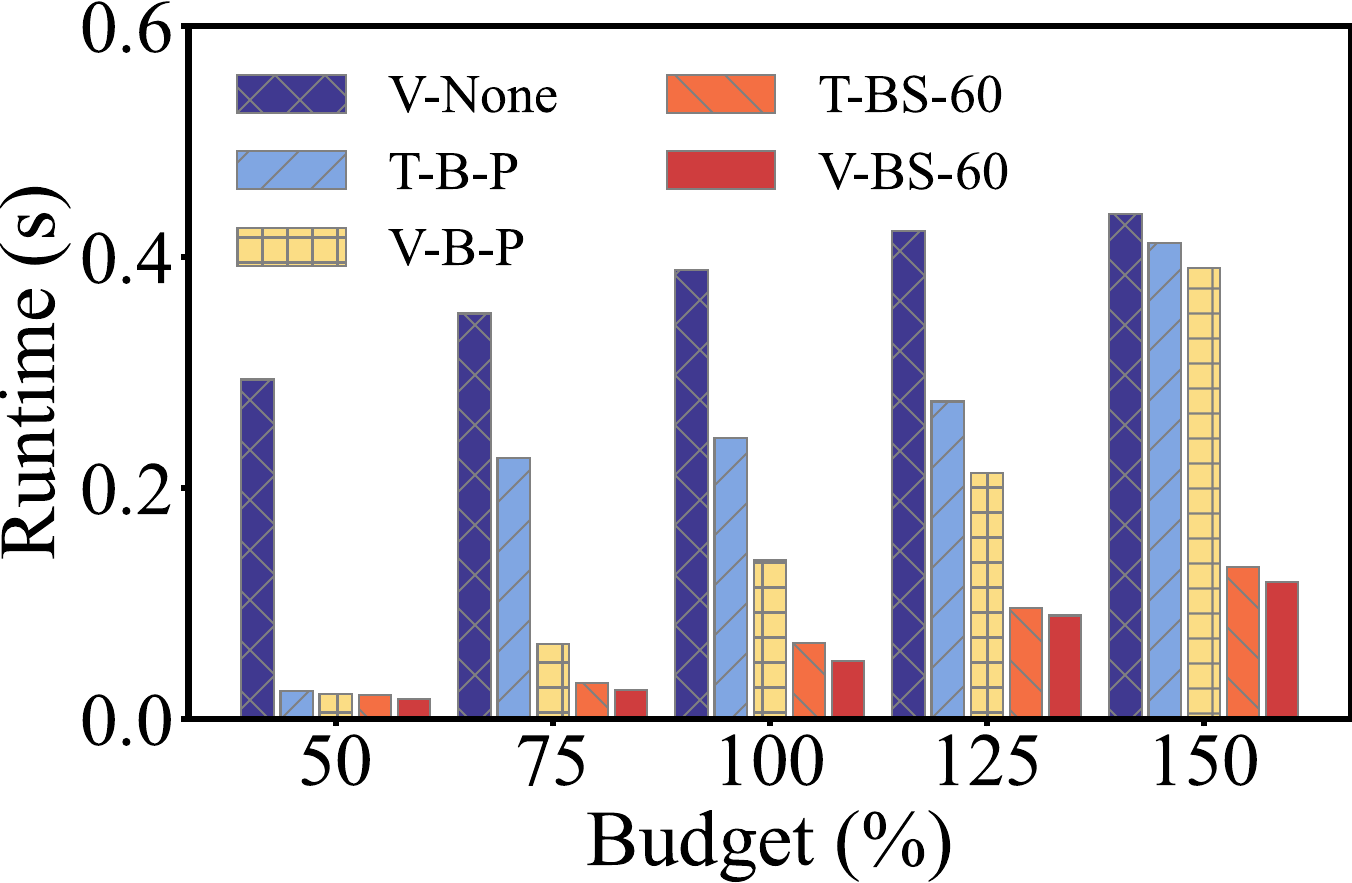}
	}
	\subfigure[By budget values, $D_2$]{
	\centering
			\includegraphics[width=0.23\linewidth]{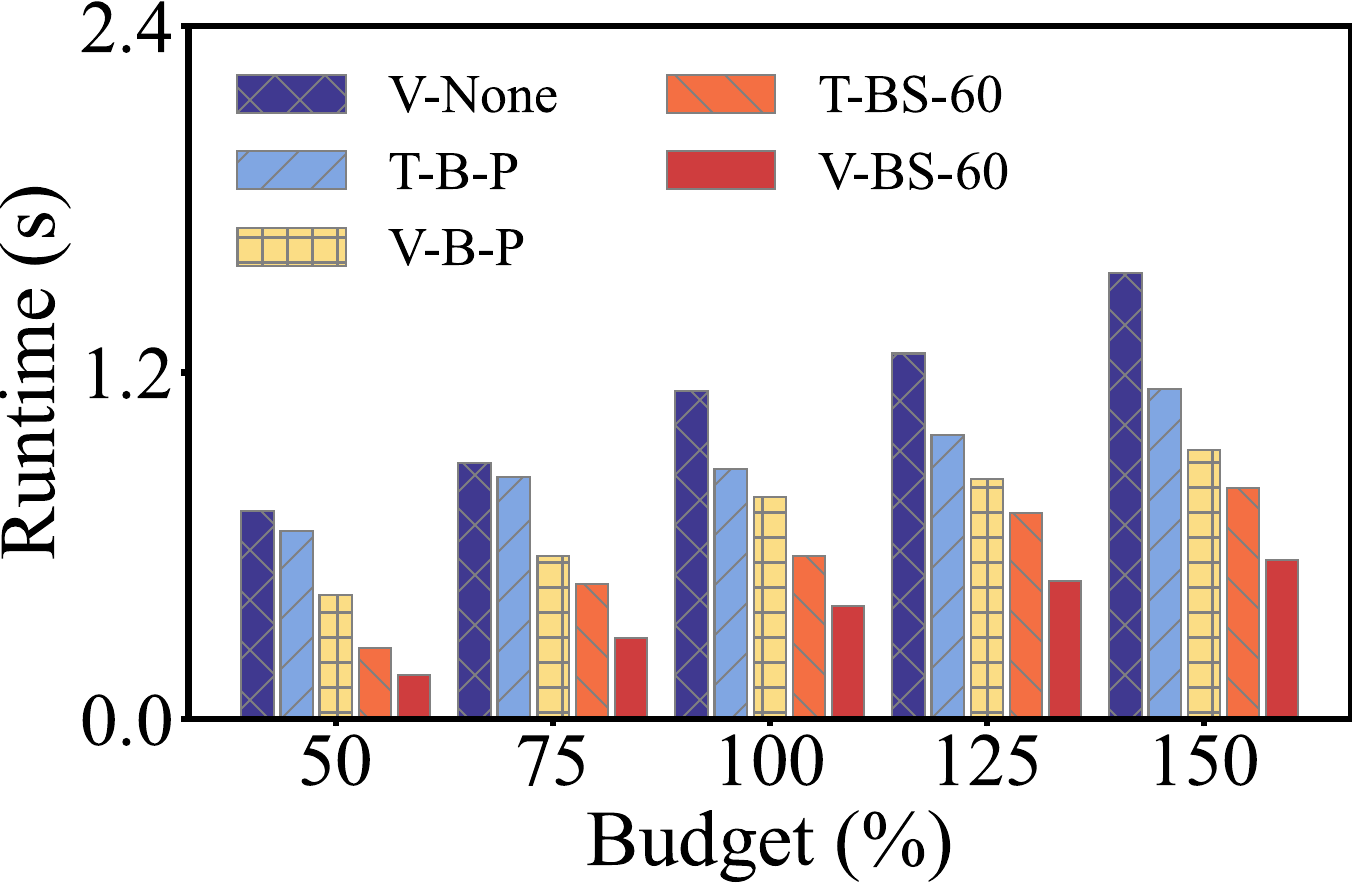}
	}
	\vspace{-1.5em}
	\caption{V-Path based Stochastic Routing at Peak Hours}
	\vspace{-10pt}
	\label{fig:vpathRouting-peak}
\end{figure*}

\begin{figure*}[!htp]
\centering
  \centering
	\subfigure[By distances, $D_1$]{
	\centering
			\includegraphics[width=0.23\linewidth]{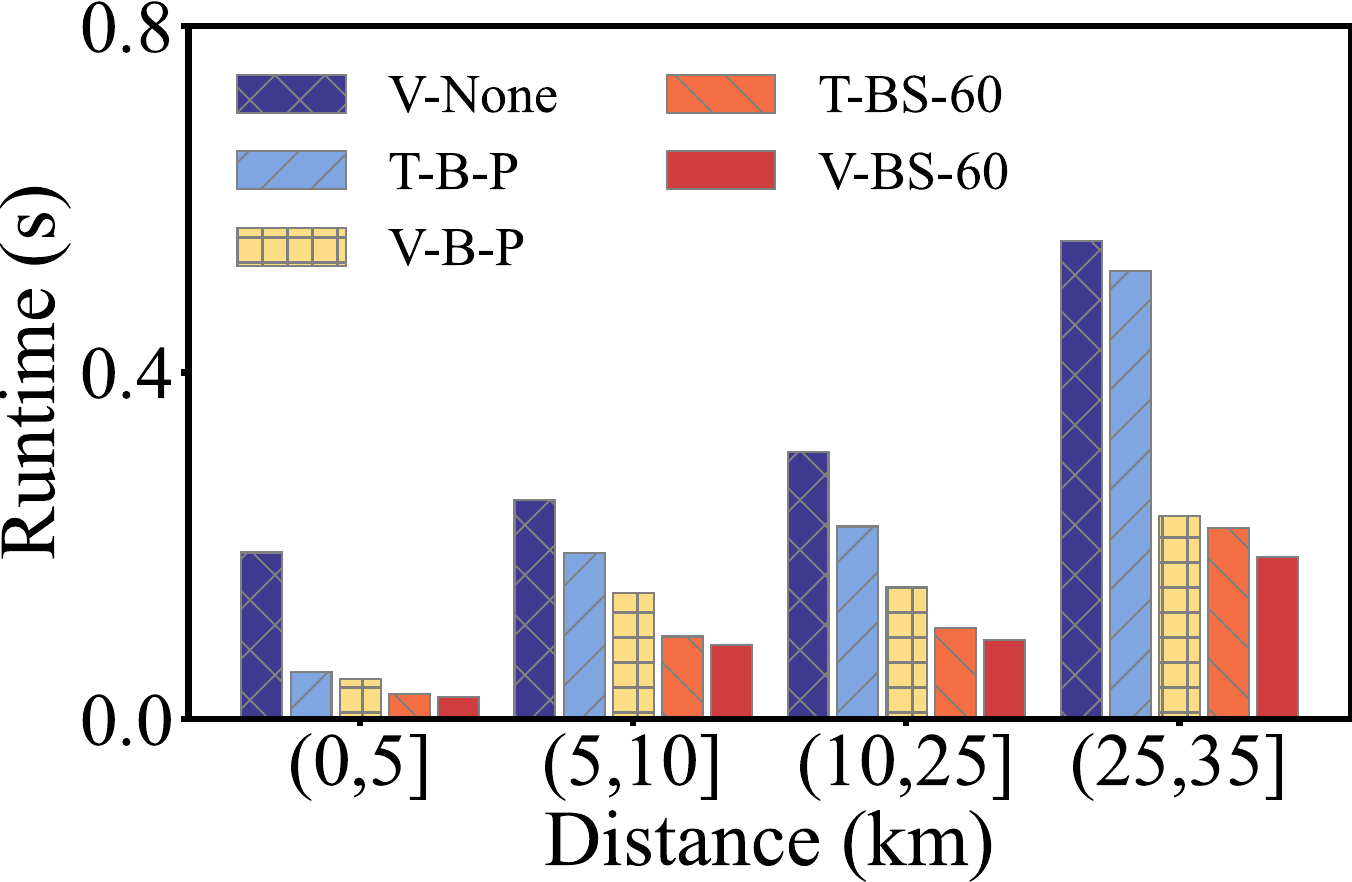}
	}
	\subfigure[By distances, $D_2$]{
	\centering
			\includegraphics[width=0.23\linewidth]{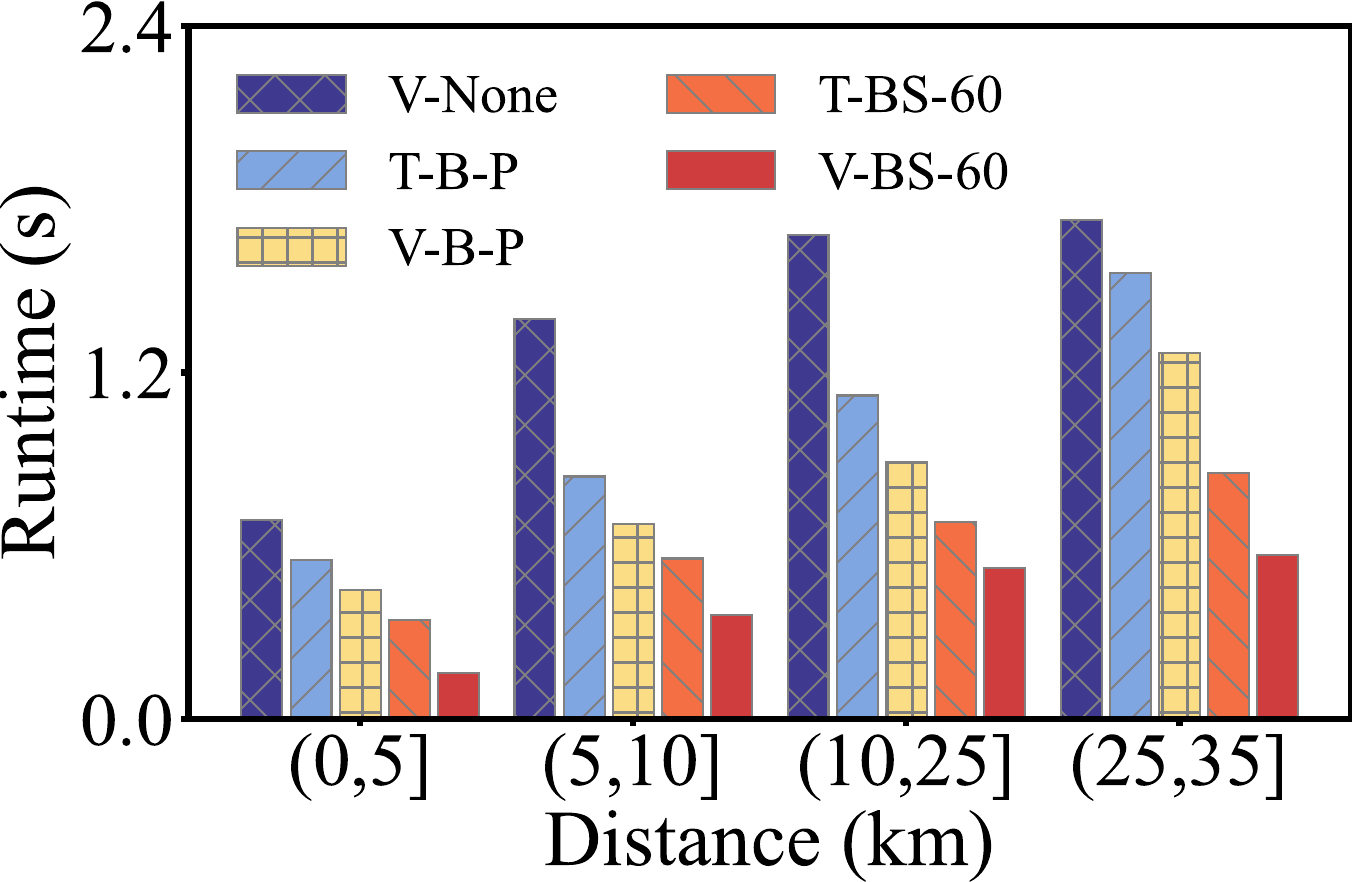}
	}
    \subfigure[By budget values, $D_1$]{
	\centering
			\includegraphics[width=0.23\linewidth]{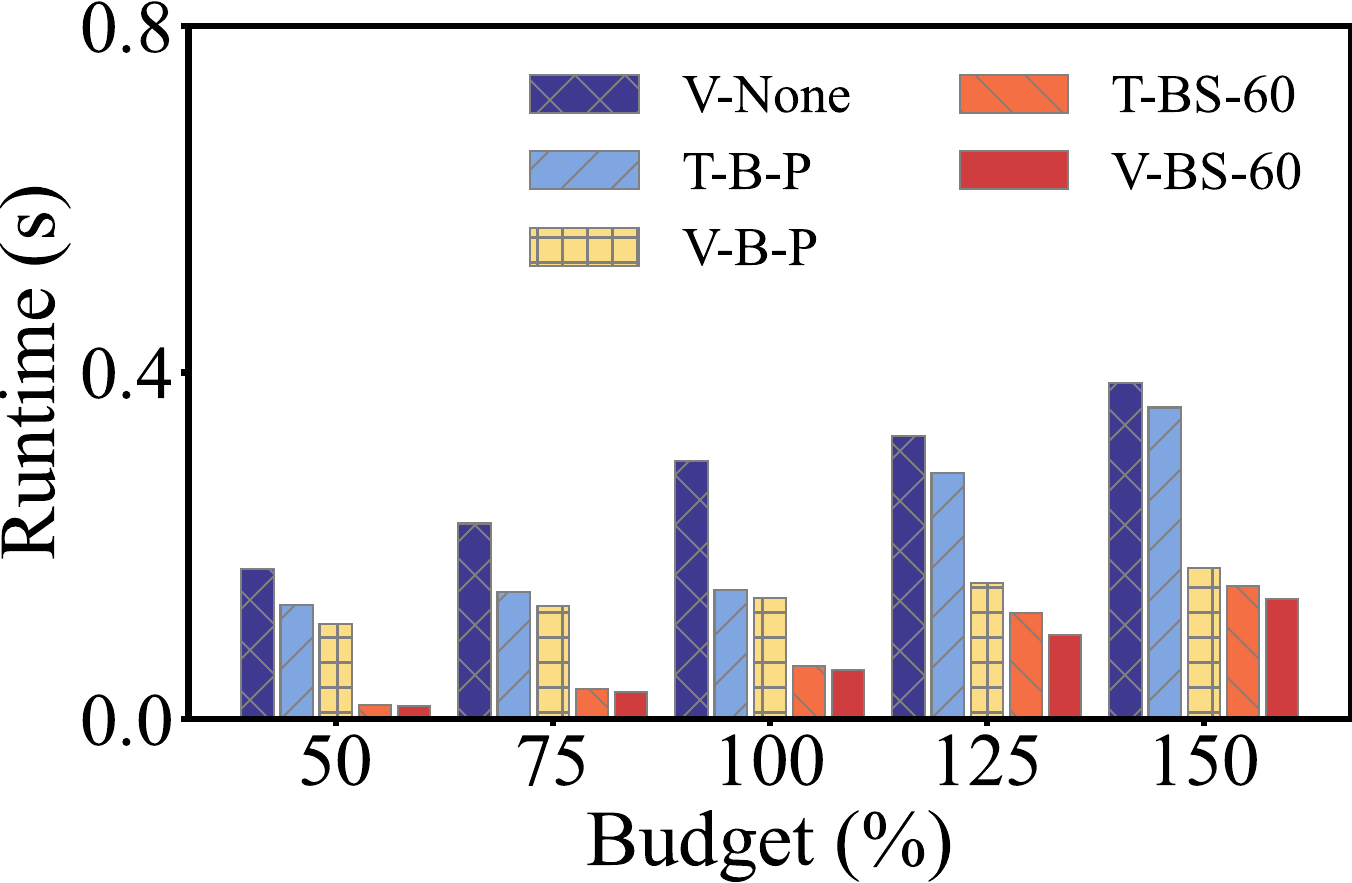}
	}
	\subfigure[By budget values, $D_2$]{
	\centering
			\includegraphics[width=0.23\linewidth]{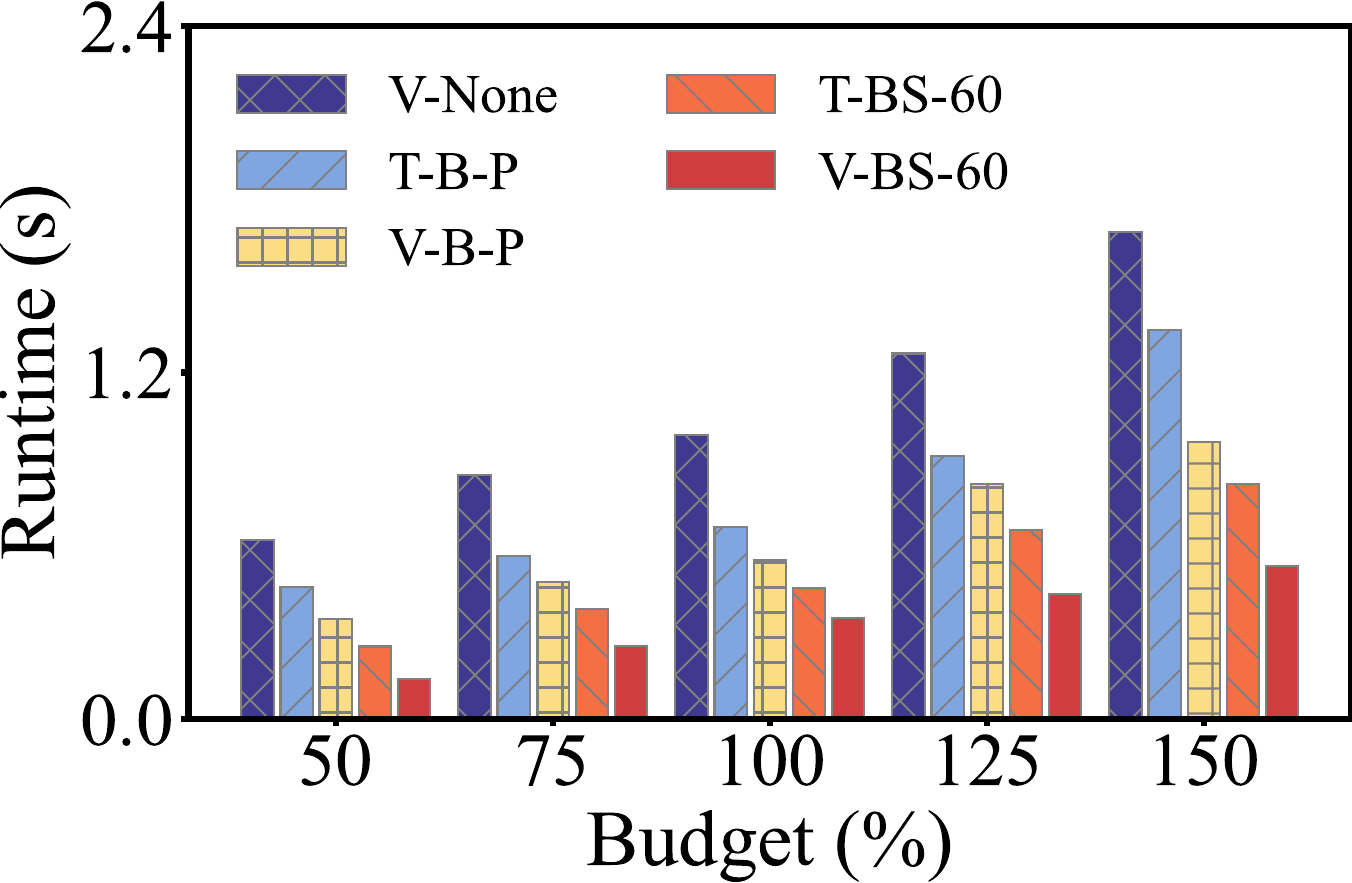}
	}
	\vspace{-1.5em}
	\caption{V-Path based Stochastic Routing at Off-Peak Hours}
	\vspace{-10pt}
	\label{fig:vpathRouting-offpeak}
\end{figure*}

The storage needed for different variations is the same---we maintain a single integer value $v.\mathit{getMin}()$, regardless of the variation. 
Maintaining $v.\mathit{getMin}()$ for all destinations takes 0.92 GB, which can easily fit into main memory. We report the overall offline pre-computation cost for building binary heuristics for $D_1$ and $D_2$ in Table~\ref{tbl:binary}, where run time and storage is reported by hours (h) and  gigabytes (GB), respectively.

\begin{table}[!htp]
\footnotesize
	\centering
	\renewcommand{\arraystretch}{1}
	\begin{tabular}{|c||c|c|c||c|c|c|}
		\hline		\multicolumn{7}{|c|}{Data Set $D_1$}\\ \hline
		&\multicolumn{3}{c||}{{\bf Off-Peak Hours}}&\multicolumn{3}{c|}{{\bf Peak Hours}}\\ \cline{2-7}
	     Methods & T-B-EU & T-B-E & T-B-P& T-B-EU & T-B-E & T-B-P \\ \hline
		Run time (h) & 1.8  & 17.6  & 32.3  & 2  & 19.2  & 24.7 \\ \hline
		Storage (GB) & 0.92  & 0.92  & 0.92  & 0.92  &  0.92 & 0.92 \\ \hline
	\end{tabular}
	\begin{tabular}{|c||c|c|c||c|c|c|}
		\hline
		\multicolumn{7}{|c|}{Data Set $D_2$}\\ \hline
		&\multicolumn{3}{c||}{{\bf Off-Peak Hours}}&\multicolumn{3}{c|}{{\bf Peak Hours}}\\ \cline{2-7}
	          & T-B-EU & T-B-E & T-B-P& T-B-EU & T-B-E & T-B-P \\ \hline
            Run time (h) & 8.5  & 92.8  & 224.5  & 12.6  & 108.1  & 239.2 \\ \hline
		Storage (GB) & 3.06  & 3.06  & 3.06  & 3.06  & 3.06   & 3.06 \\ \hline
	\end{tabular}
	\caption{Binary Heuristics Pre-computation}
	\label{tbl:binary}
 \vspace{-2em}
\end{table}

We next investigate the runtime when using binary heuristics vs. T-None that uses no heuristics in Figures~\ref{fig:binaryRouting-peak} and~\ref{fig:binaryRouting-offpeak}. 
We categorize the routing queries based on their source-destination distances (cf. Figure~\ref{fig:binaryRouting-peak} (a) and (b), Figure~\ref{fig:binaryRouting-offpeak} (a) and (b)) and their time budgets (cf. Figure~\ref{fig:binaryRouting-peak} (c) and (d), Figure~\ref{fig:binaryRouting-offpeak} (c) and (d)). 

The longer the distances, the more time it takes to perform the routing. 
Larger travel time budgets also incur longer runtime. T-B-EU takes longer, i.e., sometimes over 0.5 second, than T-B-E and T-B-P at sub 0.1 second, suggesting that the shortest path tree based binary heuristics used by T-B-E and T-B-P are more accurate and are effective at pruning unpromising candidate paths. 
Further, T-B-P offers the best efficiency, indicating that shortest path tree derived from both edges and T-paths is the most accurate and effective at pruning. 

\noindent
{\bf Budget-specific Heuristics:} 
We consider the largest budget of $5,000$ seconds for both data sets and vary $\delta$ among 30, 60, and 120, and 240. 
We argue that 5,000 seconds is a very large time budget when traveling up to 35 km, i.e., the longest distance considered in the routing queries. %

We first consider the offline phase of building the budget-specific heuristics, i.e., the heuristics tables. These tables are destination-specific.
When building budget-specific heuristics, we did not use parallelization, which, however, is possible because building the heuristics for different vertices is independent. The offline pre-computation cost for $D_1$ and $D_2$ is reported in Table~\ref{tbl:budget}, where run time and storage is reported by hours (h) and gigabytes (GB), respectively.

\vspace{-2mm}
\begin{table}[!htp]
\footnotesize
\renewcommand{\arraystretch}{1}
	\centering
	\begin{tabular}{|c||c|c|c|c||c|c|c|c|}
		\hline
		\multicolumn{9}{|c|}{Data Set $D_1$}\\ \hline
		&\multicolumn{4}{c||}{{\bf Off-Peak Hours}}&\multicolumn{4}{c|}{{\bf Peak Hours}}\\ \cline{2-9}
	      $\delta$ & 30 & 60 & 120 & 240 & 30 & 60 & 120 & 240 \\ \hline
		Run time (h) & 43.2  & 33.4  & 20.4  & 16.8  & 56  & 28.5  & 17.4  & 13.2 \\ \hline
		Storage (GB) & 1.87  & 1.37  & 1.02  & 0.92 & 1.34  & 1.07  & 0.97  & 0.92   \\ \hline
	\end{tabular}
	\setlength{\tabcolsep}{1.2mm}{\begin{tabular}{|c||c|c|c|c||c|c|c|c|}
		\hline
		\multicolumn{9}{|c|}{Data Set $D_2$}\\ \hline
		&\multicolumn{4}{c||}{{\bf Off-Peak Hours}}&\multicolumn{4}{c|}{{\bf Peak Hours}}\\ \cline{2-9}
	      $\delta$ & 30 & 60 & 120 & 240 & 30 & 60 & 120 & 240 \\ \hline
            Run time (h) & 318.8  & 258.2  & 225.7  & 207.9  & 289.2  & 247.1  & 223.8  & 213.2 \\ \hline
		Storage (GB) & 5.95  & 4.45  & 3.41  & 3.06  & 4.38  & 3.54  & 3.21  & 3.06 \\ \hline
	\end{tabular}}
	\caption{ Budget-Specific Heuristics Pre-computation}
	\label{tbl:budget}
 \vspace{-2em}
\end{table}

Figure~\ref{fig:Budget-Specific-preprocessing} reports the average pre-computation time of constructing a heuristic table and its average size for both $D_1$ and $D_2$ at peak hours. The results indicate that smaller $\delta$ values result in larger heuristics tables with more columns and that longer pre-computation times are needed to obtain the heuristics tables. 

Next, we investigate how the budget-specific heuristics help stochastic routing. 
We again categorize routing queries w.r.t. both distances and budge values. 
Figures~\ref{fig:budgetRouting-peak} and~\ref{fig:budgetRouting-offpeak} 
show that queries with longer distances and travel time budgets incur longer online runtime.
Smaller $\delta$ values produce heuristic tables at finer granularities, thus yielding more effective pruning and lower online runtimes. 

We notice that, for $D_2$ at peak hours, the online runtimes for $\delta=30$ and $\delta=60$ are similar, but they require 289.2 hours and 247.1 hours (without using parallelization) to generate and occupy 4.38 GB and 3.54 GB to store heuristic tables for all destinations, respectively. 
%
As $\delta=60$ requires less pre-computation time and space, we choose it as the default value and conclude that it is practical to trade pre-computation time and storage for interactive response times.

To compare the binary heuristics vs. the budget-specific heuristics (using the default $\delta$ value), we include T-BS-60 in Figures~\ref{fig:binaryRouting-peak} and~\ref{fig:binaryRouting-offpeak}). 
We observe that the budget-specific heuristics significantly improve the runtime compared to the binary heuristics.  
\subsection{V-Path Based Routing}
\label{ssec:vpathrouting}

Figures~\ref{fig:vpathRouting-peak} and ~\ref{fig:vpathRouting-offpeak} show results for V-Path based routing, which in general are much faster than routing with only T-paths. 
In this set of experiments, we only use default $\delta=60$. V-BS-60 achieves the least rounting runtime, which further improves over T-BS-60, suggesting that the  V-path based routing offers the best efficiency, because using long V-paths often makes it possible to go to somewhere close to the destination very quickly. 
When combined with stochastic dominance based pruning, the efficiency is further improved. 

Table~\ref{tbl:comparison} presents comprehensive statistics of all heuristic methods. Although V-BS-60 requires slightly more storage and a longer precomputation time than the other methods, it excels by offering the fastest routing runtime.

\vspace{-1em}
\begin{table}[h]
\footnotesize
	\centering
	\renewcommand{\arraystretch}{1.1}
	\setlength{\tabcolsep}{0.5mm}{\begin{tabular}{|c|c|c|c|c|c|c|}
		\hline		\multicolumn{7}{|c|}{Data Set $D_1$, Peak Hours}\\ \hline
	     Methods & T-B-EU & T-B-E & T-B-P&V-B-P& T-BS-60 & V-BS-60 \\ \hline
		Storage (GB) & 0.92  & 0.92  & 0.92  &0.92& 1.07  & 1.09 \\ \hline
    Precomputation (h) & 2  & 19.2  & 24.7   & 32.3 &  29.1 & 37.0 \\ 
    \hline
    Routing (s) &0.364&0.235&0.214& 0.142&0.078 &0.067\\ \hline
	\end{tabular}
	\begin{tabular}{|c|c|c|c|c|c|c|}
		\hline
		\multicolumn{7}{|c|}{Data Set $D_2$, Peak Hours}\\ \hline
	          Methods & T-B-EU & T-B-E & T-B-P& V-B-P & T-BS-60 &  V-BS-60 \\ \hline
    Storage (GB) & 3.06  & 3.06  & 3.06  & 3.06 & 3.52& 3.54   \\ \hline
      Precomputation (h)  & 12.6  & 108.1  & 239.2   & 284.5  & 248.5 & 302.1 \\ \hline
    Routing (s) & 1.157  & 1.126 & 1.017    & 0.993  & 0.577 & 0.402 \\ 
    \hline
	\end{tabular}}
	\caption{Comparison of different methods}
	\label{tbl:comparison}
  \vspace{-4em}
\end{table}

\subsection{Case Study}
To illustrate the benefits of the proposed methods, we analyze two representative queries from the datasets for peak hours. We compare V-BS-60 against Google Maps for Aalborg and Baidu Maps for Xi’an, as these platforms are widely utilized in the respective cities. Blue lines denote the routes provided by the platforms and red lines denote the routes provided by the proposed method. In instances of overlap, the blue lines are placed on top of the red lines. We assess the probability of arrival within the budget time for all paths based on the cost distributions. 

As depicted in Figure~\ref{fig:case} (a), we set a budget of 15 minutes, within which the V-BS-60 route has a 72.1\% probability of arrival, outperforming the Google Maps route that has a 66.7\% probability. Similarly, for the Xi’an dataset with a 35-minute budget (as shown in Figure~\ref{fig:case} (b)), the probability of the V-BS-60 route exceeds that of the Baidu Maps route. These cases illustrate how the proposed method is capable of enhancing the probability of arriving within a time budget.
\vspace{-1em}
\begin{figure}[!htp]
	\centering
	\subfigure[Aalborg (budget for 15 minutes at peak hours)]{
	\centering
			\includegraphics[height=0.67\linewidth]{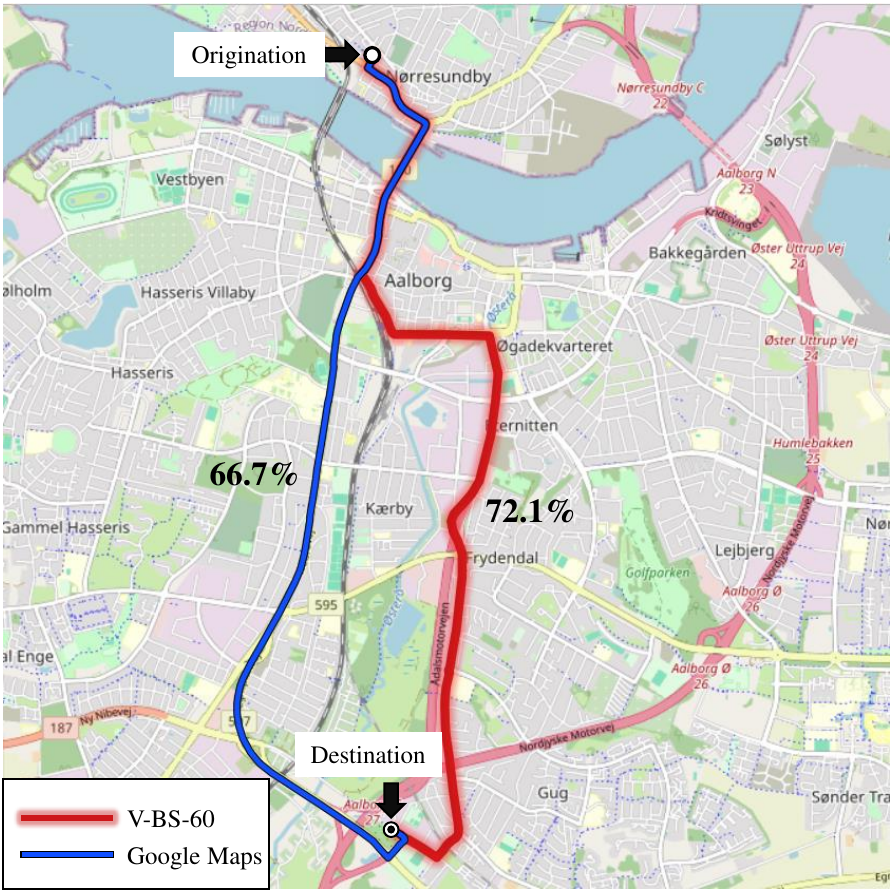}
	}
	\subfigure[Xi'an (budget for 35 minutes at peak hours)]{
	\centering
			\includegraphics[height=0.67\linewidth]{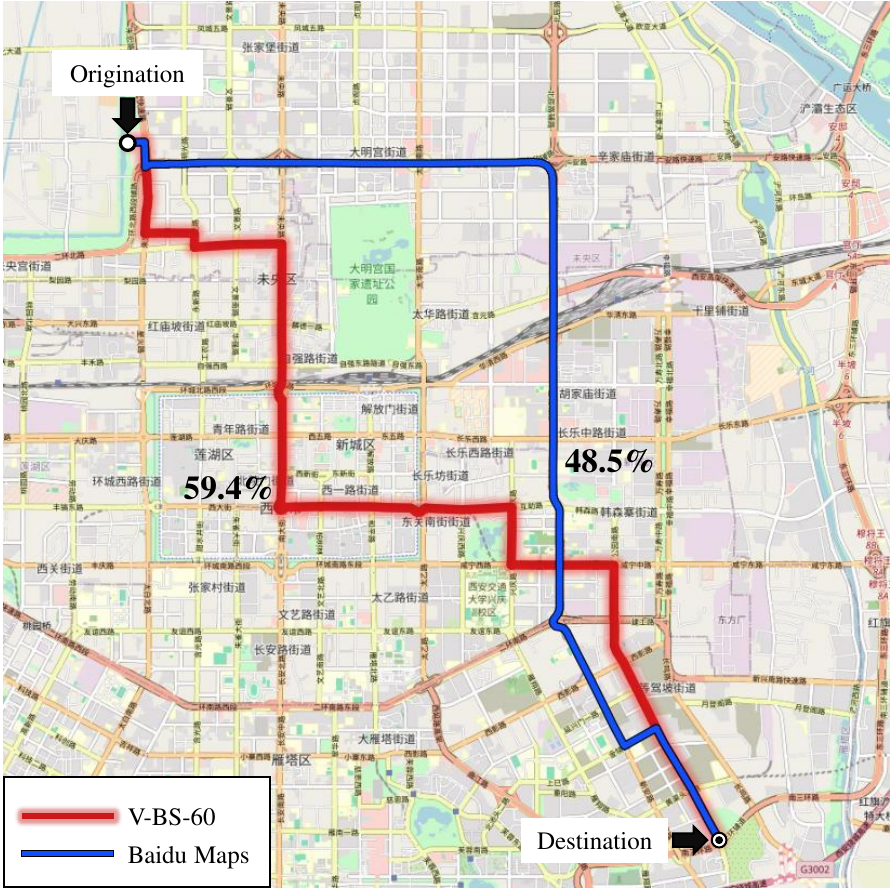}
	}
	\vspace{-1.5em}
	\caption{Case Study}
	\vspace{-8pt}
	\label{fig:case}
\end{figure}

\section{Related Work}
\label{sec:related}

\noindent
{\bf Modeling Travel Costs in Road Networks:} 
Most studies model travel costs in the edge-centric model, meaning that weights are assigned to edges. 
First, some models assign deterministic weights to edges. An edge weight can be a single deterministic cost representing, e.g., average travel-time~\cite{DBLP:conf/aaai/LetchnerKH06,DBLP:conf/aaai/ZhengN13,DBLP:journals/tkde/YangKJ14}, or it can be several deterministic costs representing different travel costs such as average travel-time and fuel consumption~\cite{DBLP:journals/geoinformatica/Guo0AJT15,DBLP:conf/icde/Guo0AJT15}. 
Second, some models assign uncertain weights to edges. 
Most such studies assume that the weights on different edges are independent~\cite{DBLP:conf/icde/YangGJKS14,DBLP:conf/aips/NikolovaBK06, DBLP:conf/uai/WellmanFL95}. 
Some studies consider the dependency between two edges, but disregards dependencies among multiple edges~\cite{DBLP:conf/edbt/HuaP10,DBLP:journals/pvldb/PedersenYJ20, DBLP:journals/pvldb/0002GJ13}, including 
two recent studies that employ machine learning models to infer dependencies 
between two edges where trajectories are unavailable~\cite{DBLP:journals/pvldb/PedersenYJ20, DBLP:journals/pvldb/0002GJ13}. Path representation learning produces effective path representations to enable accurate travel cost estimation~\cite{DBLP:journals/tkde/YangGY22, DBLP:conf/icde/Yang020, DBLP:conf/icde/YangGHYTJ22, DBLP:conf/ijcai/YangGHT021, DBLP:conf/kdd/YangHGYJ23}.  However, such methods are difficult to be integrated with existing routing algorithms.  
The path-centric model~\cite{pvldb17pathcost, DBLP:journals/vldb/YangDGJH18} aims at fully capture the cost dependencies among multiple edges along paths by maintaining joint distributions for the paths when sufficient trajectories are available. 
The path-centric model offers better accuracy when estimating cost distributions of paths over the edge-centric model~\cite{pvldb17pathcost, DBLP:journals/vldb/YangDGJH18}. Thus, we base our work on the path-centric model. 

\noindent
{\bf Stochastic Routing:} Stochastic routing is to determine the optimal routes for vehicles in a network under uncertain conditions~\cite{WangLT19, Pedersen0J20,lu2020, lin2022, pan2023}. Extensive research has been performed on stochastic routing under the edge-centric model~\cite{DBLP:journals/tsas/PedersenYJM23,DBLP:conf/sigmod/Ma0J14,DBLP:journals/vldb/GuoYHJC20,nie2006arriving,DBLP:conf/ijcai/BentH07,DBLP:journals/pvldb/JinLDW11,DBLP:conf/icde/YangGJKS14,DBLP:journals/vldb/PedersenYJ20,DBLP:conf/aips/NikolovaBK06, DBLP:conf/uai/WellmanFL95}. 
Efficient stochastic routing is often based on stochastic dominance based pruning, based on the assumption that edge weights are independent~\cite{DBLP:conf/aips/NikolovaBK06, DBLP:conf/uai/WellmanFL95,DBLP:journals/vldb/HuYGJ18}. 
In addition, efficient convolution computation techniques have been proposed and integrated into routing algorithms to improve efficiency~\cite{nie2006arriving,dean2010speeding}. 
Well-known speed-up techniques for classic shortest path finding, e.g., partition-based reach, arc-flags, arc-potentials, and contraction hierarchies, have also been extended to support uncertain weights in the edge-centric model~\cite{DBLP:conf/alenex/SabranSB14,DBLP:journals/vldb/PedersenYJ20}. 
However, these speed-up techniques rely on the independence assumption that does not hold in the path-centric model, and thus they cannot be applied readily.  
For the path-centric model, stochastic routing algorithms suffer from low efficiency~\cite{DBLP:journals/vldb/YangDGJH18,DBLP:conf/mdm/Andonov018}, where we propose search heuristics and virtual-path based routing to improve routing efficiency. 

\section{Conclusion}
\label{sec:con}

We study stochastic routing under the path-centric model that can capture accurate travel cost distributions. Instead of providing stochastic routing services using edge devices, we target cloud-based stochastic routing services to enable applications such as coordinated fleet transportation. In this setting, sufficient storage is accessible. Thus, we trade increased off-line storage for efficient online routing services. Specifically, we propose binary heuristics and a budget-specific heuristic, which help us explore more promising candidate paths. Then, we introduce virtual paths to make it possible to use stochastic dominance based pruning in the path-centric model. Experimental results suggest that the proposed methods enable efficient stochastic routing under the path-centric model. 
\\



\bibliographystyle{ACM-Reference-Format}
\balance
\bibliography{erc}


\begin{thebibliography}{41}


\ifx \showCODEN    \undefined \def \showCODEN     #1{\unskip}     \fi
\ifx \showDOI      \undefined \def \showDOI       #1{#1}\fi
\ifx \showISBNx    \undefined \def \showISBNx     #1{\unskip}     \fi
\ifx \showISBNxiii \undefined \def \showISBNxiii  #1{\unskip}     \fi
\ifx \showISSN     \undefined \def \showISSN      #1{\unskip}     \fi
\ifx \showLCCN     \undefined \def \showLCCN      #1{\unskip}     \fi
\ifx \shownote     \undefined \def \shownote      #1{#1}          \fi
\ifx \showarticletitle \undefined \def \showarticletitle #1{#1}   \fi
\ifx \showURL      \undefined \def \showURL       {\relax}        \fi
\providecommand\bibfield[2]{#2}
\providecommand\bibinfo[2]{#2}
\providecommand\natexlab[1]{#1}
\providecommand\showeprint[2][]{arXiv:#2}

\bibitem[\protect\citeauthoryear{Andonov and Yang}{Andonov and Yang}{2018}]%
        {DBLP:conf/mdm/Andonov018}
\bibfield{author}{\bibinfo{person}{Georgi Andonov} {and} \bibinfo{person}{Bin Yang}.} \bibinfo{year}{2018}\natexlab{}.
\newblock \showarticletitle{Stochastic Shortest Path Finding in Path-Centric Uncertain Road Networks}. In \bibinfo{booktitle}{\emph{MDM}}. \bibinfo{pages}{40--45}.
\newblock


\bibitem[\protect\citeauthoryear{Bent and Hentenryck}{Bent and Hentenryck}{2007}]%
        {DBLP:conf/ijcai/BentH07}
\bibfield{author}{\bibinfo{person}{Russell Bent} {and} \bibinfo{person}{Pascal~Van Hentenryck}.} \bibinfo{year}{2007}\natexlab{}.
\newblock \showarticletitle{Waiting and Relocation Strategies in Online Stochastic Vehicle Routing}. In \bibinfo{booktitle}{\emph{{IJCAI}}}. \bibinfo{pages}{1816--1821}.
\newblock


\bibitem[\protect\citeauthoryear{Campos, Kieu, Guo, Huang, Zheng, Yang, and Jensen}{Campos et~al\mbox{.}}{2022}]%
        {davidpvldb}
\bibfield{author}{\bibinfo{person}{David Campos}, \bibinfo{person}{Tung Kieu}, \bibinfo{person}{Chenjuan Guo}, \bibinfo{person}{Feiteng Huang}, \bibinfo{person}{Kai Zheng}, \bibinfo{person}{Bin Yang}, {and} \bibinfo{person}{Christian~S. Jensen}.} \bibinfo{year}{2022}\natexlab{}.
\newblock \showarticletitle{Unsupervised Time Series Outlier Detection with Diversity-Driven Convolutional Ensembles}.
\newblock \bibinfo{journal}{\emph{Proc. {VLDB} Endow.}} \bibinfo{volume}{15}, \bibinfo{number}{3} (\bibinfo{year}{2022}), \bibinfo{pages}{611--623}.
\newblock


\bibitem[\protect\citeauthoryear{Dai, Yang, Guo, Jensen, and Hu}{Dai et~al\mbox{.}}{2017}]%
        {pvldb17pathcost}
\bibfield{author}{\bibinfo{person}{Jian Dai}, \bibinfo{person}{Bin Yang}, \bibinfo{person}{Chenjuan Guo}, \bibinfo{person}{Christian~S. Jensen}, {and} \bibinfo{person}{Jilin Hu}.} \bibinfo{year}{2017}\natexlab{}.
\newblock \showarticletitle{Path Cost Distribution Estimation Using Trajectory Data}.
\newblock \bibinfo{journal}{\emph{{PVLDB}}} \bibinfo{volume}{10}, \bibinfo{number}{3} (\bibinfo{year}{2017}).
\newblock


\bibitem[\protect\citeauthoryear{Dean}{Dean}{2010}]%
        {dean2010speeding}
\bibfield{author}{\bibinfo{person}{Brian~C Dean}.} \bibinfo{year}{2010}\natexlab{}.
\newblock \showarticletitle{Speeding up Stochastic Dynamic Programming with Zero-delay Convolution}.
\newblock \bibinfo{journal}{\emph{Algorithmic Operations Research}} \bibinfo{volume}{5}, \bibinfo{number}{2} (\bibinfo{year}{2010}), \bibinfo{pages}{96--104}.
\newblock


\bibitem[\protect\citeauthoryear{Goodrich and Tamassia}{Goodrich and Tamassia}{2006}]%
        {goodrich2006algorithm}
\bibfield{author}{\bibinfo{person}{Michael~T Goodrich} {and} \bibinfo{person}{Roberto Tamassia}.} \bibinfo{year}{2006}\natexlab{}.
\newblock \bibinfo{booktitle}{\emph{Algorithm Design: Foundation, Analysis and Internet Examples}}.
\newblock \bibinfo{publisher}{John Wiley \& Sons}.
\newblock


\bibitem[\protect\citeauthoryear{Guo, Jensen, and Yang}{Guo et~al\mbox{.}}{2014}]%
        {DBLP:journals/sigmod/GuoJ014}
\bibfield{author}{\bibinfo{person}{Chenjuan Guo}, \bibinfo{person}{Christian~S. Jensen}, {and} \bibinfo{person}{Bin Yang}.} \bibinfo{year}{2014}\natexlab{}.
\newblock \showarticletitle{Towards Total Traffic Awareness}.
\newblock \bibinfo{journal}{\emph{{SIGMOD} Record}} \bibinfo{volume}{43}, \bibinfo{number}{3} (\bibinfo{year}{2014}), \bibinfo{pages}{18--23}.
\newblock


\bibitem[\protect\citeauthoryear{Guo, Yang, Andersen, Jensen, and Torp}{Guo et~al\mbox{.}}{2015a}]%
        {DBLP:journals/geoinformatica/Guo0AJT15}
\bibfield{author}{\bibinfo{person}{Chenjuan Guo}, \bibinfo{person}{Bin Yang}, \bibinfo{person}{Ove Andersen}, \bibinfo{person}{Christian~S. Jensen}, {and} \bibinfo{person}{Kristian Torp}.} \bibinfo{year}{2015}\natexlab{a}.
\newblock \showarticletitle{EcoMark 2.0: Empowering Eco-routing with Vehicular Environmental Models and Actual Vehicle Fuel Consumption Data}.
\newblock \bibinfo{journal}{\emph{GeoInformatica}} \bibinfo{volume}{19}, \bibinfo{number}{3} (\bibinfo{year}{2015}), \bibinfo{pages}{567--599}.
\newblock


\bibitem[\protect\citeauthoryear{Guo, Yang, Andersen, Jensen, and Torp}{Guo et~al\mbox{.}}{2015b}]%
        {DBLP:conf/icde/Guo0AJT15}
\bibfield{author}{\bibinfo{person}{Chenjuan Guo}, \bibinfo{person}{Bin Yang}, \bibinfo{person}{Ove Andersen}, \bibinfo{person}{Christian~S. Jensen}, {and} \bibinfo{person}{Kristian Torp}.} \bibinfo{year}{2015}\natexlab{b}.
\newblock \showarticletitle{EcoSky: Reducing vehicular environmental impact through eco-routing}. In \bibinfo{booktitle}{\emph{{ICDE}}}. \bibinfo{pages}{1412--1415}.
\newblock


\bibitem[\protect\citeauthoryear{Guo, Yang, Hu, Jensen, and Chen}{Guo et~al\mbox{.}}{2020}]%
        {DBLP:journals/vldb/GuoYHJC20}
\bibfield{author}{\bibinfo{person}{Chenjuan Guo}, \bibinfo{person}{Bin Yang}, \bibinfo{person}{Jilin Hu}, \bibinfo{person}{Christian~S. Jensen}, {and} \bibinfo{person}{Lu Chen}.} \bibinfo{year}{2020}\natexlab{}.
\newblock \showarticletitle{Context-aware, preference-based vehicle routing}.
\newblock \bibinfo{journal}{\emph{{VLDB} J.}} \bibinfo{volume}{29}, \bibinfo{number}{5} (\bibinfo{year}{2020}), \bibinfo{pages}{1149--1170}.
\newblock


\bibitem[\protect\citeauthoryear{Hu, Yang, Guo, and Jensen}{Hu et~al\mbox{.}}{2018}]%
        {DBLP:journals/vldb/HuYGJ18}
\bibfield{author}{\bibinfo{person}{Jilin Hu}, \bibinfo{person}{Bin Yang}, \bibinfo{person}{Chenjuan Guo}, {and} \bibinfo{person}{Christian~S. Jensen}.} \bibinfo{year}{2018}\natexlab{}.
\newblock \showarticletitle{Risk-aware Path Selection with Time-varying, Uncertain Travel Costs: A Time Series Approach}.
\newblock \bibinfo{journal}{\emph{{VLDB} J.}} \bibinfo{volume}{27}, \bibinfo{number}{2} (\bibinfo{year}{2018}), \bibinfo{pages}{179--200}.
\newblock


\bibitem[\protect\citeauthoryear{Hua and Pei}{Hua and Pei}{2010}]%
        {DBLP:conf/edbt/HuaP10}
\bibfield{author}{\bibinfo{person}{Ming Hua} {and} \bibinfo{person}{Jian Pei}.} \bibinfo{year}{2010}\natexlab{}.
\newblock \showarticletitle{Probabilistic Path Queries in Road Networks: Traffic Uncertainty Aware Path Selection}. In \bibinfo{booktitle}{\emph{{EDBT}}}. \bibinfo{pages}{347--358}.
\newblock


\bibitem[\protect\citeauthoryear{Jin, Liu, Ding, and Wang}{Jin et~al\mbox{.}}{2011}]%
        {DBLP:journals/pvldb/JinLDW11}
\bibfield{author}{\bibinfo{person}{Ruoming Jin}, \bibinfo{person}{Lin Liu}, \bibinfo{person}{Bolin Ding}, {and} \bibinfo{person}{Haixun Wang}.} \bibinfo{year}{2011}\natexlab{}.
\newblock \showarticletitle{Distance-Constraint Reachability Computation in Uncertain Graphs}.
\newblock \bibinfo{journal}{\emph{{PVLDB}}} \bibinfo{volume}{4}, \bibinfo{number}{9} (\bibinfo{year}{2011}), \bibinfo{pages}{551--562}.
\newblock


\bibitem[\protect\citeauthoryear{Kieu, Yang, Guo, Cirstea, Zhao, Song, and Jensen}{Kieu et~al\mbox{.}}{2022a}]%
        {DBLP:conf/icde/KieuYGCZSJ22}
\bibfield{author}{\bibinfo{person}{Tung Kieu}, \bibinfo{person}{Bin Yang}, \bibinfo{person}{Chenjuan Guo}, \bibinfo{person}{Razvan{-}Gabriel Cirstea}, \bibinfo{person}{Yan Zhao}, \bibinfo{person}{Yale Song}, {and} \bibinfo{person}{Christian~S. Jensen}.} \bibinfo{year}{2022}\natexlab{a}.
\newblock \showarticletitle{Anomaly Detection in Time Series with Robust Variational Quasi-Recurrent Autoencoders}. In \bibinfo{booktitle}{\emph{{ICDE}}}. \bibinfo{pages}{1342--1354}.
\newblock


\bibitem[\protect\citeauthoryear{Kieu, Yang, Guo, Jensen, Zhao, Huang, and Zheng}{Kieu et~al\mbox{.}}{2022b}]%
        {DBLP:conf/icde/KieuYGJZHZ22}
\bibfield{author}{\bibinfo{person}{Tung Kieu}, \bibinfo{person}{Bin Yang}, \bibinfo{person}{Chenjuan Guo}, \bibinfo{person}{Christian~S. Jensen}, \bibinfo{person}{Yan Zhao}, \bibinfo{person}{Feiteng Huang}, {and} \bibinfo{person}{Kai Zheng}.} \bibinfo{year}{2022}\natexlab{b}.
\newblock \showarticletitle{Robust and Explainable Autoencoders for Unsupervised Time Series Outlier Detection}. In \bibinfo{booktitle}{\emph{{ICDE}}}. \bibinfo{pages}{3038--3050}.
\newblock


\bibitem[\protect\citeauthoryear{Letchner, Krumm, and Horvitz}{Letchner et~al\mbox{.}}{2006}]%
        {DBLP:conf/aaai/LetchnerKH06}
\bibfield{author}{\bibinfo{person}{Julia Letchner}, \bibinfo{person}{John Krumm}, {and} \bibinfo{person}{Eric Horvitz}.} \bibinfo{year}{2006}\natexlab{}.
\newblock \showarticletitle{Trip Router with Individualized Preferences {(TRIP):} Incorporating Personalization into Route Planning}. In \bibinfo{booktitle}{\emph{{AAAI}}}. \bibinfo{pages}{1795--1800}.
\newblock


\bibitem[\protect\citeauthoryear{Lin, Ghaddar, and Nathwani}{Lin et~al\mbox{.}}{2022}]%
        {lin2022}
\bibfield{author}{\bibinfo{person}{Bo Lin}, \bibinfo{person}{Bissan Ghaddar}, {and} \bibinfo{person}{Jatin Nathwani}.} \bibinfo{year}{2022}\natexlab{}.
\newblock \showarticletitle{Deep Reinforcement Learning for the Electric Vehicle Routing Problem With Time Windows}.
\newblock \bibinfo{journal}{\emph{{TITS}}} \bibinfo{volume}{23}, \bibinfo{number}{8} (\bibinfo{year}{2022}), \bibinfo{pages}{11528--11538}.
\newblock


\bibitem[\protect\citeauthoryear{Lu, Chen, Hao, and He}{Lu et~al\mbox{.}}{2020}]%
        {lu2020}
\bibfield{author}{\bibinfo{person}{Ji Lu}, \bibinfo{person}{Yuning Chen}, \bibinfo{person}{Jin-Kao Hao}, {and} \bibinfo{person}{Renjie He}.} \bibinfo{year}{2020}\natexlab{}.
\newblock \showarticletitle{The Time-dependent Electric Vehicle Routing Problem: Model and solution}.
\newblock \bibinfo{journal}{\emph{EXPERT SYST APPL}}  \bibinfo{volume}{161} (\bibinfo{year}{2020}), \bibinfo{pages}{113593}.
\newblock
\showISSN{0957-4174}


\bibitem[\protect\citeauthoryear{Ma, Yang, and Jensen}{Ma et~al\mbox{.}}{2014}]%
        {DBLP:conf/sigmod/Ma0J14}
\bibfield{author}{\bibinfo{person}{Yu Ma}, \bibinfo{person}{Bin Yang}, {and} \bibinfo{person}{Christian~S. Jensen}.} \bibinfo{year}{2014}\natexlab{}.
\newblock \showarticletitle{Enabling Time-Dependent Uncertain Eco-Weights For Road Networks}. In \bibinfo{booktitle}{\emph{GeoRich@SIGMOD}}. \bibinfo{pages}{1:1--1:6}.
\newblock


\bibitem[\protect\citeauthoryear{Newson and Krumm}{Newson and Krumm}{2009}]%
        {DBLP:conf/gis/NewsonK09}
\bibfield{author}{\bibinfo{person}{Paul Newson} {and} \bibinfo{person}{John Krumm}.} \bibinfo{year}{2009}\natexlab{}.
\newblock \showarticletitle{Hidden {M}arkov Map Matching through Noise and Sparseness}. In \bibinfo{booktitle}{\emph{SIGSPATIAL}}. \bibinfo{pages}{336--343}.
\newblock


\bibitem[\protect\citeauthoryear{Nie and Fan}{Nie and Fan}{2006}]%
        {nie2006arriving}
\bibfield{author}{\bibinfo{person}{Yu Nie} {and} \bibinfo{person}{Yueyue Fan}.} \bibinfo{year}{2006}\natexlab{}.
\newblock \showarticletitle{Arriving-on-time Problem: Discrete Algorithm that Ensures Convergence}.
\newblock \bibinfo{journal}{\emph{Transportation Research Record}} \bibinfo{volume}{1964}, \bibinfo{number}{1} (\bibinfo{year}{2006}), \bibinfo{pages}{193--200}.
\newblock


\bibitem[\protect\citeauthoryear{Niknami and Samaranayake}{Niknami and Samaranayake}{2016}]%
        {niknami2016tractable}
\bibfield{author}{\bibinfo{person}{Mehrdad Niknami} {and} \bibinfo{person}{Samitha Samaranayake}.} \bibinfo{year}{2016}\natexlab{}.
\newblock \showarticletitle{Tractable Pathfinding for the Stochastic On-time Arrival Problem}. In \bibinfo{booktitle}{\emph{International Symposium on Experimental Algorithms}}. Springer, \bibinfo{pages}{231--245}.
\newblock


\bibitem[\protect\citeauthoryear{Nikolova, Brand, and Karger}{Nikolova et~al\mbox{.}}{2006}]%
        {DBLP:conf/aips/NikolovaBK06}
\bibfield{author}{\bibinfo{person}{Evdokia Nikolova}, \bibinfo{person}{Matthew Brand}, {and} \bibinfo{person}{David~R. Karger}.} \bibinfo{year}{2006}\natexlab{}.
\newblock \showarticletitle{Optimal Route Planning under Uncertainty}. In \bibinfo{booktitle}{\emph{{ICAPS}}}. \bibinfo{pages}{131--141}.
\newblock


\bibitem[\protect\citeauthoryear{Pan and Liu}{Pan and Liu}{2023}]%
        {pan2023}
\bibfield{author}{\bibinfo{person}{Weixu Pan} {and} \bibinfo{person}{Shiqiang Liu}.} \bibinfo{year}{2023}\natexlab{}.
\newblock \showarticletitle{Deep Reinforcement Learning for the Dynamic and Uncertain Vehicle Routing Problem}.
\newblock \bibinfo{journal}{\emph{Applied Intelligence}}  \bibinfo{volume}{53} (\bibinfo{year}{2023}), \bibinfo{pages}{405–422}.
\newblock


\bibitem[\protect\citeauthoryear{Pedersen, Yang, and Jensen}{Pedersen et~al\mbox{.}}{2020a}]%
        {DBLP:journals/pvldb/PedersenYJ20}
\bibfield{author}{\bibinfo{person}{Simon~Aagaard Pedersen}, \bibinfo{person}{Bin Yang}, {and} \bibinfo{person}{Christian~S. Jensen}.} \bibinfo{year}{2020}\natexlab{a}.
\newblock \showarticletitle{Anytime Stochastic Routing with Hybrid Learning}.
\newblock \bibinfo{journal}{\emph{PVLDB}} \bibinfo{volume}{13}, \bibinfo{number}{9} (\bibinfo{year}{2020}), \bibinfo{pages}{1555--1567}.
\newblock


\bibitem[\protect\citeauthoryear{Pedersen, Yang, and Jensen}{Pedersen et~al\mbox{.}}{2020b}]%
        {DBLP:journals/vldb/PedersenYJ20}
\bibfield{author}{\bibinfo{person}{Simon~Aagaard Pedersen}, \bibinfo{person}{Bin Yang}, {and} \bibinfo{person}{Christian~S. Jensen}.} \bibinfo{year}{2020}\natexlab{b}.
\newblock \showarticletitle{Fast Stochastic Routing under Time-varying Uncertainty}.
\newblock \bibinfo{journal}{\emph{{VLDB} J.}} \bibinfo{volume}{29}, \bibinfo{number}{4} (\bibinfo{year}{2020}), \bibinfo{pages}{819--839}.
\newblock


\bibitem[\protect\citeauthoryear{Pedersen, Yang, and Jensen}{Pedersen et~al\mbox{.}}{2020c}]%
        {Pedersen0J20}
\bibfield{author}{\bibinfo{person}{Simon~Aagaard Pedersen}, \bibinfo{person}{Bin Yang}, {and} \bibinfo{person}{Christian~S. Jensen}.} \bibinfo{year}{2020}\natexlab{c}.
\newblock \showarticletitle{A Hybrid Learning Approach to Stochastic Routing}. In \bibinfo{booktitle}{\emph{{ICDE}}}. \bibinfo{pages}{1910--1913}.
\newblock


\bibitem[\protect\citeauthoryear{Pedersen, Yang, Jensen, and M{\o}ller}{Pedersen et~al\mbox{.}}{2023}]%
        {DBLP:journals/tsas/PedersenYJM23}
\bibfield{author}{\bibinfo{person}{Simon~Aagaard Pedersen}, \bibinfo{person}{Bin Yang}, \bibinfo{person}{Christian~S. Jensen}, {and} \bibinfo{person}{Jesper M{\o}ller}.} \bibinfo{year}{2023}\natexlab{}.
\newblock \showarticletitle{Stochastic Routing with Arrival Windows}.
\newblock \bibinfo{journal}{\emph{{ACM} Trans. Spatial Algorithms Syst.}} \bibinfo{volume}{9}, \bibinfo{number}{4} (\bibinfo{year}{2023}), \bibinfo{pages}{30:1--30:48}.
\newblock


\bibitem[\protect\citeauthoryear{Sabran, Samaranayake, and Bayen}{Sabran et~al\mbox{.}}{2014}]%
        {DBLP:conf/alenex/SabranSB14}
\bibfield{author}{\bibinfo{person}{Guillaume Sabran}, \bibinfo{person}{Samitha Samaranayake}, {and} \bibinfo{person}{Alexandre~M. Bayen}.} \bibinfo{year}{2014}\natexlab{}.
\newblock \showarticletitle{Precomputation Techniques for the Stochastic on Time Arrival Problem}. In \bibinfo{booktitle}{\emph{{ALENEX} 2014}}. \bibinfo{pages}{138--146}.
\newblock


\bibitem[\protect\citeauthoryear{Wang, Li, and Tang}{Wang et~al\mbox{.}}{2019}]%
        {WangLT19}
\bibfield{author}{\bibinfo{person}{Yong Wang}, \bibinfo{person}{Guoliang Li}, {and} \bibinfo{person}{Nan Tang}.} \bibinfo{year}{2019}\natexlab{}.
\newblock \showarticletitle{Querying Shortest Paths on Time Dependent Road Networks}.
\newblock \bibinfo{journal}{\emph{{PVLDB}}} \bibinfo{volume}{12}, \bibinfo{number}{11} (\bibinfo{year}{2019}), \bibinfo{pages}{1249--1261}.
\newblock


\bibitem[\protect\citeauthoryear{Wellman, Ford, and Larson}{Wellman et~al\mbox{.}}{1995}]%
        {DBLP:conf/uai/WellmanFL95}
\bibfield{author}{\bibinfo{person}{Michael~P. Wellman}, \bibinfo{person}{Matthew Ford}, {and} \bibinfo{person}{Kenneth Larson}.} \bibinfo{year}{1995}\natexlab{}.
\newblock \showarticletitle{Path Planning under Time-Dependent Uncertainty}. In \bibinfo{booktitle}{\emph{{UAI} 1995}}. \bibinfo{pages}{532--539}.
\newblock


\bibitem[\protect\citeauthoryear{Yang, Dai, Guo, Jensen, and Hu}{Yang et~al\mbox{.}}{2018}]%
        {DBLP:journals/vldb/YangDGJH18}
\bibfield{author}{\bibinfo{person}{Bin Yang}, \bibinfo{person}{Jian Dai}, \bibinfo{person}{Chenjuan Guo}, \bibinfo{person}{Christian~S. Jensen}, {and} \bibinfo{person}{Jilin Hu}.} \bibinfo{year}{2018}\natexlab{}.
\newblock \showarticletitle{{PACE:} A PAth-CEntric Paradigm for Stochastic Path Finding}.
\newblock \bibinfo{journal}{\emph{{VLDB} J.}} \bibinfo{volume}{27}, \bibinfo{number}{2} (\bibinfo{year}{2018}), \bibinfo{pages}{153--178}.
\newblock


\bibitem[\protect\citeauthoryear{Yang, Guo, and Jensen}{Yang et~al\mbox{.}}{2013}]%
        {DBLP:journals/pvldb/0002GJ13}
\bibfield{author}{\bibinfo{person}{Bin Yang}, \bibinfo{person}{Chenjuan Guo}, {and} \bibinfo{person}{Christian~S. Jensen}.} \bibinfo{year}{2013}\natexlab{}.
\newblock \showarticletitle{Travel Cost Inference from Sparse, Spatio-Temporally Correlated Time Series Using Markov Models}.
\newblock \bibinfo{journal}{\emph{{PVLDB}}} \bibinfo{volume}{6}, \bibinfo{number}{9} (\bibinfo{year}{2013}), \bibinfo{pages}{769--780}.
\newblock


\bibitem[\protect\citeauthoryear{Yang, Guo, Jensen, Kaul, and Shang}{Yang et~al\mbox{.}}{2014a}]%
        {DBLP:conf/icde/YangGJKS14}
\bibfield{author}{\bibinfo{person}{Bin Yang}, \bibinfo{person}{Chenjuan Guo}, \bibinfo{person}{Christian~S. Jensen}, \bibinfo{person}{Manohar Kaul}, {and} \bibinfo{person}{Shuo Shang}.} \bibinfo{year}{2014}\natexlab{a}.
\newblock \showarticletitle{Stochastic Skyline Route Planning under Time-varying Uncertainty}. In \bibinfo{booktitle}{\emph{{ICDE}}}. \bibinfo{pages}{136--147}.
\newblock


\bibitem[\protect\citeauthoryear{Yang, Kaul, and Jensen}{Yang et~al\mbox{.}}{2014b}]%
        {DBLP:journals/tkde/YangKJ14}
\bibfield{author}{\bibinfo{person}{Bin Yang}, \bibinfo{person}{Manohar Kaul}, {and} \bibinfo{person}{Christian~S. Jensen}.} \bibinfo{year}{2014}\natexlab{b}.
\newblock \showarticletitle{Using Incomplete Information for Complete Weight Annotation of Road Networks}.
\newblock \bibinfo{journal}{\emph{{IEEE} Trans. Knowl. Data Eng.}} \bibinfo{volume}{26}, \bibinfo{number}{5} (\bibinfo{year}{2014}), \bibinfo{pages}{1267--1279}.
\newblock


\bibitem[\protect\citeauthoryear{Yang, Guo, Hu, Tang, and Yang}{Yang et~al\mbox{.}}{2021}]%
        {DBLP:conf/ijcai/YangGHT021}
\bibfield{author}{\bibinfo{person}{Sean~Bin Yang}, \bibinfo{person}{Chenjuan Guo}, \bibinfo{person}{Jilin Hu}, \bibinfo{person}{Jian Tang}, {and} \bibinfo{person}{Bin Yang}.} \bibinfo{year}{2021}\natexlab{}.
\newblock \showarticletitle{Unsupervised Path Representation Learning with Curriculum Negative Sampling}. In \bibinfo{booktitle}{\emph{{IJCAI}}}. \bibinfo{pages}{3286--3292}.
\newblock


\bibitem[\protect\citeauthoryear{Yang, Guo, Hu, Yang, Tang, and Jensen}{Yang et~al\mbox{.}}{2022b}]%
        {DBLP:conf/icde/YangGHYTJ22}
\bibfield{author}{\bibinfo{person}{Sean~Bin Yang}, \bibinfo{person}{Chenjuan Guo}, \bibinfo{person}{Jilin Hu}, \bibinfo{person}{Bin Yang}, \bibinfo{person}{Jian Tang}, {and} \bibinfo{person}{Christian~S. Jensen}.} \bibinfo{year}{2022}\natexlab{b}.
\newblock \showarticletitle{Weakly-supervised Temporal Path Representation Learning with Contrastive Curriculum Learning}. In \bibinfo{booktitle}{\emph{{ICDE}}}. \bibinfo{pages}{2873--2885}.
\newblock


\bibitem[\protect\citeauthoryear{Yang, Guo, and Yang}{Yang et~al\mbox{.}}{2022a}]%
        {DBLP:journals/tkde/YangGY22}
\bibfield{author}{\bibinfo{person}{Sean~Bin Yang}, \bibinfo{person}{Chenjuan Guo}, {and} \bibinfo{person}{Bin Yang}.} \bibinfo{year}{2022}\natexlab{a}.
\newblock \showarticletitle{Context-Aware Path Ranking in Road Networks}.
\newblock \bibinfo{journal}{\emph{{IEEE} Trans. Knowl. Data Eng.}} \bibinfo{volume}{34}, \bibinfo{number}{7} (\bibinfo{year}{2022}), \bibinfo{pages}{3153--3168}.
\newblock


\bibitem[\protect\citeauthoryear{Yang, Hu, Guo, Yang, and Jensen}{Yang et~al\mbox{.}}{2023}]%
        {DBLP:conf/kdd/YangHGYJ23}
\bibfield{author}{\bibinfo{person}{Sean~Bin Yang}, \bibinfo{person}{Jilin Hu}, \bibinfo{person}{Chenjuan Guo}, \bibinfo{person}{Bin Yang}, {and} \bibinfo{person}{Christian~S. Jensen}.} \bibinfo{year}{2023}\natexlab{}.
\newblock \showarticletitle{LightPath: Lightweight and Scalable Path Representation Learning}. In \bibinfo{booktitle}{\emph{{KDD}}}. \bibinfo{publisher}{{ACM}}, \bibinfo{pages}{2999--3010}.
\newblock


\bibitem[\protect\citeauthoryear{Yang and Yang}{Yang and Yang}{2020}]%
        {DBLP:conf/icde/Yang020}
\bibfield{author}{\bibinfo{person}{Sean~Bin Yang} {and} \bibinfo{person}{Bin Yang}.} \bibinfo{year}{2020}\natexlab{}.
\newblock \showarticletitle{Learning to Rank Paths in Spatial Networks}. In \bibinfo{booktitle}{\emph{{ICDE}}}. \bibinfo{pages}{2006--2009}.
\newblock


\bibitem[\protect\citeauthoryear{Zheng and Ni}{Zheng and Ni}{2013}]%
        {DBLP:conf/aaai/ZhengN13}
\bibfield{author}{\bibinfo{person}{Jiangchuan Zheng} {and} \bibinfo{person}{Lionel~M. Ni}.} \bibinfo{year}{2013}\natexlab{}.
\newblock \showarticletitle{Time-Dependent Trajectory Regression on Road Networks via Multi-Task Learning}. In \bibinfo{booktitle}{\emph{{AAAI}}}. \bibinfo{pages}{1048--1055}.
\newblock


\end{thebibliography}

\end{document}